\documentclass[11pt]{article}
\usepackage{CSTheoryToolkitCMUStyle}
\usepackage{Custom}
\usepackage{cleveref}

\makeatletter
\renewcommand{\@seccntformat}[1]{%
  \ifcsname prefix@#1\endcsname
    \csname prefix@#1\endcsname
  \else
    \csname the#1\endcsname\quad
  \fi}
\makeatother

\begin{document}

\title{Single-Shot Universality in Quantum LDPC Codes via Code-Switching}

\author[1,4,$\ast,\dagger$]{Shi Jie Samuel Tan}
\author[1,2,$\ast,\ddagger$]{Yifan Hong}
\author[3,4,$\ast,\S$]{Ting-Chun Lin}
\author[1,$\P$]{\\Michael J. Gullans}
\author[4,$\parallel$]{Min-Hsiu Hsieh}
\affil[1]{\small Joint Center for Quantum Information and Computer Science,
University of Maryland and NIST, College Park, MD, USA}
\affil[2]{\small Joint Quantum Institute, University of Maryland and NIST, College Park, MD, USA}
\affil[3]{\small University of California San Diego,
La Jolla, CA, USA}
\affil[4]{\small Hon Hai Research Institute,
Taipei, Taiwan}

\date{\today}
\maketitle

\newcommand{\col}{\mathrm{col}}
\newcommand{\row}{\mathrm{row}}

\newcommand{\SAM}[1]{}
\newcommand{\yifan}[1]{}
\newcommand{\david}[1]{}
\renewcommand{\thefootnote}{\fnsymbol{footnote}}
\setcounter{footnote}{1}
\footnotetext{{These authors contributed equally to this work.}}
\setcounter{footnote}{2}
\footnotetext{\href{mailto:stan97@umd.edu}{stan97@umd.edu}}
\setcounter{footnote}{3}
\footnotetext{\href{mailto:yhong137@umd.edu}{yhong137@umd.edu}}
\setcounter{footnote}{4}
\footnotetext{\href{mailto:til022@ucsd.edu}{til022@ucsd.edu}}
\setcounter{footnote}{5}
\footnotetext{\href{mailto:mgullans@umd.edu}{mgullans@umd.edu}}
\setcounter{footnote}{6}
\footnotetext{\href{mailto:min-hsiu.hsieh@foxconn.com}{min-hsiu.hsieh@foxconn.com}}
\setcounter{footnote}{1}
\renewcommand{\thefootnote}{\arabic{footnote}}
\begin{abstract}
  Code-switching is a powerful technique in quantum error correction that allows one to leverage the complementary strengths of different codes to achieve fault-tolerant universal quantum computation. However, existing code-switching protocols that encapsulate recent generalized lattice surgery approaches often either require many rounds of measurements to ensure fault-tolerance or suffer from low code rates. We present a single-shot, universal protocol that uses code-switching between high-rate quantum codes to perform fault-tolerant quantum computation. To our best knowledge, our work contains the first universal fault-tolerant quantum computation protocol that achieves what we term \emph{single-shot universality} on high-rate codes that is characterized by (i) single-shot error correction, (ii) single-shot state preparation, as well as (iii) universal logical gates and logical measurements with constant depth circuits. We achieve this feat with single-shot code-switching between constant-rate 2D hypergraph product (HGP) codes and high-rate 3D HGP codes that can be viewed as a generalization of Bomb\'{i}n's dimensional jump for color codes and Hillmann et al.'s single-shot lattice surgery for higher-dimensional topological codes. In addition, we prove the fault-tolerance of our code-switching protocol under both the adversarial and local-stochastic noise models. We introduce a vastly simpler recipe to construct high-rate 3D HGP codes with transversal CCZ gates that grants immense flexibility in the choice of expander graphs and local codes, allowing us to expand the search space for codes with good parameters and interesting logical gates. Our work opens an alternative path towards universal fault-tolerant quantum computation with low space-time overhead by circumventing the need for magic state distillation.
\end{abstract}
\newpage
\tableofcontents

\section{Introduction}
\label{sec:intro}
Quantum error correction—and, more broadly, fault-tolerant quantum computation—seeks to deliver reliable large-scale quantum algorithms by expending extra resources. For a long time, the surface code~\cite{Kitaev_2003} and other topological codes~\cite{bombin2006topological} have been the leading candidates for fault-tolerant quantum computation due to their high thresholds for different noise models and compatibility with 2D qubit architectures~\cite{dennis2002topological, fowler2012surface, litinski2019game, tan2024resilience}. However, these codes suffer from high space-time overheads because they encode only a constant number of logical qubits per code block~\cite{BPT_2010}.
Following the introduction of topological codes, there has been sustained effort to shrink these costs, culminating in the emergence of high-rate quantum low-density parity-check (QLDPC) codes~\cite{tillich2014quantum, breuckmann2021balanced,panteleev2022asymptotically,leverrier2022quantum,dinur2023good, bravyi2024high}. In recent years, asymptotically good QLDPC codes that have constant rate and linear distance have been proven to exist~\cite{panteleev2022asymptotically, leverrier2022quantum,dinur2023good}. These codes forgo geometric locality to encode logical qubits more efficiently than topological codes. Quantum Tanner codes, one of these asymptotically good QLDPC codes, have also been shown to be amenable to single-shot error correction, making them extremely efficient both in terms of space and time overheads~\cite{gu2024single,dinur2023good,leverrier2023decoding, leverrier2025efficient}. Gottesman showed that constant-rate qLDPC families can support fault-tolerant computation with asymptotically constant spatial overhead, though with linear time overhead relative to the unencoded circuit~\cite{gottesman2013fault}. Since then, there has been a flurry of work to reduce this time overhead~\cite{fawzi2020constant,yamasaki2024time,tamiya2024polylog,nguyen2025quantum}. Most if not all of these schemes rely on magic state distillation (MSD)~\cite{bravyi2012magic} to introduce non-Clifford gates and promote the native Clifford gates of the codes to a universal gate set to achieve universal fault-tolerant quantum computation. Because our ultimate goal is to achieve fast and efficient universal fault-tolerant quantum computation, it is natural to ask whether we can bypass expensive MSDs altogether by directly implementing a universal set of logical gates transversally. Nonetheless, the Eastin-Knill theorem~\cite{eastin2009restrictions} forbids the existence of a quantum code that can implement a universal set of logical gates transversally.

Code-switching which refers to the process of switching between different quantum error-correcting codes, has been proposed as an alternative method to sidestep the Eastin-Knill theorem~\cite{eastin2009restrictions} and expand the set of transversal logical gates that can be implemented fault-tolerantly. It was first proposed in the form of code deformation for topological codes where a given surface code is iteratively deformed to obtain CNOT gates~\cite{bombin2009quantum}. The idea was later generalized to topological subsystem color codes which can now implement Clifford logical gates via code deformation~\cite{bombin2011clifford}.
Other works have investigated code-switching between small quantum stabilzier codes to toggle between memory and computation codes~\cite{hill2011fault}. Hill, et al. looked at using a five-qubit code for quantum memory and switching to a seven-qubit Steane code for computation. This allows for some physical overhead savings as the five-qubit code is more space-efficient than the Steane code. Paetznick and Reichardt~\cite{paetznick2013universal} proposed a scheme to achieve universal fault-tolerant quantum computation by using the triorthogonal codes introduced by Bravyi and Haah~\cite{bravyi2012magic}. Their scheme uses transversal CCZ and H gates to achieve universality. However, the triorthogonal codes do not have good distance scaling and hence are not suitable for large-scale quantum computation. Beyond code-switching, concatenation has also been explored as a method for achieving universal computation at the expense of the distance of the code~\cite{jochym2014using}. Yoder et al. uses concatenation and pieceable fault-tolerance to achieve universal fault-tolerant quantum computation. However, the scalability and efficiency of their scheme is unclear and their work primarily focused on small codes with small distances~\cite{yoder2016universal}. 

The first scalable protocol proposed for code-switching with the aim of achieving universal fault-tolerant quantum computation was by Bomb\'{i}n~\cite{bombin2015gauge,bombin2016dimensional,kubica2015universal}. His proposal was to toggle between 2D and 3D gauge color codes via the method of gauge-fixing, exploiting the fact that a quantum code can access transversal gates from different levels of the Clifford hierarchy if they are defined in different spatial dimensions~\cite{bravyi2013classification}. To be explicit, the 2D color code admits transversal Cliffords gates like the Hadamard gate and the CNOT gate while the 3D color code supports a transversal non-Clifford T gate. Bomb\'{i}n's proposal uses \emph{dimensional jumping} to promise constant-time logical gates using only local operations and leverages single-shot error correction in subsystem color codes. Anderson et al.~\cite{anderson2014fault} also stated a scheme involving switching between the Steane and Quantum Reed-Muller code shares a lot of similarities with Bomb\'{i}n's proposal.
Beyond the color code, Brown utilized the just-in-time gauge fixing technique introduced by Bomb\'{i}n~\cite{bombin2018transversal} to propose a code deformation scheme for the 2D surface code to implement the non-Clifford CCZ gate~\cite{brown2020fault}. This is interesting because the surface code is known to have better error correcting performance than the color code due to its lower stabilizer weight and better decoding algorithms~\cite{dennis2002topological, delfosse2014decoding, delfosse2021almost}. Kubica and Vasmer were also able to adapt Bomb\'{i}n's dimensional jump scheme to the case of the 3D subsystem toric code, which is often viewed as the canonical instance of 3D topological codes, to achieve single-shot error correction gadgets and logical gates~\cite{kubica2022single}. Coupled with the single-shot state preparation ideas in Ref.~\cite{bravyi2020quantum}, it shows that it is possible to achieve universal fault-tolerant quantum computation with single-shot error correction and state preparation using only local operations.

While these protocols are very promising and exciting due to the existence of a growing family of codes that work well in practice, they rely on topological codes that are known to have poor asymptotic spatial overhead scaling due to the fact that they only encode a constant number of logical qubits~\cite{BPT_2010}. Recently, Hong proposed a dimensional collapse protocol~\cite{hong2024single} that is a generalization of a part of Bomb\'{i}n's dimensional jump to the hypergraph product (HGP) codes that are known to have a constant rate~\cite{tillich2014quantum}. The protocol allows a \emph{one-way} switch from a specific 3D HGP code to a 2D HGP code. This is a particularly exciting result because it opens up the possibility for code-switching between high-rate HGP codes.
However, a recent result by Fu, et al.~\cite{fu2025no} shows that there are fundamental limitations to the set of transversal logical gates that can be implemented on HGP codes. To be specific, they show that 3D HGP codes cannot afford transversal single-qubit non-Clifford gates.
Despite these limitations, several recent constructions of high-performance quantum codes with transversal CCZ gates have been proposed~\cite{zhu2023non,wills2024constant,nguyen2024good,golowich2024asymptotically,golowich2025quantum,lin2024transversal,breuckmann2024cups,zhu2025topological,zhu2025transversal}. Some of these constructions are actually 3D HGP codes that are carefully constructed using classical Sipser-Spielman codes~\cite{sipser2002expander}. One can think of these as the high-rate generalizations of past results that showed how one can climb the Clifford hierarchy by increasing the spatial dimension of topological codes~\cite{kubica2015unfolding, moussa2016transversal, vasmer2019three,brown2020fault}. 

While these code constructions are very promising for shaving space overhead, we also require every operation in the protocol to only require a constant-depth circuit implementation for fast and efficient fault-tolerant quantum computation. We introduce the term \emph{single-shot universality} to refer to a universal fault-tolerant quantum computation protocol that satisfies the following properties: (i) single-shot error correction, (ii) single-shot state preparation, as well as (iii) universal logical gates and logical measurements with constant-depth circuit implementations. Single-shot universality is a desirable property for efficient quantum computation because it allows us to perform every necessary operation for quantum computation in a constant-depth circuit, thereby reducing the time overhead and the accumulation of errors. This is especially of relevance for quantum hardware devices with fast qubit decoherence times as well as slow measurement times. For the case of the recently constructed asymptotically good qLDPC codes, 
Given the recent code construction developments discussed above and the need for single-shot universality for efficient quantum computation, it is natural to ask:
\begin{center}
\emph{Can we generalize Bomb\'{i}n's dimensional jump scheme for the color codes to a single-shot code-switching scheme for high-rate codes with transversal CCZ gates to achieve low spatial overhead universal fault tolerant quantum computation with single-shot universality?}
\end{center}

\subsection{Related Works}
\label{subsec:related_works}
As discussed above, Hong proposed a dimensional collapse protocol~\cite{hong2024single} that allows a \emph{one-way} switch from a 3D HGP code to a 2D HGP code. This is particularly useful for single-shot state preparation of logical $\ket{0}$ and $\ket{+}$ states in the 2D HGP code since it harnesses the soundness property in the higher-dimensional HGP code. However, it requires certain 1D-like structure in the third dimension, and it is not immediate how to apply the protocol to our setting.

Lattice surgery, a method that can also be understood as a form of code-switching, has also been explored as a method to expand the set of logical gates that can be implemented fault-tolerantly. It was first proposed for the surface code to implement CNOT gates~\cite{horsman2012surface}. Recently, generalized lattice surgery has been proposed for arbitrary high-rate QLDPC codes as a means to perform Pauli-based computation~\cite{bravyi2016trading}. For these high-rate codes with many logical qubits, it is often challenging to address individual logical qubits through targeted Pauli measurements to perform logical gates. These generalized lattice surgery protocols merge the code block together with an ancilla code block to form a larger code block, allowing one to perform joint logical measurements fault-tolerantly on multiple logical qubits before unmerging the code blocks~\cite{cowtan2024css,cowtan2024ssip,cross2024improved,swaroop2024universal,ide2025fault,cowtan2025parallel,he2025extractors}. While the state-of-the-art generalized lattice surgery protocols are extremely powerful and general, they ultimately still suffer from a time overhead that scales with the distance of the code due to the need for multiple rounds of measurements to ensure fault-tolerance. A recent work by Hillmann, et al.~\cite{hillmann2024single} showed that it is possible to perform single-shot lattice surgery that only requires a single round of measurement to reliably merge and split two higher-dimensional topological code blocks. While they were able to reduce the time overhead of the surgery procedure to a constant, their protocol still relies on topological codes that can only encode a constant number of logical qubits. 

We note that Xu et al. addressed the problem of single-shot targeted joint logical Pauli measurement for constant rate HGP codes in Ref.~\cite{xu2025fast} where they constructed special homomorphic CNOTs and used a Steane measurement scheme to perform the joint logical measurements. Their scheme allows for efficient logical teleportation between between HGP codes that differ by punctures and augmentations. However, their scheme does not address the problem of switching between HGP codes of different dimensions to access different transversal logical gates.

Heu{\ss}en and Hilder recently proposed a code-switching protocol that leverages one-way transversal CNOT gates to enable efficient fault-tolerant code-switching between codes~\cite{heussen2025efficient} via logical teleportation. Their protocol is interesting because it uses teleportation to switch between different codes, thereby avoiding the need for direct code deformation. However, the one-way CNOT gates they use are designed for color codes and it is unclear how to generalize their protocol to high-rate codes. 

\subsection{Our Results}
\label{subsec:results}
In this work, we answer the question posed in the introduction in the affirmative by providing a framework for single-shot code-switching for high-rate HGP codes that are known to support single-shot error correction~\cite{fawzi2018efficient, campbell2019theory,quintavalle2021single}, efficient syndrome extraction circuits~\cite{manes2025distance, tan2024effective,berthusen2025adaptive}, fast Clifford gates~\cite{xu2025fast}, and transversal CCZ gates~\cite{golowich2025quantum,lin2024transversal,zhu2025topological,zhu2025transversal}. By carefully designing new gadgets and incorporating existing state-of-the-art gadgets, our proposed scheme achieves \emph{single-shot universality}. In other words, In other words, our protocol possesses single-shot code-switching, single-shot error correction, single-shot state preparation, as well as universal logical gates and measurements that can be done in a constant-depth circuit (without $d$ rounds of measurements) fault-tolerantly. To our best knowledge, this is the first universal FTQC protocol on high-rate codes that possesses this property.  We summarize our main result in the informal theorem statement below:

\begin{theorem}[Informal Statement of Theorem~\ref{thm:universal_FTQC}]
  \label{theorem:informal_universal_FTQC}
  There exists a fault-tolerant quantum computation protocol that implements a universal gate set via single-shot code-switching between high-rate HGP codes using only gadgets that can be performed in a single-shot manner with a constant-depth circuit. The protocol can tolerate a number of adversarial errors that scales linearly with the distance of the HGP codes. In addition, the protocol is fault-tolerant under the local-stochastic noise model.
\end{theorem}

From a high level, our scheme takes advantage of the Clifford gates and single-shot error correction in constant-rate 2D HGP codes as well as the transversal CCZ gate and single-shot state preparation in certain 3D HGP codes to achieve single-shot universality. Compared to the usual efficient HGP constructions that utilizes lossless expanders, we construct HGP codes using classical Tanner codes. The slight difference in our construction means that we have to prove certain error correcting and single-shot properties for this different class of HGP codes.
While it is possible to construct constant-rate 3D HGP codes, we make the design choice to trade the rate of the 3D HGP code in the new third dimension for greater flexibility in the first two dimensions, allowing us to use any constant-rate 2D HGP code for code-switching into high-rate 3D HGP codes with transversal CCZ gates. Even though these high-rate 3D HGP codes essentially encode the same number of logical qubits as the original constant-rate 2D HGP codes, this is sufficient for all practical purposes because we always have to code-switch back to the 2D HGP code for error correction and Clifford gates so there is typically no need to have the capability of encoding more logical qubits in the 3D case. This is discussed in greater detail in Section~\ref{paragraph:transversal_logical_gates}. The key ingredient that enables our scheme is a new single-shot code-switching protocol that allows us to fault-tolerantly switch between 2D and 3D HGP codes using only a constant-depth circuit. By switching back and forth between the two codes, we can implement the Hadamard gate in the 2D HGP code and the CCZ gate in the 3D HGP code to achieve universal quantum computation. We provide a schematic overview of our scheme in Figure~\ref{fig:dimensional_expansion_contraction_overview}. We now proceed to provide a technical overview of the results and techniques used in this work to prove Theorem~\ref{theorem:informal_universal_FTQC}.
\begin{figure}[htbp]
\centering
\usetikzlibrary{shapes.geometric, arrows.meta, positioning, calc, patterns, decorations.pathreplacing, shadows}

\begin{tikzpicture}[
    code2d/.style={
        rectangle, 
        draw=blue!60, 
        fill=blue!20, 
        thick, 
        minimum width=4cm, 
        minimum height=2cm,
        rounded corners=5pt
    },
    code3d/.style={
        rectangle, 
        draw=red!60, 
        fill=red!20, 
        thick, 
        minimum width=4cm, 
        minimum height=2cm,
        rounded corners=5pt
    },
    prep3d/.style={
        rectangle, 
        draw=green!60, 
        fill=green!20, 
        thick, 
        minimum width=4cm, 
        minimum height=2cm,
        rounded corners=5pt
    },
    operation/.style={
        rectangle, 
        draw=orange!70, 
        fill=orange!10, 
        thick, 
        minimum width=3cm, 
        minimum height=0.8cm,
        rounded corners=3pt,
        font=\small
    },
    transition/.style={
        <->, 
        thick, 
        >=stealth,
        color=purple!70
    },
    singleshotbox/.style={
        draw=gray!50,
        dashed,
        thick,
        rounded corners=3pt
    },
    label/.style={
        font=\footnotesize,
        text=black!70
    }
]

\node[font=\Large\bfseries] at (0, 8) {Single-Shot code-switching Scheme};
\node[font=\normalsize] at (0, 7.3) {for Fault-Tolerant Universal Quantum Computation};

\node[code2d] (main2d) at (-5, 4) {
    \begin{tabular}{c}
        \textbf{2D HGP Code} $\mathcal{Q}$ \\
        $[\![n, k, d_X, d_Z]\!]$ \\
        \small (Clifford Operations) \\
        • \small Fold-transversal Gates \\
        • \small Pauli-based Computation \\
        • \small Single-shot EC
    \end{tabular}
};


\node[code3d] (middle3d) at (5, 4) {
    \begin{tabular}{c}
        \textbf{3D HGP Code} $\mathcal{Q}_G$ \\
        $[\![n_G, k_G, d_{X,G}, d_{Z,G}]\!]$ \\
        \small (Non-Clifford Operations) \\
        • \small Transversal CCZ Gates
    \end{tabular}
};


\draw[transition, line width=2pt] (main2d.east) -- (middle3d.west);
\node[above, font=\small\bfseries] at (0, 4.5) {Dimensional};
\node[above, font=\small\bfseries] at (0, 4.1) {Expansion/Contraction};



\node[singleshotbox, fit=(main2d.north west) (main2d.south east), inner sep=8pt] {};
\node[label, above=0.5cm of main2d] {Addressable Logical Cliffords};

\node[singleshotbox, fit=(middle3d.north west) (middle3d.south east), inner sep=8pt] {};
\node[label, above=0.5cm of middle3d] {Transversal Inter-block CCZs};


\node[draw=black!30, fill=yellow!10, rounded corners, align=left] at (-5.5, -1) {
    \textbf{Key Features:}\\
    • Constant spatial overhead\\
    \quad (except during non-Clifford \\
    \quad and state preparation)\\
    • Single-shot Universality
};

\node[draw=black!30, fill=cyan!10, rounded corners, align=left] at (0, -1) {
    \textbf{Code Properties:}\\
    • Sipser-Spielman HGP\\
    • Single-shot EC
};

\node[draw=black!30, fill=pink!10, rounded corners, align=left] at (5.5, -1) {
    \textbf{Computational Universality:}\\
    • Clifford: 2D HGP code\\
    • Non-Clifford: 3D HGP code\\
};

\draw[<->, thick, gray] (-7.35, 2.5) -- (-7.35, 5.5);
\node[rotate=90, gray] at (-7.75, 4) {2D};

\draw[<->, thick, gray] (7.25, 3) -- (7.25, 5);
\node[rotate=90, gray] at (7.65, 4) {3D};


\node[font=\footnotesize, text=purple!70] at (0, 3.4) {($\leftarrow$ Teleportation $\rightarrow$)};


\end{tikzpicture}
\caption{Overview of the single-shot universal protocol for fault-tolerant quantum computation using 2D and 3D HGP codes. The scheme leverages single-shot error correction, state preparation, and code-switching to achieve addressable Clifford and non-Clifford operations with low spatial overhead and time complexity.\label{fig:dimensional_expansion_contraction_overview}}
\end{figure}

\subsubsection{Dimensional Mechanics through Single-Shot code-switching}
\label{subsec:dimensional_mechanics}
Single-shot code-switching between HGP codes can be challenging because of circuit depth concerns. From a high level, a 2D HGP code is amenable to (fold)-transversal Hadamard gates because a constant-depth circuit is sufficient to swap the X and Z logical operators that are supported on 1D objects. On the other hand, a 3D HGP code is amenable to transversal inter-block CCZ gates because the 2D X logical operators always intersect at an odd number of points which grants the necessary X stabilizer invariance that the CCZ logical action requires. When we only have 1D X logical operators, there is no guarantee on the parity of the intersection points when we shift the logical operators with stabilizers. Therefore, the goal for our single-shot code-switching protocol is to expand a 2D HGP code into a 3D HGP code in a way that grows the X logical operators from 1D to 2D objects while ensuring that the circuit is fault-tolerant and constant-depth. However, when the second dimension of the X logical operators scales with the distance of the code, one might expect that the circuit depth must also scale with the distance of the code. However, we show that it is possible to achieve this goal.




We introduce a new single-shot code-switching protocol that allows us to fault-tolerantly switch between a $D$-dimensional HGP code and a $(D+1)$-dimensional HGP code built from the former and a generalized repetition code constructed using the Sipser-Spielman classical expander codes~\cite{sipser2002expander}. This can be iterpreted as a single-shot lattice surgery protocol that is performed via logical teleportation instead of code deformation. The informal theorem statement below summarizes our second protocol:

\begin{theorem}[Informal Statement of Theorem~\ref{thm:SSCS}]
  \label{thm:informal_SSCS}
  Let $\calQ$ be an $D$-dimensional HGP code constructed from classical expander codes for $D \geq 2$. Let $\calQ_G$ be an $(D+1)$-dimensional HGP codes constructed from $\calQ$ and a classical generalized repetition code.
  There exists a SSCS protocol that fault-tolerantly switches between $\calQ$ and $\calQ_G$ using only a constant-depth logical teleportation circuit that can tolerate a number of adversarial errors that scales linearly with the distance of the codes. In addition, the protocol exhibits a threshold against local-stochastic noise.
\end{theorem}

The high-level idea of our protocol is to utilize the homomorphic CNOT framework developed by Huang, et al.~\cite{huang2023homomorphic} to develop a \emph{one-way} logical CNOT gate between HGP codes. This addresses the open question raised in Ref.~\cite{heussen2025efficient} of how to generalize the one-way CNOT gate for color codes to high-rate codes for the purpose of code-switching via logical teleportation. We prove that there exists a sparse chain map between the chain complexes of $\calQ$ and $\calQ_G$ that allows us to implement the homomorphic CNOT between them with a depth-1 circuit. We utilize the grid Pauli product measurement (GPPM) framework developed by Xu, et al.~\cite{xu2025fast} to perform the necessary single-shot joint logical measurements along with the homomorphic CNOTs for the entanglement to complete the logical teleportation between $\calQ$ and $\calQ_G$ for the code-switching.

\subsubsection{Single-Shot Universality}
\label{subsec:single_shot_universality}
In this section, we discuss how every single gadget can be made single-shot or constant-depth to achieve single-shot universality. We discuss each gadget in a separate paragraph below.

\paragraph{Single-Shot State Preparation}
\label{paragraph:single_shot_state_prep}
To prepare logical $\ket{\overline{0}}$ and $\ket{\overline{+}}$ states in the 2D HGP code $\calQ$ in a single shot, we use the protocol of Bergamaschi and Liu \cite{bergamaschi2024fault}, which requires a spatial overhead of $\Theta(\sqrt{n})$, i.e. $\Theta(n^{3/2})$ additional ancilla qubits. Since 2D HGP codes generically lack any metacheck structure, they do not typically possess single-shot state preparation without additional overhead.

\paragraph{Single-Shot Error Correction.}
\label{paragraph:single_shot_error_corr}
It is known that 2D HGP codes built from classical expander codes can perform single-shot error correction~\cite{fawzi2018efficient, campbell2019theory,quintavalle2021single}. To be specific, 2D HGP codes are known to exhibit single-shot error correction under both the adversarial and local-stochastic noise models~\cite{fawzi2018efficient, campbell2019theory}. Less is known about single-shot error correction in higher-dimensional HGP codes. Campbell and Quintavalle, et al.~\cite{quintavalle2021single} showed that 3D HGP codes built from classical expander codes can perform single-shot error correction under the adversarial noise model in one Pauli basis. Our new result discussed above in the single-shot state preparation section implies that 3D HGP codes can perform single-shot error correction under the local-stochastic noise model in one Pauli basis. Even though we can only perform single-shot error correction in a single Pauli basis under both noise models, we can always delay the error correction in the other Pauli basis until we switch back to a 2D HGP code where we can perform single-shot error correction in both Pauli bases. This is sufficient for the purpose of our fault-tolerance analysis because all of our gadgets have a constant-depth implementation and hence we can bound the accumulation of errors. Most importantly, for most applications, our code only stays in 3D momentarily to perform the transversal CCZ gate before switching back to 2D. Therefore, we do not expect a significant accumulation of errors in the other Pauli basis during the short time spent in 3D. 

\paragraph{Transversal Logical Clifford and Non-Clifford Gates}
\label{paragraph:transversal_logical_gates}
There are many options for implementing logical Clifford gates in 2D HGP codes. We can use the fold-transversal framework developed by Breuckmann and Burton~\cite{breuckmann2024fold} to implement Clifford gates in a constant-depth circuit~\cite{quintavalle2023partitioning}. Automorphisms are also a powerful tool for implementing logical Clifford gates in HGP codes~\cite{berthusen2025automorphism}. Pauli-based computation~\cite{bravyi2016trading} is another powerful framework for implementing logical Clifford gates in HGP codes. We can use the GPPM framework developed by Xu, et al.~\cite{xu2025fast} to perform targeted logical Pauli measurements by entangling the HGP code block with a punctured or/and augmented HGP code block using homomorphic CNOTs. This allows us to perform addressable logical Clifford gates in a constant-depth circuit when coupled with fold-transversal gates and logical automorphisms.

Before we discuss how the transversal CCZ gates can be implemented on our HGP codes, we first address a possible question that readers might have.
One might ask why do we need to code switch to $\calQ_G$ when we already have an adapter 3D HGP code block $\calQ_S$ that has the ability to perform dimensional jump when coupled with our single-shot lattice surgery protocol. The reason is that the 3D HGP code $\calQ_S$ may not have a transversal CCZ gate because the constraints for allowing CCZ gates are typically extremely restrictive and difficult to satisfy. Using results regarding sheaf codes from Ref.~\cite{lin2024transversal} and the HGP constructions from Ref.~\cite{zhu2025topological,zhu2025transversal}, we prove that $\calQ_G$ is amenable to a transversal CCZ gate even when $\calQ$ is freely constructed from classical Sipser-Spielman codes built from arbitrary expanders and local codes. The only constraint we have to impose is that the three blocks of 3D HGP code must have matching expander graphs in the same physical dimension and the corresponding local codes in the same physical dimension must satisfy the \emph{multiplication property}. While choosing an asymptotically good classical code for the third dimension can help with achieving constant rate for our 3D HGP codes, we intentionally choose a generalized repetition code for the third dimension to allow us to satisfy the multiplication property in a simple, systematic way even though it may not be the most space-efficient choice. 
Because we are using a repetition code as one of the local codes, this multiplication property is automatically satisfied when the other two local codes in the same physical dimension of the two other 3D HGP codes are chosen to be dual to each other. This vastly simplifies the construction of high-rate 3D HGP codes with transversal CCZ gates and allows us to search and consider a much larger family of classical codes as constituents for our HGP codes than before, thus expanding the design space for finite-size codes with good parameters, logical gates, and code automorphisms~\cite{berthusen2025automorphism}. This is because we can effectively pick any two classical codes for the construction of the 2D HGP code $\calQ$ that is used to construct our first 3D HGP code block and then use it to constrain the construction of the other two 3D HGP code blocks. 
We summarize our result regarding transversal CCZ gates in the informal theorem statement below:
\begin{theorem}[Informal Statement of Theorem~\ref{thm:satisfaction_of_local_multiplication_property}]
  \label{thm:informal_dual_multiplication_property}
  Let $\calQ_1, \calQ_2, \calQ_3$ be three 3D HGP code constructed with a generalized repetition code in one of its dimensions and two classical codes in the other two dimensions. Let the following three conditions hold:
  \begin{itemize}
    \item The three 3D HGP codes have the same expander graphs in the same physical dimensions.
    \item The generalized repetition codes are all oriented in different physical dimensions. 
    \item For each physical dimension, the two classical codes used in the other two 3D HGP codes are dual to each other in a local way.
  \end{itemize}
  Then the transversal CCZ gate on the three code blocks implements a logical CCZ gate.
\end{theorem}

\paragraph{Transversal Measurement}
\label{paragraph:transversal_measurement}
Transversal measurement of logical Pauli operators in HGP codes can be done in a single-shot manner without the need for multiple rounds of measurements. By using the GPPM framework developed by Xu, et al.~\cite{xu2025fast}, we can perform targeted logical Pauli measurements by entangling the HGP code block with a punctured or/and augmented HGP code block using homomorphic CNOTs before performing transversal measurements on the ancilla code block using the Steane measurement framework. The measurement outcomes can then be processed classically to infer the desired logical measurement outcome fault-tolerantly. If we are interested in measuring all the logical qubits in some Pauli basis, we can simply perform transversal measurements on all the physical qubits in the code block without the need for any ancilla code blocks just like the case for any CSS code.

\subsection{Discussion and Outlook}
\label{subsec:discussion_outlook}
In this work, we provided a framework for single-shot code-switching between high-rate HGP codes that allows us to achieve single-shot universality for fault-tolerant quantum computation. Our protocol leverages existing and new gadgets that can be performed in a single-shot manner with constant-depth circuits to achieve single-shot universality. To our best knowledge, this is the first universal fault-tolerant quantum computation protocol for high-rate codes that possesses this property. 

\begin{remark}
    We recently became aware of independent and concurrent work by Li, et al.~\cite{li2025universal} that also considers universal fault-tolerant quantum computation that uses ideas related to single-shot code-switching. Our code-switching approach is different from theirs and we believe that both approaches are interesting and complementary to each other. In addition, we provide an explicit construction of high-rate 3D HGP codes with transversal CCZ gates that can be used for universal fault-tolerant quantum computation.

    Independently, Golowich, et al.~\cite{golowich2025constant} and Qian, et al.~\cite{xu2025batched} also have upcoming works that consider code-switching between high-rate codes for addressable gates or/and high-rate logical operations. Their work is complementary to ours because they focus on code-switching between high-rate codes for logical Clifford gates while we focus on code-switching between 2D and 3D HGP codes to access transversal CCZ gates.
\end{remark}

We believe that our work opens up many exciting future directions. We discuss some of the avenues below.

\newtheorem*{problem1}{Open Problem 1}

\begin{problem1}
Can we achieve constant spatial overhead for universal fault-tolerant quantum computation with single-shot universality for HGP codes under the local-stochastic noise model?
\end{problem1}

Achieving constant rate in 3D is not typically required for code-switching since we ultimately have to discard the newly obtained logical qubits when we revert back to the 2D code. Nonetheless, having constant rate in 3D could still be useful for logical compilation purposes and the additional logical qubits can be useful as additional logical workspace in certain quantum algorithms. In addition, fault-tolerance under the adversarial noise is a strong assumption that is not typically satisfied in practice. Therefore, it would be interesting to see if we can go beyond individual fault-tolerance for the gadgets with respect to local-stochastic noise but also show that an arbitrary composition of these gadgets can also tolerate local-stochastic noise. A recent result by He~et al.~\cite{he2025composable} provides some insights into this direction.

\newtheorem*{problem2}{Open Problem 2}
\begin{problem2}
Can we generalize our single-shot code-switching framework to other families of quantum error-correcting codes beyond HGP codes?
\end{problem2}

Another promising direction is to explore the application of our single-shot code-switching framework to other exotic families of quantum error-correcting codes. This could potentially lead to new insights and techniques for achieving fault tolerance in a wider range of quantum computing architectures. Bivariate bicycle codes and trivariate tricycle codes are promising candidates because of their extremely high encoding rates and translational symmetries~\cite{bravyi2024high,jacob2025single, menon2025magic}. Other exotic product codes may also give us better distance scalings and qubit overheads~\cite{panteleev2021quantum, breuckmann2021balanced}.

\newtheorem*{problem3}{Open Problem 3}
\begin{problem3}
Does our single-shot code-switching framework possess an advantage over the state-of-the-art magic state distillation (MSD) techniques?
\end{problem3}

We note that recent advances in MSD and lattice surgery techniques have made them extremely efficient both in the practical and asymptotic regime~\cite{yoder2025tour,nguyen2025quantum}. Comparisons between code-switching and magic state distillation in Ref.~\cite{beverland2021cost} for topological codes have showed that code-switching does not
offer substantial savings over state distillation both in terms of space and space-time overhead. Therefore, most of the community has focused on improving MSD and lattice surgery techniques instead of code-switching. However, these comparisons were done for topological codes. It remains an open question whether code-switching can outperform MSD for high-performing QLDPC codes. Our work takes a step towards answering this question by providing a framework for code-switching for high-rate QLDPC codes. A more detailed comparison is left for future work.

\subsection{Outline}
\label{subsec:outline}
The rest of the paper is organized as follows. In Section~\ref{sec:preliminaries}, we provide the necessary preliminaries regarding our notation as well as graphs, classical and quantum codes, chain complexes and homology, HGP codes, soundness and confinement, and the coboundary invariance property required for CCZ gates. In Section~\ref{sec:static}, we present a simple static scheme for universal fault-tolerant quantum computation using 3D HGP codes with transversal CCZ gates just to illustrate the power of having CCZ gates. In Section~\ref{sec:homomorphic_cnot_3d_hgp_codes}, we present a primitive that allows us to perform a one-way homomorphic CNOT between two different 3D HGP codes that would be useful for our single-shot code-switching protocol. In Section~\ref{sec:single_shot_code_switching}, we detail our single-shot code-switching protocol between 2D and 3D HGP codes using the homormophic CNOT developed in Section~\ref{sec:homomorphic_cnot_3d_hgp_codes} and demonstrate how it can be used to implement a universal set of logical gates with single-shot universality. In Section~\ref{sec:3D_soundness}, we prove that 3D HGP codes can exhibit linear confinement which is a new result that is useful for our single-shot state preparation and error correction. In Section~\ref{sec:adversarial_noise}, we analyze the fault tolerance of our single-shot universal protocol as well as its constitutent gadgets under the adversarial noise model. In Section~\ref{sec:local-stochastic}, we analyze the fault tolerance of the gadgets we utilize in our single-shot universal protocol under the local-stochastic noise model. 
\subsection{Acknowledgements}
\label{sec:acknowledgements}
We thank Noah Berthusen, Xiaozhen Fu, Daniel Gottesman, Yaoyun Shi, and Qian Xu for helpful discussions. We also thank Aleksander Kubica for helpful discussions on single-shot error correction.
This material is based upon work supported in part by the
Defense Advanced Research Projects Agency (DARPA)
under Agreement HR00112490357, the NSF QLCI award OMA2120757, the NSF-funded NQVL:QSTD: Pilot: DLPQC and the DoE ASCR Quantum Testbed
Pathfinder program (awards No.~DE-SC0019040 and No.~DE-SC0024220). This work was performed in part at the Kavli Institute for Theoretical Physics (KITP), which is supported by grant NSF PHY-2309135.
SJST acknowledges funding and support from Joint Center for Quantum Information and Computer Science (QuICS) Lanczos Graduate Fellowship and the National University of Singapore (NUS) Development Grant. 

\textbf{AI Usage:} the authors acknowledges the use of Claude in generating some of the figures in this paper and ChatGPT in proofreading for grammar and spelling.  
\section{Preliminaries}
\label{sec:preliminaries}
This section provides the necessary notation and preliminaries for the paper.

\subsection{Notation}
\label{subsec:notation}
For a positive integer $n$, let $[n] = \{1, 2, \ldots, n\}$.
For a prime power $q$, let $\F_q$ denote the finite field of order $q$.
Given a vector $\mbf{v} \in \F_q^n$, let $\|\mbf{v}\| = \left|\set{i \in [n]: \mbf{v}_i \neq 0}\right|$ denote the Hamming weight  of $\mbf{v}$, i.e., the number of nonzero coordinates in the standard basis.
We use boldface uppercase letters (e.g., $\mbf{H}$) to denote matrices and boldface lowercase letters (e.g., $\mbf{v}$) to denote vectors.
We denote the standard basis vector as $\mbf{e}_i \in \F_q^n$, which has a 1 in the $i$-th position and 0s elsewhere.
We also assume that all $\F_q$-vector spaces in our note are finite-dimensional which implies that all Hamming weights are finite. We provide a summary for the notation and conventions that we use in our paper in Table~\ref{tab:notation}.

\begin{table}[htbp]
\centering
\renewcommand{\arraystretch}{1.15}
\begin{tabularx}{\textwidth}{@{}>{$}l<{$} >{\raggedright\arraybackslash}X >{\raggedright\arraybackslash}p{3.0cm} >{\raggedright\arraybackslash}X@{}}
\toprule
\text{Symbol} & Meaning & Type / domain & Notes \\
\midrule
q & prime power (field size) & $\Z_{>0}$ & Field $\F_q$ \\[2pt]
\F_q^{n} & length-$n$ vector space & vector space & Standard basis unless stated \\[2pt]
\mbf{e}_i & standard basis vector & vector & $i$-th entry is 1, others 0 \\[2pt]
\|\mbf{v}\| & Hamming weight & $\N$ & Number of nonzeros in $\mbf{v}$ \\[2pt]
\calC \subseteq \F_q^{n} & classical linear code $[n,k,d]_q$ & code & Dual $\calC^{\perp}$, $\rho=k/n$ \\[2pt]
\mbf{G},\,\mbf{H} & generator / parity-check of $\calC$ & matrices & $\mbf{H}\mbf{G}^{\top}=0$ \\[2pt]
\mathcal{Q} & quantum CSS code $\llbracket n,k,d \rrbracket$ & quantum code & $\mbf{H}_X \mbf{H}_Z^{\top}=0$ \\[2pt]
d_X,\, d_Z,\, d & X/Z distances;  $d=\min(d_X,d_Z)$ & $\N$ & — \\[2pt]
w & LDPC locality (max row/col weight) & $\N$ & — \\[2pt]
\overline{X}_i,\ \overline{Z}_{i} & logical Paulis & operators & $i$-indexed logical qubits\\[2pt]
\mathcal{A} & chain complex & complex & $\cdots \to A_{i+1} \xrightarrow{\partial_{i+1}} A_i \xrightarrow{\partial_i} A_{i-1} \to \cdots$ \\[2pt]
\partial_i, \delta_i & boundary and coboundary maps & linear map & $\delta_i=\partial_{i+1}^{\top}$ \\[2pt]
Z_i(\calA),\, B_i(\calA),\, H_i(\calA) & cycles, boundaries, homology & subspaces & — \\[2pt]
d_i(\calA),\, d^{\,i}(\calA) & systolic and cosystolic distance & $\N$ & Min.\ weight $Z$ and $X$ logical operators \\[2pt]
\mathrm{HGP}(\mbf{H}_1,\mbf{H}_2) & hypergraph product (HGP) code & construction & — \\[2pt]
\calG=(L,R,E) & biregular bipartite graph & graph & Degrees \ $(\ell,r)$,\ neighbors \ N($\cdot$) \\[2pt]
(\gamma,\alpha)  & expansion parameters & reals & Vertex/edge expansion \\[2pt]
T(\calG,\calC) & Tanner code with local code \ $\calC$ & classical code & — \\[2pt]
\mbf{M}_Z & Z-metacheck matrix & matrix & Syndrome repair \\[2pt]
\mathrm{C}^{(\ell)}\mathrm{Z} & multi-controlled-Z gates & gates & Level \ $\ell{+}1$ \ of Clifford hierarchy \\[2pt]
\mathcal{Q}',\ \tilde{\mathcal{Q}} & thickened / expanded code & quantum codes & — \\
\bottomrule
\end{tabularx}
\caption{Symbols and conventions used throughout the paper.\SAM{CHECK EXPANSION PARAMS}\label{tab:notation}}
\end{table}

\subsection{Graphs and Expansion Parameters}
\label{subsec:graphs_and_expansion_parameters}
In this section, we introduce the necessary graph-theoretic concepts and expansion parameters used in our paper. We will be interested in regular graphs with uniform degree $\Delta$. We begin by first reciting the definition of a spectral expander.

\begin{definition}[Spectral expander] \label{def:spectral-expander}
    Given a simple graph $\calG = (V,E)$ with adjacency matrix $A(\calG) \in \F^{\abs{G}\times\abs{G}}_2$. We say that $\calG$ is a $\lambda$-spectral expander when $\lambda$ is the second-largest eigenvalue of $A(\calG)$.
\end{definition}

For a $\Delta$-regular, $\lambda$-spectral expander graph $\calG = (V,E)$, the Cheeger inequalities tell us that the edge expansion constant
\begin{align}
    h(\calG) = \min_{S \subset V: \abs{S}\leq\abs{V}/2} \frac{\abs{\partial S}}{\abs{S}} \, ,
\end{align}
where $\partial S \subset E$, obeys $\frac{1}{2}(\Delta-\lambda) \leq h(\calG) \leq \sqrt{2\Delta(\Delta-\lambda)}$ \cite{hoory06}. Thus, if we desire large edge-expansion, then we want $\lambda$ to be as small as possible. Explicit constructions of regular graphs with such optimal spectral expansion are known \cite{lubotzky2014ramanujan}.

The spectral expander graph is a useful tool for constructing Tanner codes with good distance properties.

\subsection{Classical Linear Codes}
\label{subsec:classical-codes}
This section states the basic definition for classical linear codes that anchor classical coding theory. 
We also briefly define the dual code of a classical linear code.

\begin{definition}[Classical Linear Code]
  \label{def:classical-linear-code}
  For a finite field $\F_q$, a classical linear code of length $n$ and dimension $k$ over $\F_q$ is a $k$-dimensional linear subspace $\calC \subseteq \F_q^n$.
  The rate of $\calC$ is $\rho \coloneqq k/n$. 
  The distance $d$ of $\calC$ is the minimum Hamming weight of a nonzero element of $\calC$, that is $d = \min_{\mbf{c} \in \calC \setminus \{0\}}\|\mbf{c}\|$.
  In general, we refer to $\calC$ as an $[n, k, d]_q$ code.
  The dual code $\calC^\perp \subseteq \F_q^n$ of $\calC$ is defined by $\calC^\perp = \left\{\mbf{x} \in \F_q^n : \mbf{x}\cdot \mbf{y} = 0\; \forall \mbf{y} \in \calC\right\}$, where $\mbf{x} \cdot \mbf{y} = \sum_{i \in [n]}\mbf{x}_i \mbf{y}_i$ denotes the standard bilinear form. 
\end{definition}

\subsection{Classical Tanner Codes on Spectral Expanders}
\label{subsec:classical-tanner-codes-on-left-right-expanders}
In this section, we introduce the classical Tanner codes defined on spectral expanders that we introduced in Definition~\ref{def:spectral-expander}. 
These classical Tanner codes have good error-correcting capabilities that are inherited from the spectral expanders. These classical codes are then used as ingredients in our quantum code constructions.

\begin{definition}
\label{def:tanner_code}
Given a $(\Delta_\ell,\Delta_r)$-biregular bipartitite graph $\calG = (L, R, E)$ where $|L| = n$ and $|R| = m$, and a code $\calC_0 \subseteq \F_2^{\Delta_r}$, we define the Tanner code of $\calG$ and $\calC_0$ as
\[T(\calG, \calC_0) = \set{\mbf{c} \in \F_2^n\;|\; \forall u \in R, \mbf{c}|_{N(u)} \in \calC_0}\]
where $\mbf{c}|_{N(u)} \in \F_2^{\Delta_r}$ denotes the subsequence of $\mbf{c}$ formed by the bits corresponding to the neighbors of $u \in L$ in the graph $\calG$.
\end{definition}

When $\calC_0$ is the regular parity check code, then $T(\calG, \calC_0)$ is simply a classical code whose parity check matrix $H$ is given by the graph adjacency matrix of the graph $\calG$. When $\calC_0$ is linear, the Tanner code $T(\calG, \calC_0)$ is also a classical linear code.
This means that the Tanner code $T(\calG, \calC_0)$ can be described by a parity check matrix $H$ that is constructed from the graph adjacency matrix of the graph $\calG$ and the parity check matrix of the code $\calC_0$. 
For the sake of concreteness, let us suppose that $\calC_0$ is a $[n_0 = \Delta_r, k_0, d_0]$ code with $m_0$ checks. The parity check matrix $H_T(\calG, \calC_0)$ which we abbreviate as $H_T$ of the Tanner code $T(\calG, \calC_0)$ can be constructed as follows:
\begin{enumerate}
  \item Initialize an empty matrix with a number of rows equal to $m \cdot m_0$ and $n$ columns (one for each node in $L$).
  \item For each check node $w_i \in R$ for $i \in [m]$:
  \begin{enumerate}
    \item Take the matrix $H_0 \in \F_2^{m_0 \times n_0}$ of the local code $\calC_0$ which acts only on $v \in L$ such that $v \in N(w_i)$.
    \item Embed each row of $H_0$ into the global matrix $H_T$ by placing the coefficients at the columns corresponding to $N(w_i)$ and zeros elsewhere. 
  \end{enumerate}
\end{enumerate}
\SAM{Restriction to Sipser-Spielman codes should only be done in the CCZ section. Change it later.}
For our work, we are interested in Tanner codes defined on $\Delta$-regular spectral expanders and a linear local code of length $\Delta$. In particular, we can position the bits of the Tanner code on the edges of a $\Delta$-regular graph whereas the checks are the vertices of the graph. Each check imposes the local code's constraints on the $\Delta$ edges incident to it. This class of codes are referred to as Sipser-Spielman codes or classical expander codes. We shall refer to these classical codes as classical expander codes, Tanner codes, and Sipser-Spielman codes interchangeably in this paper. Building the codes with spectral expanders and short linear codes gives us classical expander codes with asymptotically good parameters and efficient, linear-time decoding algorithms~\cite{sipser2002expander}.

\subsection{Quantum Codes}
\label{subsec:quantum-codes}
This section states the basic definitions in quantum coding theory.
In particular, we pay close attention to quantum CSS codes.

\begin{definition}[Quantum CSS Codes]
  \label{def:quantum-css-codes}
  For a finite field $\F_q$, a quantum CSS code of length $n$ over $\F_q$ is specified by a pair of classical codes $\calC_X, \calC_Z \subseteq \F_q^n$ such that $\calC_X^\perp \subseteq \calC_Z$.
  Let $\calQ$ denote the resulting CSS code. Its dimension is $k\coloneqq \dim(\calC_Z) - \dim(\calC_X^\perp)$, and its rate is $\rho \coloneqq k/n$.
  The distance of $\calQ$ is 
  \[d \coloneqq \min\left\{\min_{\mbf{c} \in \calC_X \setminus \calC_Z^\perp}\|\mbf{c}\|,\ \min_{\mbf{c} \in \calC_Z \setminus \calC_X^\perp}\|\mbf{c}\|\right\}.\]
  Sometimes, we differentiate between the $X$ and $Z$ distances of the code and define them as follows:
  \[d_X \coloneqq \min_{\mbf{c} \in \calC_X \setminus \calC_Z^\perp}\|\mbf{c}\|,\quad\quad d_Z \coloneqq \min_{\mbf{c} \in \calC_Z \setminus \calC_X^\perp}\|\mbf{c}\|.\]
  The locality $w$ of $\calQ$ is the maximum number of nonzero entries in any row or column of the parity-check matrices $\mbf{H}_X$ and $\mbf{H}_Z$ of $\calC_X$ and $\calC_Z$.
\end{definition}

For qLDPC codes, the locality $w$ is a constant that is independent of the code length $n$. This implies that each stabilizer generator checks at most a constant number of qubits and each qubit is checked by at most a constant number of stabilizer generators. 

\subsection{Chain Complexes}
\label{subsec:chain-complexes}
This section states the basic definitions in homological algebra.

\begin{definition}[Chain Complexes]\label{def:chain-complexes}
  A chain complex $\calA_\ast$ over a field $\F_q$ consists of a sequence of $\F_q$-vector spaces $\left(A_i\right)_{i \in \Z}$ and linear boundary maps $\left(\partial_i^\calA: A_i \to A_{i-1}\right)_{i \in \Z}$ satisfying $\partial_{i-1}^\calA \circ \partial_i^\calA = 0$ for all $i \in \Z$.
  When clear from context, we omit the superscript and subscript and write $\partial = \partial_i = \partial^\calA = \partial_i^\calA$. 
  Assuming that each $A_i$ has a fixed basis, then the locality $w^\calA$ of $\calA_\ast$ is the maximum number of nonzero entries in any row or column of any matrix $\partial_i$ in this fixed basis.
  If there exist bounds $\ell < m \in \Z$ such that for all $i < \ell$ and $i > m$ we have $A_i = 0$, then we may truncate the sequence and say that $\calA_\ast$ is the $(m-\ell + 1)$-term chain complex
  \[\calA_\ast = \left(A_m \xrightarrow[]{\partial_m} A_{m-1} \xrightarrow[]{\partial_{m-1}} \ldots \xrightarrow[]{\partial_{\ell + 1}} A_\ell\right).\]
  We furthermore define the following (standard) vector spaces for $i \in \Z$:
  \begin{align}
    \text{the space of } i\text{-cycles: } Z_i\left(\calA\right) &\coloneqq \ker\left(\partial_i\right) \subseteq A_i,\\
    \text{the space of } i\text{-boundaries: } B_i\left(\calA\right) &\coloneqq \mathrm{im}\left(\partial_{i+1}\right) \subseteq A_i,\\
    \text{the space of } i\text{-homology: } H_i\left(\calA\right) &\coloneqq Z_i\left(\calA\right)/B_i\left(\calA\right).
  \end{align}

  The cochain complex $\calA^\ast$ associated to $\calA_\ast$ has vector spaces $\left(A^i \coloneqq A_i\right)_{i \in \Z}$ and boundary maps given by the coboundary maps $\left(\delta_i = \partial_{i+1}^{\top} : A^i \to A^{i+1}\right)_{i \in \Z}$ obtained by transposing all the boundary maps of $\calA_\ast$. 
  Thus, the cochain complex is defined as such:
  \[\calA^\ast = \left(A^m \xleftarrow[]{\delta_{m-1}} A^{m - 1} \xleftarrow[]{\delta_{m-2}} \ldots \xleftarrow[]{\delta_{\ell}} A^\ell\right).\]
  We can analogously define the spaces of cohomology $H^i\left(\calA\right) = Z^i(\calA)/B^i(\calA)$, cocycles $Z^i\left(\calA\right) = \ker\left(\delta_i\right)$, and coboundaries $B^i\left(\calA\right) = \mathrm{im}\left(\delta_{i-1}\right)$.
  We typically refer to the chain complex instead of the cochain complex. We therefore often denote the chain complex by $\calA$ and omit the $\ast$ for notational convenience.
\end{definition}

\begin{definition}
  For a chain complex $\calA$, the $i$-systolic distance $d_i(\calA)$ and the $i$-cosystolic distance $d^{\,i}(\calA)$ are defined as
  \[d_i(\calA) = \min_{\mbf{c} \in Z_i(\calA)\setminus B_i(\calA)}\|\mbf{c}\|,\quad\quad d^{\,i}(\calA) = \min_{\mbf{c} \in Z^i(\calA)\setminus B^i(\calA)}\|\mbf{c}\|.\]
\end{definition}

It is well-known that classical linear codes can be described by 2-term chain complexes where the two vector spaces are the spaces of bits and checks respectively. These two vector spaces are connected by a linear boundary map that can be written as the check matrix $\mbf{H}$.
Quantum CSS codes can be described with a 3-term chain complex by associating the $X$ stabilizers, qubits, and $Z$ stabilizers with the three vector spaces in a 3-term chain complex. 
The condition ``the boundary of a boundary is trivial'' is compatible with the CSS orthogonality condition i.e.,  $\partial_{i-1} \circ \partial_i = 0$ and $\mbf{H}_X \mbf{H}_Z^\top = 0$.
To see how the Pauli logical operators fit in the chain complex picture, we associate the $X$ stabilizers, qubits, and $Z$ stabilizers to the $\F_2$-vector spaces $A_0, A_1$, and $A_2$, then the $X$ and $Z$ logical operator representatives are given by the basis elements of the $1$-cohomology space and $1$-homology space respectively. The $X$ and $Z$ distances of the quantum code are then $d^{\,1}(\calA)$ and $d_1(\calA)$.

\subsection{Homological Product}
\label{subsec:homological-product}
This section states the basic notions of the homological product.

\begin{definition}[Homological Product]\label{def:homological-product}
  For chain complexes $\calA$ and $\calB$, the homological product $\calD = \calA \otimes \calB$ is the chain complex given by the vector spaces
  \[D_i \coloneqq \bigoplus_{j \in \Z} A_j \otimes B_{i-j}\]
  and the boundary maps
  \begin{align}\partial_i^\calD \coloneqq \bigoplus_{j \in \Z}\left(\partial_j^\calA \otimes I + (-1)^j I \otimes \partial_{i-j}^\calB\right).\label{eq:boundary-map}\end{align}
  Because we typically work with $\F_2$, the $(-1)^j$ signs can be ignored.
\end{definition}

We now state a well-known result about the homological product.

\begin{proposition}[K\"unneth Formula]\label{prop:kunneth-formula}
  Let $\calA$ and $\calB$ be chain complexes over a field $\F_q$, each with a finite number of nonzero terms.
  Then for every $i \in \Z$,
  \begin{align}H_i(\calA \otimes \calB) \cong \bigoplus_{j \in \Z} H_j(\calA) \otimes H_{i -j}(\calB).\label{eq:kunneth-formula}\end{align}
  Furthermore, for $\mbf{a} \in Z_j(\calA)$ and $\mbf{b} \in Z_{i-j}(\calB)$, the isomorphism above maps
  \[\mbf{a} \otimes \mbf{b} + B_i(\calA \otimes \calB) \mapsto \left(\mbf{a}+B_j(\calA)\right)\otimes \left(\mbf{b}+B_{i-j}(\calB)\right).\]
\end{proposition}

\subsection{Hypergraph Product Codes and Higher-Dimensional Variants}
\label{subsec:hypergraph-product-codes}
The hypergraph product (HGP) code is one of the first constructions of quantum CSS codes that has constant rate and large distance~\cite{tillich2014quantum}. 
The HGP construction is essentially the same as taking the homological product of two 2-term chain complexes that correspond to two classical linear codes.
Let parity-check matrices $\mbf{H}_1 \in \F_2^{m_1 \times n_1}$ and $\mbf{H}_2 \in \F_2^{m_2 \times n_2}$ correspond to two classical linear codes with parameters $[n_1,k_1,d_1]$ and $[n_2, k_2, d_2]$. By taking a homological product of their chain complexes, we obtain the product complex

The hypergraph product (HGP) construction is among the earliest families of CSS quantum codes achieving constant rate together with large distance~\cite{tillich2014quantum}.
Conceptually, HGP is realized by taking the homological product of two 2-term chain complexes associated with a pair of classical linear codes.
Let $\mbf{H}_1 \in \F_2^{m_1 \times n_1}$ and $\mbf{H}_2 \in \F_2^{m_2 \times n_2}$ correspond to two classical linear codes with parameters $[n_1,k_1,d_1]$ and $[n_2, k_2, d_2]$. Forming the homological product of their chain complexes yields the product complex \yifan{Need to rewrite to match Kunneth formula convention}

\begin{equation}\label{eq:HGP complex}
    \begin{quantikz}[wire types={n,n},nodes={inner sep=2pt}, mystyle]
        & {S_1 \otimes S_2} &&& {S_Z} \\
        {B_1 \otimes S_2} && {S_1 \otimes B_2} && Q \\
        \setwiretype{n} & {B_1 \otimes B_2} &&& {S_X}
        \arrow["{\mbf{H}^{}_1 \otimes \ident}", from=2-1, to=1-2]
        \arrow["{\ident \otimes \mbf{H}^{}_2}"', from=2-3, to=1-2]
        \arrow["{\mbf{H}^{}_Z}", from=2-5, to=1-5]
        \arrow["{\ident \otimes \mbf{H}^{}_2}", from=3-2, to=2-1]
        \arrow["{\mbf{H}^{}_1 \otimes \ident}"', from=3-2, to=2-3]
        \arrow["{\mbf{H}^\top_X}", from=3-5, to=2-5]
    \end{quantikz}
\end{equation}
where $S_X, Q, S_Z$ denote, respectively, the binary spaces of $X$-syndromes, qubit errors, and $Z$-syndromes. If we split the columns so that the left block corresponds to $B_1 \otimes S_2$ and the right block to $S_1 \otimes B_2$, then from \eqref{eq:HGP complex} the CSS parity-check matrices of the HGP code are
\begin{subequations}\label{eq:HGP H_X, H_Z}
\begin{align}
    \mbf{H}_X &= \big( \ident_{n_1} \otimes \mbf{H}^{\top}_2 \;\big|\;  \mbf{H}^\top_1 \otimes \ident_{n_2} \big)  \label{eq:HGP H_X} \\
    \mbf{H}_Z &= \big(\mbf{H}^{}_1 \otimes \ident_{m_2} \;\big|\; \ident_{m_1} \otimes \mbf{H}^{}_2 \big) \, . \label{eq:HGP H_Z}
\end{align}
\end{subequations}
Commutation of X and Z stabilizers is immediate, since $\mbf{H}^{}_X \mbf{H}^\top_Z = 2\mbf{H}^{\top}_1 \otimes \mbf{H}^\top_2 = 0$. The quantum code $\mathrm{HGP}(\mbf{H}_1,\mbf{H}_2)$ thus obtained has parameters
\begin{subequations}\label{eq:HGP code parameters}
\begin{align}
    n &= n_1m_2 + m_1n_2  \\
    k &= k^{}_1k^{\top}_2 + k^\top_1k_2  \\
    d_Z &= \min\big(d^{\top}_1,d^\top_2\big)  \\
    d_X &= \min\big(d^{}_1, d_2\big) \, ,
\end{align}
\end{subequations}
with $d_i^\top$ the distance of the dual code defined by $\mbf{H}_i$ for $i \in \{1,2\}$. As usual, we have the quantum code distance $d = \min(d_X,d_Z)$.

Given $\mbf{H}_1,\mbf{H}_2$, choose generator matrices $\mbf{G}_1,\mbf{G}_2$ (so that $\mbf{H}\mbf{G}^{\top} = 0$) whose rows span the respective classical code spaces. In the HGP setting—and mirroring the data-qubit decomposition—the logical qubits split into $k_1k_2^\top$ ``left'' and $k_1^\top k_2$ ``right'' logicals. A convenient canonical basis for the left logical $\overline{X}$ and $\overline{Z}$ operators is \cite{quintavalle2022reshape, quintavalle2023partitioning}
\begin{subequations}\label{eq:HGP G_X, G_Z}
\begin{align}
    \mathbf{G}_{Z,\text{L}} &= \big( \{ \mbf{e}_i \}  \otimes \mbf{G}_2^\top \;|\; \mbf{0} \big)  \label{eq:HGP G_Z}  \\
    \mathbf{G}_{X,\text{L}} &= \big( \mbf{G}_1 \otimes \{ \mbf{e}_j \}  \;|\; \mbf{0} \big)  \label{eq:HGP G_X}  \, ,
\end{align}
\end{subequations}
where $\mbf{e_i} \notin \im(\mbf{H}_1^{\top})$ and $\mbf{e_j} \notin \im(\mbf{H}_2)$ are unit vectors. This guarantees that distinct logical operators are not stabilizer-equivalent. By placing $\mbf{G}$ in standard form, $\mbf{G} = (\ident_k \;|\; A)$, one can arrange that each pair of logical $\overline{X}/\overline{Z}$ either has disjoint support or intersects in exactly one qubit indexed by $(i,j)$ in \eqref{eq:HGP G_X, G_Z}; this single-site intersection can be used to label the left logicals. This refinement is optional: for any row of \eqref{eq:HGP G_X}, there exists a linear combination of rows of \eqref{eq:HGP G_Z} that anticommutes with it and commutes with the rest. In this canonical choice, left logical $\overline{Z}$ operators occupy single columns of the left block and left logical $\overline{X}$ operators occupy single rows of the left block; the same statements hold for the right sector on the right block. Moreover, left logicals have minimum weight $\min(d_1^\top,d_2)$, while right logicals have minimum weight $\min(d_1,d_2^\top)$. If the right logicals are treated as gauge and ignored, the distance reduces to $d = \min(d_1^\top,d_2)$. A precise justification follows from examining the stabilizer structure of the left/right blocks.

\subsection{Higher-Dimensional Hypergraph Product Codes}
\label{subsec:higher-dimensional-hypergraph-product-codes}
In this section, we discuss the construction of higher-dimensional hypergraph product codes as well as review some of their well-known facts. Let $D$ denote the dimension of the HGP code we wish to construct. In particular, we consider the case where $D \geq 3$ and let $\set{\mbf{H}^{(a)} \in \F_2^{m_a \times n_a}}_{a \in [D]}$ be parity-check matrices of $D$ classical linear codes. For each $a \in [D]$, let $\calA^{(a)}$ be the 2-term chain complex associated with $\mbf{H}^{(a)}$ and vector spaces $A^{(a)}_0 = B_a$ and $A^{(a)}_1 = S_a$ correspond to the bit and check vector spaces of the classical code respectively. Then, consider the $D$-fold tensor product complex:
\[\calA = \calA^{(1)} \otimes \calA^{(2)} \otimes \cdots \otimes \calA^{(D)}\]
with degree-$s$ cell spaces given by
\[A_s = \bigoplus_{\substack{J \subseteq [D] \\ |J| = s}}\left(\bigotimes_{i \in J} S_i \otimes \bigotimes_{j \in [D]\setminus J} B_j\right)\]
for $s \in \{0,1,\ldots,D\}$. This is the chain complex that corresponds to the $D$-dimensional hypergraph product code constructed from the $D$ classical codes. To obtain a valid CSS code, we can choose any interior degree $s \in \{1,2,\ldots,D-1\}$ and associate the vector space $A_s$ with the qubits, $A_{s-1}$ with the $X$ checks, and $A_{s+1}$ with the $Z$ checks. The resulting CSS parity-check matrices are given by
\begin{subequations}\label{eq:D-HGP H_X, H_Z}
\begin{align}
    \mbf{H}_X &= \partial_s\,:\,A_S \to A_{s-1}   \\
    \mbf{H}_Z &= \partial_{s+1}\,:\,A_{s+1} \to A_s \, .
\end{align}
\end{subequations}
where the boundary operators can be derived from Equation~\eqref{eq:boundary-map}. 
With a consistent choice of degree labeling, we see that $s = 1$ reproduces the usual $D = 2$ formulas above for the HGP code. We will adopt this convention throughout so that the 2D case is recovered by substitution and the 3D case is given by the $D = 3$ instance of the same construction.

  The logical operators of the $D$-dimensional HGP code can be derived from the K\"unneth formula in Proposition~\ref{prop:kunneth-formula}. In particular, the $X$ and $Z$ logical operators correspond to the basis elements of the $s$-cohomology space $H^s(\calA)$ and $s$-homology space $H_s(\calA)$ respectively which can be derived from Equation~\eqref{eq:kunneth-formula}. 
  \begin{align}
    H_s(\calA) \cong \bigoplus_{t_1 + \ldots + t_D = s} \bigotimes_{a \in [D]} H_{t_a}(\calA^{(a)}) \\
    H^s(\calA) \cong \bigoplus_{t_1 + \ldots + t_D = s} \bigotimes_{a \in [D]} H^{t_a}(\calA^{(a)})
  \end{align}

  The $X$ and $Z$ distances of the code are then given by $d^{\,s}(\calA)$ and $d_s(\calA)$, i.e., the $s$-cosystolic and $s$-systolic distances of the chain complex $\calA$ respectively.

Because we mainly work with 2D and 3D HGP codes in our paper, we provide some concrete details regarding the construction of a 3D HGP code. For an input CSS code $\llbracket n_Q, k_Q, (d_X, d_Z) \rrbracket$ with parity-checks $\mbf{H}_X, \mbf{H}_Z \in \F_2^{m_Q \times n_Q}$ and a classical code $[n_c, k_c, d_c]$ with parity-check $\mbf{H} \in \F_2^{m_c \times n_c}$, we can form the Kronecker (tensor) product of their underlying chain complexes and associate the result with a homological product code \cite{freedman2013quantum, bravyi2014homological, zeng2019higher, campbell2019theory}. This construction generalizes the usual hypergraph product to higher ``dimensions''/folds in the algebraic—not strictly topological—sense (though several higher-dimensional topological codes fit this framework).
As an illustration, the 3D surface code can be viewed either as the homological product of the 2D surface code with a 1D repetition code, or equivalently as a 3-fold hypergraph product of classical repetition codes, i.e., a 3D HGP code. Mirroring the 3D surface code's freedom over which Pauli type becomes a membrane, the homological product admits two choices. Without loss of generality, we take the version that extends the quantum code's logical $Z$-type; the $X$ case follows by reversing the product-complex orientation. The tensor-product complex is
\begin{equation}\label{eq:3D HGP complex}
\begin{quantikz}[wire types={n,n},nodes={inner sep=2pt}, mystyle]
      & {(S_Z,A_1)} &&& {\tilde{R}_{Z}} \\
      {(S_Z,A_0)} && {(Q,A_1)} && {\tilde{S}_{Z}} \\
      \setwiretype{n}{(Q,A_0)} && {(S_X,A_1)} && \tilde{Q} \\
      \setwiretype{n}& {(S_X, A_0)} &&& {\tilde{S}_{X}}
      \arrow["{\ident \otimes \mbf{H}}", from=2-1, to=1-2]
      \arrow["{\mbf{H}_Z \otimes \ident}"', from=2-3, to=1-2]
      \arrow["{\tilde{\mbf{M}}_{Z}}", from=2-5, to=1-5]
      \arrow["{\mbf{H}_Z \otimes \ident}", from=3-1, to=2-1]
      \arrow["{\ident \otimes \mbf{H}}", from=3-1, to=2-3]
      \arrow["{\mbf{H}_X^\top \otimes \ident}"', from=3-3, to=2-3]
      \arrow["{\tilde{\mbf{H}}_{Z}}", from=3-5, to=2-5]
      \arrow["{\mbf{H}^\top_X \otimes \ident}", from=4-2, to=3-1]
      \arrow["{\ident \otimes \mbf{H}}", from=4-2, to=3-3]
      \arrow["{\tilde{\mbf{H}}^\top_{X}}", from=4-5, to=3-5]
  \end{quantikz}
\end{equation}
where $A_0$ and $A_1$ are the bit and check vector spaces of the classical code respectively. The new CSS parity-check matrices are given by
\begin{subequations}\renewcommand*{\arraystretch}{1.3}
\label{eq:3D HGP H_X,H_Z}
\begin{align}
    \tilde{\mbf{H}}_X &= \big(\, \mbf{H}_X \otimes \ident \;\big|\; \ident \otimes \mbf{H}^\top \,\big)  \label{eq:3D HGP H_X}  \\
    \tilde{\mbf{H}}_Z &= \left(\begin{array}{c|c}
        \mbf{H}_Z \otimes \ident \;\; & \mathbf{0}  \\
        \ident \otimes \mbf{H} & \; \mbf{H}^\top_X \otimes \ident
    \end{array}\right) \, .   \label{eq:3D HGP H_Z}
\end{align}
\end{subequations}
Geometrically, the Tanner graph looks like the Euclidean graph product of the inputs, with nodes relabeled as in \eqref{eq:3D HGP complex}. When the input parity-checks are full rank, homology yields the parameters \cite{zeng2019higher}:
\begin{subequations}\label{eqn:homological_classical_code_params}
\begin{align} 
    \tilde{n} &= n_{\rm Q}n_{\rm c} + m_Zm_{\rm c}  \\
    \tilde{k} &\geq k_{\rm Q}k_{\rm c} \label{eqn:kunneth_example_1}  \\
    \tilde{d}_X &= d_X d_{\rm c} \\
    \tilde{d}_Z &= d_Z \, .
\end{align}
\end{subequations}

In \eqref{eqn:kunneth_example_1}, we used the K\"unneth formula (Proposition~\ref{prop:kunneth-formula}) to compute $\tilde{k}$. Let us write the (co)homologies of the quantum and classical codes as $H[{\calQ}]$ and $H[{\calC}]$ respectively. Adopting the convention that logical $\overline{X}$ operators of the CSS code correspond to $H^1[{\calQ}]$, while the classical code's codespace is $H^0[{\calC}]$ (bits/checks placed on the 0/1-cells of its 2-term complex), K\"unneth gives
\[\tilde{H}^1 \cong H^1[{\calQ}] \otimes H^0[{\calC}].\]
Since the rank of $\tilde{H}^1$ counts the logical qubits of the product code, summing the products of ranks of the constituent groups yields \eqref{eqn:kunneth_example_1}.

From \eqref{eq:3D HGP H_X,H_Z} we also see, as in HGP, a left/right decomposition of qubits. With full-rank inputs, only the left sector carries logical qubits. For these, a canonical basis of logical $\overline{X}$ and $\overline{Z}$ operators is 
\begin{subequations}\label{eq:3D HGP G_X,G_Z}
\begin{align}
    \tilde{\mbf{G}}_{Z,{\rm L}} &= \left(\, \mbf{G}_Z \otimes \{ \mathbf{e}_i \} \;\big|\; \mathbf{0} \,\right) \label{eq:3D HGP G_Z} \\
    \tilde{\mbf{G}}_{X,{\rm L}} &= \left(\, \mbf{G}_X \otimes \mbf{G} \;\big|\; \mathbf{0} \,\right)  \, , \label{eq:3D HGP G_X}
\end{align}
\end{subequations}
with $\mbf{G}_X$ and $\mbf{G}_Z$ a symplectic logical basis for the input quantum code, $\{ \mathbf{e}_i \} \notin \mathrm{rs}(h)$, and $\mbf{G}$ is the generator matrix for the classical code $\calC$. If the input CSS code is itself HGP, we may take $\mbf{G}_Z$ and $\mbf{G}_X$ from the canonical basis \eqref{eq:HGP G_X, G_Z} and index logicals by triples $(i,j,k)$.

\subsection{Confinement and Soundness}
\label{subsec:confinement-soundness}
Confinement and soundness are two important properties of quantum codes that are used to analyze the performance of quantum error-correcting codes.
In this section, we state the definitions of confinement and soundness for quantum codes.

\begin{definition}[Confinement~{\cite{quintavalle2021single}}] \label{def:confinement}
  Let $t > 0$ be an integer and $f\,:\, \Z \to \R$ be an increasing function. For a parity-check matrix $\mbf{H} \in \F_2^{m \times n}$, we say it is $(t, f)$-confining if for any Pauli errors $\mbf{x}$ with reduced weight $\|\mbf{x}\| \leq t$, its syndrome $\sigma(\mbf{x}) = \mbf{H}\mbf{x}$ obeys
  \[f(\|\sigma(\mbf{x})\|) \geq \|\mbf{x}\|.\]
\end{definition}

Confinement can be understood as a condition that lower bounds the weight of the error syndrome given the weight of the error. In other words, if the weight of the error is large, then the weight of the syndrome cannot be too small. Given Definition~\ref{def:confinement}, we can also define what good linear confinement entails.

\begin{definition}[Good Linear Confinement~{\cite{quintavalle2021single}}] \label{def:good-linear-confinement}
  Consider an infinite check family $\calM_n$. We say the family has good linear confinement if each $\calM_n$ is $(t, f)$-confined where:
  \begin{enumerate}
    \item $t$ grows with $n$ such that $t \geq an^b$ for some positive constants $a, b$. That is $t \in \Omega(n^b)$ with $b > 0$;
    \item and $f(x)$ is some linear function that is monotonically increasing with $x$ and independent of $n$.
  \end{enumerate}
\end{definition}

In Ref.~\cite{quintavalle2021single}, the authors show that good linear confinement is sufficient to ensure that the quantum code is single-shot against both adversarial and stochastic noise. Prior to their work, Leverrier, Tillich and Zemor showed that hypergraph products of classical expander codes built from lossless expanders have good linear confinement~\cite{leverrier2015quantum}.

\begin{lemma}[Linear Confinement of Expander HGP Codes~{\cite{leverrier2015quantum}}] \label{lem:linear-confinement-expander-codes}
  There exist families of hypergraph product codes, called quantum expander codes, such that an arbitrary error $\mbf{x}$ with reduced weight $\|\mbf{x}\| < t = \Theta(\sqrt{n})$ has a syndrome with weight bounded from below as $\|\sigma(\mbf{x})\| \geq \|\mbf{x}\|/3$.
\end{lemma}

In other words, expander HGP codes have good linear confinement and are therefore single-shot against both adversarial and stochastic noise. Thus, it is possible to just perform a single round of (noisy) syndrome measurement for error correction. This is in contrast to the case of surface codes where we need to perform $O(d)$ rounds of syndrome measurements to perform error correction. In addition, the authors of Ref.~\cite{quintavalle2021single} also prove the following for all 3D homological product codes. Letting $X$-confinement be the confinement with respect to Pauli $Z$ errors, we have the following theorem.

\begin{theorem}[Good $X$ Confinement of 3D Homological Product Codes~{\cite[Theorem 3]{quintavalle2021single}}] \label{thm:good-X-linear-confinement-3D-homological-product-codes}
  All 3D homological product codes have $(t, f)$ X-confinement, where $t = d_Z$ and $f(x) = x^3/4$ or better.  
\end{theorem}

Now, we can define soundness for quantum codes.

\begin{definition}[Soundness~{\cite{campbell2019theory}}]\label{def:soundness}
  Let $t$ be an integer and $f\,:\,\Z\to \R$ be some function called the soundness function with $f(0) = 0$. Given some set of Pauli checks $M$, we say it is $(t, f)$-sound if for all Pauli errors $\mbf{x}$ with $|\sigma(\mbf{x})| = x < t$, it follows that there exists an $\mbf{e}'$ with $\sigma(\mbf{e}') = \sigma(\mbf{x})$ such that $|\mbf{e}'| \leq f(x)$.
\end{definition}

Soundness can be understood as a condition that upper bounds the weight of the error given the weight of the syndrome. In other words, if the weight of the syndrome is small, then there exists an error that matches the syndrome and has small weight. We can also define what good soundness entails.

\begin{definition}[Good Soundness~{\cite{campbell2019theory}}] \label{def:good-soundness}
  Consider an infinite check family $\calM_n$. We say the family has good soundness if each $\calM_n$ is $(t, f)$-sound where:
  \begin{enumerate}
    \item $t$ grows with $n$ such that $t \geq an^b$ for some positive constants $a, b$. That is $t \in \Omega(n^b)$ with $b > 0$;
    \item and $f(x)$ is some polynomial function that is monotonically increasing with $x$ and independent of $n$.
  \end{enumerate}
\end{definition}

While redundant checks are not strictly necessary for good soundness of the quantum code, it should not be too hard to see that increasing the number of redundant checks will typically improve the soundness of the quantum code as it raises $t$ and makes the function $f$ ``smaller''.
We can also use Definitions~\ref{def:confinement} and \ref{def:soundness} to conclude that good soundness implies good confinement for qLDPC codes.

\begin{lemma}[Good Soundness Implies Good Confinement~{\cite[Lemma 2]{quintavalle2021single}}]
\label{lem:soundness implies confinement}
  Consider a qLDPC code that is $(t, f)$-sound with increasing $f$. If its qubit degree is at most $\omega$, then it has $(t/\omega, f)$-confinement.
\end{lemma}

Next, we state the following lemma from Ref.~\cite{campbell2019theory}:
\begin{lemma}[Soundness of 4D Homological Product Codes~{\cite[Restatement of Lemma 6]{campbell2019theory}}]
  \label{lem:4D HGP soundness}
  Let 
  \[\calA = \left(A_{0} \xrightarrow[]{\delta_{0}} A_1 \xrightarrow[]{\delta_1} A_2\right)\]
  be a chain complex such that $\delta_1^\top$ is $(t, f)$-sound and $\delta_{0}$ is $(t, f)$-sound with $f(x) = x^2/4$. 
  Applying the homological product with two 1-term chain complexes of two classical codes, we obtain a new length-5 chain complex 
  \[\breve{\calA} = \left(\breve{A}_{-1} \xrightarrow[]{\breve{\delta}_{-1}} \breve{A}_{0} \xrightarrow[]{\breve{\delta}_{0}} \breve{A}_{1} \xrightarrow[]{\breve{\delta}_{1}} \breve{A}_{2} \xrightarrow[]{\breve{\delta}_{2}} \breve{A}_{3}\right)\]
  where the map $\breve{\delta}_1$ is $(t, g)$-sound and $\breve{\delta}_{0}^\top$ is $(t, g)$-sound with $g(x) = x^3/4$.
\end{lemma}

\subsection{Coboundary expansion}
\label{subsec:coboundary expansion}

In this section, we review the notions of (co)boundary expansion for (co)chain complexes and their relations to confinement and soundness in the previous subsection.

\begin{definition}[Small-set coboundary expansion \cite{hopkins2022explicit}]
    Let $A_0 \xrightarrow{\delta_0} A_1 \xrightarrow{\delta_1} A_2$ be a cochain complex. We say that $\delta_1$ is $(t,\alpha)$-small-set coboundary expanding if for all $a_1 \in A_1$ with $\abs{a_1} \leq t$, we have
    \begin{align}
        \abs{\delta_1 a_1} \geq \alpha \cdot \min_{a_0\in A_0} \abs{a_1 + \delta_0 a_0} \, .
    \end{align}
\end{definition}
Intuitively, coboundary expansion asserts that if $a_1$ is far from any coboundary with respect to $\delta_0$, then its coboundary with respect $\delta_0$ must be large. The small-set label just means that this property only holds below a certain radius $t$.

There is also a closely related notion of local (co)minimality and locally (co)minimal distance.

\begin{definition}[Locally co-minimal \cite{dinur2023good}]
\label{def:locally minimal}
    Let $A_0 \xrightarrow{\delta_0} A_1 \xrightarrow{\delta_1} A_2$ be a cochain complex. We say that $a_1 \in A_1$ is locally co-minimal if $\forall a_0 \in A_0$ with $\abs{a_0}=1$, we have
    \begin{align}
        \abs{a_1} \leq \abs{a_1 + \delta_0 a_0} \, .
    \end{align}
\end{definition}

Observe that if $a_1$ is not locally co-minimal, then we locally reduce its weight by applying a single coboundary.

\begin{definition}[Locally co-minimal distance \cite{dinur2024expansion}]
\label{def:locally minimal distance}
    Let $A_0 \xrightarrow{\delta_0} A_1 \xrightarrow{\delta_1} A_2$ be a cochain complex. The locally co-minimal distance (at level 1) is given by
    \begin{align}
        d_{\rm coLM} \coloneqq \min\{ \abs{a_1} : a_1 \in \ker\delta_1 \,,\, a_1 \text{ is locally co-minimal} \} \, .
    \end{align}
\end{definition}

\subsection{Coboundary Invariance and CCZ gates}
\label{subsec:coboundary-invariance-ccz}

In this section, we discuss the coboundary invariance property of certain quantum codes and its implications for the implementation of CCZ gates. Broadly speaking, coboundary invariance guarantees that a set of physical CCZ gates map codewords to codewords since the logical action remains the same even as the code state differs by a coboundary i.e. $X$ stabilizers. We mostly restate the definitions and ideas from Ref.~\cite{lin2024transversal} and refer the reader to the original paper for more details.

Consider the case where we have three quantum codes $\calQ_1, \calQ_2, \calQ_3$ and we want to implement a CCZ gate between the logical qubits of these codes. Suppose we have a trilinear map $f\,:\, \F_2^{n_1 \times n_2 \times n_3} \to \F_2$ where $n_i$ is the length of the code $\calQ_i$. Recall that a function is trilinear if the function is linear when we fix any of the other two inputs. The trilinear map $f$ is constructed such that it captures the desired action of the following unitary operator $CCZ^f$:
\[CCZ^f \ket{\mbf{x}_1, \mbf{x}_2, \mbf{x}_3} = (-1)^{f(\mbf{x}_1, \mbf{x}_2, \mbf{x}_3)}\ket{\mbf{x}_1, \mbf{x}_2, \mbf{x}_3}\]
where $\mbf{x}_i \in \F_2^{n_i}$ is any computational basis state of the code $\calQ_i$ and $\ket{\mbf{x}_1, \mbf{x}_2, \mbf{x}_3}$ is the state on $n_1 + n_2 + n_3$ qubits. Another important way to look at these states is to think of them as the qubits which we apply Pauli $X$ operators on.

In order for $CCZ^f$ to be a valid logical operator, it has to map codewords to codewords. Before we discuss the actual condition in greater detail, we note that any codeword in $\calQ_1 \otimes \calQ_2 \otimes \calQ_3$ can be described as follows:
\[\ket{[\mbf{\zeta}_1], [\mbf{\zeta}_2], [\mbf{\zeta}_3]} \coloneqq \sum_{\mbf{\beta}_1 \in B^1(\calA_1)} \sum_{\mbf{\beta}_2 \in B^1(\calA_2)} \sum_{\mbf{\beta}_3 \in B^1(\calA_3)} \ket{\zeta_1 + \beta_1, \zeta_2 + \beta_2, \zeta_3 + \beta_3}\]
for $\zeta_i \in Z^1(\calA_i)$ and $B^1(\calA_i)$ being the space of $X$ stabilizers of the code $\calQ_i$ for $i \in \{1, 2, 3\}$. Note that $[\zeta_i]$ is the logical equivalence class of the element $\zeta_i$.

Because the action of $CCZ^f$ is essentially a change of phase when applied to any computational basis state, we can guarantee that $CCZ^f$ maps codewords to codewords if and only if 
\[CCZ^f \ket{[\zeta_1], [\zeta_2], [\zeta_3]} \propto \ket{[\zeta_1], [\zeta_2], [\zeta_3]}.\]
The aforementioned condition is the same as the following:
\[f(\zeta_1, \zeta_2, \zeta_3) = f(\zeta_1 + \beta_1, \zeta_2 + \beta_2, \zeta_3 + \beta_3)\]
for all $\zeta_i \in Z^1(\calA_i)$ and $\beta_i \in B^1(\calA_i)$. This intuitively means that the action of $CCZ^f$ is uniquely defined on the logical equivalence classes of the codewords.

We now define what it means for a triple of classical codes to satisfy what we term the \emph{multiplication property}:

\begin{definition}[Multiplication Property]
  \label{def:multiplication-property}
  Let $\calC_1, \calC_2, \calC_3 \subseteq \F_2^{\Delta}$ be classical linear codes of length $\Delta$. Let us denote the classical code $\calC_1 \odot \calC_2 \odot \calC_3$ as the code generated by the following:
  \[\calC_1 \odot \calC_2 \odot \calC_3 = \spn\set{\mbf{c}_1 \ast \mbf{c}_2 \ast \mbf{c}_3\,|\,\mbf{c}_1 \in \calC_1, \mbf{c}_2 \in \calC_2, \mbf{c}_3 \in \calC_3}\]
  where $\ast$ denotes the component-wise multiplication of vectors.

  We say that the triple $(\calC_1, \calC_2, \calC_3)$ satisfies the multiplication property if $\calC_1 \odot \calC_2 \odot \calC_3 \subseteq \calC' \subseteq \F_2^{\Delta}$ for the classical parity check code $\calC'$, which is the dual of the repetition code. \david{changed} In other words, we have 
  \[\sum_{i = 1}^\Delta (\mbf{c}_1)_i (\mbf{c}_2)_i (\mbf{c}_3)_i = 0.\] 
\end{definition}

In the recent work done by Lin, an explicit connection was made between the multiplication property and a quantum code's amenability to transversal logical CCZ gates~\cite{lin2024transversal}. While it was originally stated for sheaf codes, this connection can be made for higher-dimensional HGP codes as well since they are captured by the same sheaf formalism. Before we state the theorem from Ref.~\cite{lin2024transversal}, we first define the cup product operation on cochains:
\begin{definition}[Cup Product on Cochains]
  Let $A^k$ be the space of $k$-cochains on a cell complex. The cup product is a bilinear map
  \[
  \smile\,:\,A^k \times A^{\ell} \to A^{k+\ell}
  \]
  defined on the level of cochains by
  \[
  (\mbf{c}_1 \smile \mbf{c}_2)(\tau) = \mbf{c}_1(\sigma_1) \mbf{c}_2(\sigma_2)
  \]
  where $\sigma_1$ and $\sigma_2$ are the faces of $\tau$ such that $\sigma_1 \smile \sigma_2 = \tau$, and $\mbf{c}_1$ and $\mbf{c}_2$ are cochains in $A^k$ and $A^\ell$, respectively. The cup product is associative and satisfies the Leibniz rule with respect to the coboundary operator $\delta$, i.e.,
  \[
  \delta(\mbf{c}_1 \smile \mbf{c}_2) = (\delta \mbf{c}_1) \smile \mbf{c}_2 + (-1)^k \mbf{c}_1 \smile (\delta \mbf{c}_2),
  \]
  where $\mbf{c}_1 \in A^k$ and $\mbf{c}_2 \in A^\ell$.
\end{definition}

Now, we are ready to state the theorem from Ref.~\cite{lin2024transversal} that connects the multiplication property to the implementation of transversal CCZ gates on 3D HGP codes.

\begin{theorem}[Local Multiplication and CCZ Property~{\cite[Restatement of Theorem 6.8]{lin2024transversal}}]
  \label{thm:local-multiplication-ccz-property}
  Let $\calQ_1(\set{G_{i}}_{i = 1}^{3}, \set{\calC_{1,j}}_{j=1}^3)$, $\calQ_2(\set{G_{i}}_{i = 1}^{3}, \set{\calC_{2,j}}_{j=1}^3)$, and $\calQ_3(\set{G_{i}}_{i = 1}^{3}, \set{\calC_{3,j}}_{j=1}^3)$ be three 3D HGP codes where we have $G_{i}$ for $i = 1, 2, 3$ as the respective base spectral expander graphs and $\calC_{1,j}, \calC_{2,j}, \calC_{3,j}$ are the linear local codes for each of the spectral expanders used to construct the Sipser-Spielman codes used in the HGP construction for each of the three 3D HGP codes.
  Let us define a trilinear function
  \begin{align*}
  f\,:\,A^1 \times A^1 \times A^1 &\to \F_2\\
  f(\mbf{q}_1, \mbf{q}_2, \mbf{q}_3) &\mapsto \sum_{\tau \in A_3} ((\mbf{q}_1 \smile \mbf{q}_2)\smile \mbf{q}_3)(\tau)
  \end{align*}
  where the three $A^1$s are the 1-cochains of the cochain complexes associated with the three 3D HGP codes and $\tau \in A_3$ denotes the basis elements of the 3-chains of the chain complex constructed from the tensor product between the three spectral expanders $\{G_i\}_{i=1}^3$. Define the unitary operator $CCZ^f$ as the following:
  \[CCZ^f = \prod_{j_1, j_2, j_3 \in [n]} CCZ_{j_1, j_2, j_3}^{f(\mbf{e}_{j_1}, \mbf{e}_{j_2}, \mbf{e}_{j_3})}\]
  where $CCZ_{j_1, j_2, j_3}^{f(\mbf{e}_{j_1}, \mbf{e}_{j_2}, \mbf{e}_{j_3})}$ acts on the $j_1, j_2, j_3$-th physical qubits of $\calQ_1, \calQ_2, \calQ_3$ respectively if the trilinear map $f$ evaluates nontrivially on the corresponding basis elements $\mbf{e}_{j_1}, \mbf{e}_{j_2}, \mbf{e}_{j_3}$ of the 1-cochains. If $(\calC_{1, j} \odot \calC_{2, j} \odot \calC_{3, j})$ satisfies the multiplication property for every $j = 1, 2, 3$, then the $CCZ^f$ implements a logical $CCZ$ gate across the logical qubits of the three 3D HGP codes.
\end{theorem}

\section{Static Fault-Tolerant Single-Shot Universal Computation}
\label{sec:static}

In this section, we present a \emph{static} fault-tolerant single-shot universal computation scheme based on some of the recent constructions of (near-)asymptotically good QLDPC codes with transversal CCZ gates~\cite{golowich2025quantum, lin2024transversal, breuckmann2024cups,zhu2025topological,zhu2025transversal}.
We term it \emph{static} because it does not involve dimension jumping or code switching.
The codes constructed in the references listed above are all of the homological product variety. To be precise, all of them are 3D HGP codes that are obtained from the homological product of three classical Sipser-Spielman codes~\cite{sipser2002expander}.
To describe how we might be able to achieve fault-tolerant universal computation, we state the two important facts. The first fact pertains to the computational universality of a small gate set~\cite{shi2002both, aharonov2003simple, childs2017lecture}.

\begin{fact}[{\cite[Computationally Universal Gate Set]{shi2002both, aharonov2003simple,childs2017lecture}}]
  \label{fact:computationally_universal}
  The gate set that contains the Hadamard gate $H$ and the controlled-controlled-Z gate $CCZ$ is computationally universal.   
\end{fact}

Fact~\ref{fact:computationally_universal} implies that we can approximate any orthogonal matrix to arbitrary precision using the gate set $\{H, CCZ\}$ as long as we have access to ancilla qubits.
However, because both $H$ and $CCZ$ have real matrix representations, they are clearly unable to approximate complex unitary matrices.
Nonetheless, it is possible for us to simulate the effect of arbitrary complex unitary operations using $H$ and $CCZ$ gates by simulating the real and imaginary part of the operations separately.

Let us define the $C^{(\ell)}Z$ gate as the controlled-$Z$ gate that is controlled on $\ell$ qubits and applies $Z$ to the target qubit if and only if all the control qubits are in the state $\ket{1}$.
It is known that the $C^{(\ell)}Z$ gate is in the $\left(l+1\right)$-th level of the Clifford hierarchy.
In other words, the $CCZ$ gate is in the third level of the Clifford hierarchy and lies outside of the Clifford group that contains the $CZ$ gate.
We now state another well-known fact:

\begin{fact}[Descending the Diagonal Clifford Hierarchy]
  \label{fact:descending_clifford_hierarchy}
  Let $C^{(\ell)}Z$ be a gate that is controlled on the first $\ell$ qubits and applies $Z$ to the $\left(\ell + 1\right)$-th qubit if and only if all the control qubits are in the state $\ket{1}$.
  For any $i \in [\ell]$, we have
  \[X_i C^{(\ell)}Z_{1, \ldots, \ell + 1} X_i C^{(\ell)}Z_{1, \ldots, \ell + 1} = C^{(\ell-1)}Z_{1, \ldots, i-1, i+1,\ldots, \ell+1}.\]
  Note that the $C^{(\ell - 1)}Z$ gate is controlled on the same $\ell$ qubits except for the $i$-th qubit.
\end{fact}
In other words, assuming we have a transversal implementation of $C^{(\ell)}Z$, we can descend the Clifford hierarchy by conjugating the gate with $X_i$ for any $i \in [\ell]$ to obtain a transversal implementation of $C^{(\ell-1)}Z$ as long as the target qubit remains the same. This fact was exploited in Ref.~\cite{he2025quantum} to build quantum codes with addressable CCZ gates.
Using Fact~\ref{fact:descending_clifford_hierarchy}, we obtain the following fact:
\begin{fact}[Transversal CCZ to CZ]
  \label{fact:transversal_CCZ_to_transversal_CZ}
  All the CSS codes that have transversal $CCZ$ gates constructed in Refs.~\cite{zhu2023non, wills2024constant, golowich2025quantum, lin2024transversal, nguyen2024good, breuckmann2024cups, scruby2024quantum,zhu2025topological, he2025quantum} possess a transversal implementation of the $CZ$ gate.
\end{fact}

Now, we proceed to describe a quantum teleportation scheme that allows us to implement the Hadamard gate using the CZ gate, Pauli-X measurements, and an ancilla $\ket{+}$ state.

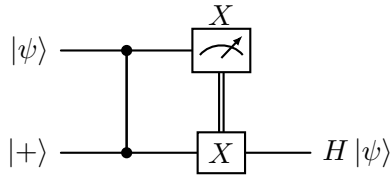
\begin{figure}[H]
  \centering
  \begin{quantikz}[column sep=0.8cm, row sep=0.8cm]
  \lstick{\(\ket{\psi}\)}
  & \ctrl{1}
  & \meter[label style={inner sep=1pt}]{X}  \\
  \lstick{\(\ket{+}\)}
  & \ctrl{-1}
  & \gate{X}\wire[u]{c}
  & \rstick{$H\ket{\psi}$}
  \end{quantikz}
\end{figure}

In the case where $\ket{\psi}$ and $\ket{+}$ are logical states of a 3D HGP code $\tilde{\calQ}$ that encodes many logical qubits, performing the logical Hadamard on any subset of logical qubits might appear to be challenging especially when we are constrained by the addressability of the transversal logical CZ gate and the tranversal logical X measurement gadget that might collapse all logical qubits in the code block. Fortunately, we can employ the Grid Pauli Product Measurement (GPPM) gadget introduced by Xu \emph{et al.}~\cite{xu2025fast} to perform the necessary logical teleportation and measurements to ensure that the circuit shown above only acts on the logical qubits that we want to perform a logical Hadamard on. 

By combining Fact~\ref{fact:transversal_CCZ_to_transversal_CZ} and the quantum teleportation scheme in the above figure, we obtain the ability to perform a computationally universal gate set that contains the Hadamard gate and the CCZ gate.
In addition, the transversal implementation of both components imply that we have a fault-tolerant implementation of these gates that take $O(1)$ logical cycles. 
The only caveat that remains is to show how we can fault-tolerantly prepare the ancilla logical $\ket{\overline{+}}$ state for the gate teleportation scheme.

Because our 3-dimensional homological product is not a 5-term chain complex but a 4-term chain complex, only one of our chain maps can be $(t, g)$-sound. We now reuse and adapt the notation from Lemma~\ref{lem:4D HGP soundness} that discusses the soundness of a 4D HGP code from the chain complex picture.
Let us assign $\tilde{A}_{0}, \tilde{A}_1, \tilde{A}_2$ to be the vector spaces that correspond to the $X$ stabilizer generators, qubits, and $Z$ stabilizer generators of the 3D HGP code respectively. In addition, let the code parameters of the 3D HGP code $\tilde{\calQ}$ be $\llbracket \tilde{n}, \tilde{k}, \tilde{d}_X, \tilde{d}_Z\rrbracket$.
Suppose we let $\tilde{A}_{3}$ be a non-trivial vector space in our 4-term chain complex that is spanned by the $Z$ metachecks that are associated to the $Z$ stabilizer generators by the map $\tilde{\delta}_1$.
The result from Lemma~\ref{lem:4D HGP soundness} stil holds for the $Z$ checks of our 3D HGP code.
Now, we can use these $Z$ metachecks to prepare the ancilla $\ket{+}$ state in our 3D HGP code $\tilde{\calQ}$ in a single-shot fashion as described in the following protocol:
\begin{enumerate}
  \item Initialize all physical qubits in the $\ket{+}^{\otimes \tilde{n}}$ state.
  \item Measure all $Z$ stabilizer generators to obtain a preliminary $Z$ syndrome $\vec{s}_0$.
  \item Digitally compute the $Z$ metacheck syndrome based on $\vec{s}_0$ to obtain a $Z$ metasyndrome $\vec{m}$.
  \item Use the metasyndrome $\vec{m}$ to obtain a repaired syndrome $\vec{s}$.
  \item Decode using the repaired syndrome $\vec{s}$ and $\tilde{\delta}_1$ and apply the correction to obtain $\ket{\overline{+}}^{\otimes \tilde{k}}$.
\end{enumerate}

With the above gadgets, we are now able to perform computationally universal fault-tolerant quantum computation on these 3D HGP codes with transversal CCZ gates.
As long as we have access to ancilla blocks that are crucial for gate teleportation for the Hadamard gates, we can achieve the desired computation.
Suppose we have 3 blocks of such 3D HGP codes that are each equipped with an ancilla 3D HGP code, we can now perform (in parallel) logical CCZ on logical qubits across the 3 blocks as well as the Hadamard gate on the logical qubits within each of the 3 blocks using their individual ancilla code blocks.
While our 3D HGP code does not directly allow for a single-shot state preparation of the $\ket{\overline{0}}$ state, we can use the $Z$ metachecks to prepare the $\ket{\overline{+}}$ state in a single-shot fashion and then apply a transversal $H$ gate using the gate teleportation gadget to obtain the $\ket{\overline{0}}$ state. In our scheme, we do not focus too much on the routing problem of moving logical qubits between different code blocks for the implementation of logical gates.
As mentioned before, the Pauli-measurement-based logical gates defined in Ref.~\cite{xu2025fast} may be integrated to the scheme to facilitate logical qubit teleportation between 3D HGP code blocks as well as standard Clifford computation using Pauli-based computation techniques.
In addition, we might also be able to take advantage of the automorphisms inherited from the underlying classical codes of the 3D HGP code and utilize the automorphism gadgets for homological product codes to permute logical qubits fault-tolerantly~\cite{berthusen2025automorphism}.
Coupled with the inter-block $CCZ$ gates and Hadamard gates, we can perform universal computation in a single-shot and fault-tolerant fashion. We defer a more in-depth discussion of the single-shotness to the later sections since they overlap heavily with the single-shotness of the dynamical scheme that we will introduce in Section~\ref{sec:dynamical}.

While the static scheme we have constructed above uses 3D homological product codes, we point out that the scheme can be generalized to any $i$\textsuperscript{th}-dimensional homological product code. By using $i$ blocks of $i$-dimensional homological product codes, we can descend the Clifford hierarchy down to any $C^{(j)}Z$ gate for $j \leq i$ and use the gadgets described above to perform Clifford gates fault-tolerantly. Thus, we can achieve fault-tolerant universal computation using $i$-dimensional homological product codes as long as we have access to ancilla blocks that are crucial for gate teleportation for the Hadamard gates.

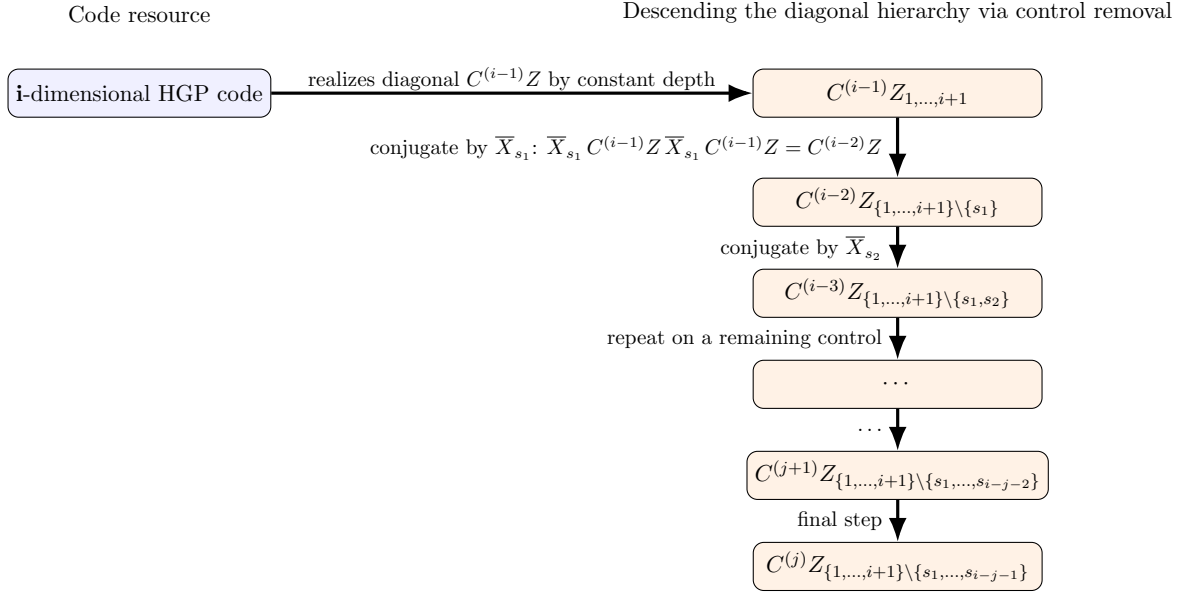
\usetikzlibrary{arrows.meta,positioning,calc,fit,decorations.pathreplacing,backgrounds}

\tikzset{
  code/.style={draw, rounded corners, fill=blue!6, minimum width=42mm, minimum height=8mm, align=center, font=\normalsize},
  gate/.style={draw, rounded corners, fill=orange!10, minimum width=48mm, minimum height=8mm, align=center, font=\normalsize},
  map/.style={-{Latex[length=3.2mm]}, very thick},
  step/.style={font=\scriptsize, fill=white, inner sep=1pt, text=black},
  note/.style={font=\scriptsize, align=left},
  brace/.style={decorate, decoration={brace, amplitude=5pt}},
}

\begin{tikzpicture}[node distance=7mm and 18mm, scale=0.8, transform shape]

\node[code] (HPC) {$\mathbf{i}$-dimensional HGP code};
\node[gate, right=80mm of HPC] (Gtop) {$C^{(i-1)}Z_{1,\ldots,i+1}$};

\draw[map] (HPC) -- node[above, font=\small, inner sep=1pt]{realizes diagonal $C^{(i-1)}Z$ by constant depth} (Gtop);

\node[gate, below=10mm of Gtop] (G1) {$C^{(i{-}2)}Z_{\{1,\ldots,i+1\} \setminus \{s_1\}}$};
\node[gate, below=7mm of G1]    (G2) {$C^{(i{-}3)}Z_{\{1,\ldots,i+1\} \setminus \{s_1,s_2\}}$};
\node[gate, below=7mm of G2]    (G3) {$\cdots$};
\node[gate, below=7mm of G3]    (Gcz) {$C^{(j+1)}Z_{\{1,\ldots,i+1\} \setminus \{s_1,\ldots, s_{i - j - 2}\}}$};
\node[gate, below=7mm of Gcz]   (Gz)  {$C^{(j)}Z_{\{1,\ldots,i+1\} \setminus \{s_1,\ldots, s_{i - j - 1}\}}$};

\draw[map] (Gtop) -- node[left=2mm, font=\small, fill=white, inner sep=1pt]{conjugate by $\overline{X}_{s_1}$: $\overline{X}_{s_1}\,C^{(i-1)}Z\,\overline{X}_{s_1}\,C^{(i-1)}Z = C^{(i-2)}Z$} (G1);
\draw[map] (G1) -- node[left=2mm, font=\small, fill=white, inner sep=1pt]{conjugate by $\overline{X}_{s_2}$} (G2);
\draw[map] (G2) -- node[left=2mm, font=\small, fill=white, inner sep=1pt]{repeat on a remaining control} (G3);
\draw[map] (G3) -- node[left=2mm, font=\small, fill=white, inner sep=1pt]{$\ldots$} (Gcz);
\draw[map] (Gcz) -- node[left=2mm, font=\small, fill=white, inner sep=1pt]{final step} (Gz);

\node[draw=none, rounded corners, fit=(HPC), inner sep=5mm, label={[note, font=\normalsize]above:Code resource}] (LBOX) {};
\node[draw=none, rounded corners, fit=(Gtop)(Gz)(G1)(G2)(G3)(Gcz), inner sep=5mm, label={[note, font=\normalsize]above:Descending the diagonal hierarchy via control removal}] (RBOX) {};

\end{tikzpicture}

We note that an alternative approach is to consider six-dimensional hypergraph product (6D HGP) codes. The 3D HGP codes we have considered so far only have one chain map that is $(t,g)$-sound, which allows us to perform single-shot state preparation and measurement in only one basis (the $X$ basis in our case). This is because the 4-term chain complex that corresponds to a 3D HGP code only has one non-trivial boundary map that can be made $(t,g)$-sound. To achieve single-shot state preparation in both bases, we need soundness in both $X$ and $Z$ bases without giving up non-Clifford transversality. A 6D HGP code yields a longer chain complex that can be arranged so that two distinct boundary maps are $(t,g)$-sound, providing independent metachecks for repairing both $Z$ and $X$ syndromes and hence single-shot state preparation in both bases. At the same time, by associating the appropriate with vector spaces in the chain complex with qubits and checks such that the $X$ logical operators are 4-dimensional membranes and the $Z$ logical operators are 2-dimensional membranes, the HGP code now supports a transversal CCZ gate. At the same time, the code has soundness in both basis which allows for single-shot error correction and state preparation in both basis. One can think of this as the high-rate version of the 6D color code that Bomb{\'\i}n explored in Ref.~\cite{bombin2009self}. Combined with Hadamard via teleportation, this preserves computational universality while achieving dual-basis soundness. However, it is known that the rate of the code worsens as we increase the dimensionality. There are practical reasons to stay in 3D instead of 6D. Suppose the classical codes used to construct the HGP codes have rate approximately 1/2. Then, a back of the envelope calculation shows that the 3D HGP code will have rate 1/8 and the 6D HGP code will have rate 1/64. 

From the above discussion regarding the rate of higher-dimensional HGP codes, it should be apparent that the static scheme we have constructed above is not necessarily the most efficient way to achieve universal computation.
While we are able to enjoy potential time savings from not having to switch between codes, a 3D homological product code suffers from a lower rate compared to a 2D code. In addition, for certain hardware architectures, it may be more efficient or even necessary to keep the code in 2D for the most parts unless we need to perform logical non-Clifford gates which we then switch to 3D codes for the duration of the non-Clifford computation. In some sense, this shares parallels with the motivations that underlie the work done by Hill \emph{et al.}~\cite{hill2011fault}.
Thus, it may be overall more efficient in terms of space-time costs to use a dynamical scheme that switches between 2D and 3D codes to achieve universal computation. An exact characterization of the conditions under which such a scheme would be beneficial is an important open question for future work.

In the subsequent sections, we first describe a homomorphic CNOT primitive before showing how it can be utilized for a dynamical fault-tolerant single-shot universal computation scheme. This primitive is designed to help facilitate the switching between different code dimensions and ensure the overall efficiency of the computation.

\section{Homomorphic CNOT between a 2D HGP code and a 3D HGP code}
\label{sec:homomorphic_cnot_3d_hgp_codes}
In this section, we describe the 2D and 3D HGP code constructions we utilize for the dimensional expansion and contraction steps. We then describe a chain map between these two HGP codes. Finally, we describe the homomorphic CNOT that we can perform between these codes.

Recall from Definition~\ref{def:tanner_code} that we denote the $n$-bit Sipser-Spielman code constructed from a $\Delta$-regular spectral expander and a code $\calC_0$ as $T(\calG, \calC_0)$. When $\calC_0$ is the regular parity check code, then $T(\calG, \calC_0)$ is simply a classical code whose parity check matrix $\mbf{H}$ is given by the graph adjacency matrix of the graph $\calG$. In our case, we are interested in the case where $\calC_0$ is a length-$\Delta$ repetition code with parameters $[\Delta, 1, \Delta]$ that has a parity-check matrix $\mbf{H}_0$.

Note that we have chosen $\mbf{H}_0$ so that every bit is involved in two checks by introducing a redundant check. The typical repetition code check matrix typically only has $\Delta-1$ linearly independent checks but we introduce an extra check so that the graph representation now forms a loop. This is useful for the expander graph construction that we will use in the subsequent sections. In this form, each of the $n$ classical bits of $T(\calG, \calC_0)$ is involved in an even number of checks.

We now state the following simple lemma that provides a useful property of the Sipser-Spielman code $T(\calG, \calC_0)$.
\begin{lemma}[Codewords of $T(\calG, \calC_0)$]
\label{lemma:tanner_codewords}
Let $\calG = (L, R, E)$ be a $\Delta$-regular spectral expander and let $\calC_0 \subseteq \F_2^{\Delta}$ be a $[\Delta, 1, \Delta]$ repetition code. If $\calG$ is a connected graph, then 
\[T(\calG, \calC_0) = \{0^{n}, 1^{n}\}\]
and $\dim(T(\calG, \calC_0)) = 1$.
\end{lemma}
\begin{proof}
We first note that the codewords of the repetition code $\calC_0$ are simply the all-zero and all-one vectors in $\F_2^{\Delta}$. Suppose the first bit $v_1 \in L$ in the codeword of $T(\calG, \calC_0)$ is $0$. From the fact that $\calG$ is bipartite and connected, for all $v' \in R$ such that $v_1 \in N(v'$), all bits corresponding to the neighbors of $v'$ in $L$ must also be $0$ because of the constraints from $\mbf{H}_0$. We can then extend this argument to all bits in $L$ iteratively because $G$ is connected. Since the only allowed vectors in $\F_2^{\Delta}$ are $0^{\Delta}$ and $1^{\Delta}$, we conclude that
\[T(\calG, \calC_0) = \{0^{n}, 1^{n}\}.\]
\end{proof}

Now, we discuss the choice of the connected bipartite graph $\calG$ that we will use to construct the Sipser-Spielman code $T(\calG, \calC_0)$. In particular, we can use the spectral expanders to construct Sipser-Spielman codes with good distance properties that are equipped with extremely efficient decoding algorithms. By choosing to plant $\calC_0$ as the local codes in the spectral expander graph, we can ensure that the resulting Sipser-Spielman code inherits the good properties of both the expander graph and the repetition code.
We refer to this code $T(\calG, \calC_0)$ as the \emph{generalized repetition code} $\calC_G$ due to the fact that it is a repetition code defined on an expanding bipartite graph with local codes $\calC_0$. 
We refer to the classical repetition code that has the standard parity check matrix $\mbf{H}_S$ with linearly independent checks as the \emph{standard repetition code} and denote it by $\calC_S$. We provide a pictorial representation of these different repetition codes in Figure~\ref{fig:repetition-codes}.

\begin{figure}[htbp]
\centering

\begin{subfigure}[b]{\textwidth}
\centering
\begin{tikzpicture}[
    bit/.style={circle, draw, fill=blue!20, minimum size=7mm, font=\small},
    check/.style={rectangle, draw, fill=red!20, minimum size=6mm, font=\small},
    edge/.style={-,thick},
    scale=0.85
]
    \node[bit] (b1) at (0,0) {$b_1$};
    \node[bit] (b2) at (1.5,0) {$b_2$};
    \node[bit] (b3) at (3,0) {$b_3$};
    \node[bit] (b4) at (4.5,0) {$b_4$};
    \node[bit] (b5) at (6,0) {$b_5$};
    \node[bit] (b6) at (7.5,0) {$b_6$};
    
    \node[check] (c1) at (0.75,-1.2) {$c_1$};
    \node[check] (c2) at (2.25,-1.2) {$c_2$};
    \node[check] (c3) at (3.75,-1.2) {$c_3$};
    \node[check] (c4) at (5.25,-1.2) {$c_4$};
    \node[check] (c5) at (6.75,-1.2) {$c_5$};
    \draw[edge] (b1) -- (c1);
    \draw[edge] (b2) -- (c1);
    \draw[edge] (b2) -- (c2);
    \draw[edge] (b3) -- (c2);
    \draw[edge] (b3) -- (c3);
    \draw[edge] (b4) -- (c3);
    \draw[edge] (b4) -- (c4);
    \draw[edge] (b5) -- (c4);
    \draw[edge] (b5) -- (c5);
    \draw[edge] (b6) -- (c5);
    \node[font=\footnotesize,text=gray] at (3.25,-2) {Each check enforces: $b_i = b_{i+1}$};
\end{tikzpicture}
\caption{Standard repetition code $\calC_S$ with $n$ bits and $n-1$ checks for $n = 6$.}
\label{fig:standard-rep}
\end{subfigure}

\vspace{0.5cm}

\begin{subfigure}[b]{\textwidth}
\centering
\begin{tikzpicture}[
    bit/.style={circle, draw, fill=blue!20, minimum size=7mm, font=\small},
    check/.style={rectangle, draw, fill=red!20, minimum size=6mm, font=\small},
    edge/.style={-,thick},
    scale=0.85
]
    \node[bit] (cb1) at (90:1.3) {$b_1$};
    \node[bit] (cb2) at (30:1.3) {$b_2$};
    \node[bit] (cb3) at (-30:1.3) {$b_3$};
    \node[bit] (cb4) at (-90:1.3) {$b_4$};
    \node[bit] (cb5) at (-150:1.3) {$b_5$};
    \node[bit] (cb6) at (150:1.3) {$b_6$};
    
    \node[check] (cc1) at (60:2.3) {$c_1$};
    \node[check] (cc2) at (0:2.3) {$c_2$};
    \node[check] (cc3) at (-60:2.3) {$c_3$};
    \node[check] (cc4) at (-120:2.3) {$c_4$};
    \node[check] (cc5) at (180:2.3) {$c_5$};
    \node[check] (cc6) at (120:2.3) {$c_6$};

    \draw[edge] (cb1) -- (cc1) -- (cb2);
    \draw[edge] (cb2) -- (cc2) -- (cb3);
    \draw[edge] (cb3) -- (cc3) -- (cb4);
    \draw[edge] (cb4) -- (cc4) -- (cb5);
    \draw[edge] (cb5) -- (cc5) -- (cb6);
    \draw[edge] (cb6) -- (cc6) -- (cb1);

    \node[font=\footnotesize,text=gray] at (0,-3) {Ring topology with cyclic constraints};
\end{tikzpicture}
\caption{Circular repetition local code $\calC_0$ with $n$ bits and $n$ checks for $n = 6$.}
\label{fig:circular-rep}
\end{subfigure}

\vspace{0.5cm}

\begin{subfigure}[b]{\textwidth}
\centering
\begin{tikzpicture}[
    bit/.style={circle, draw, fill=blue!20, minimum size=7mm, font=\small},
    check/.style={rectangle, draw, fill=red!20, minimum size=6mm, font=\small},
    edge/.style={-,thick},
    scale=0.85
]
    \node[bit] (eb1) at (-3,0) {$b_1$};
    \node[bit] (eb2) at (-1.5,0) {$b_2$};
    \node[bit] (eb3) at (0,0) {$b_3$};
    \node[bit] (eb4) at (1.5,0) {$b_4$};
    \node[bit] (eb5) at (3,0) {$b_5$};
    \node[bit] (eb6) at (4.5,0) {$b_6$};

    \node[check] (ec1) at (-5.85,-4.5) {$c_1$};
    \node[check] (ec2) at (-4.65,-4.5) {$c_2$};
    \node[check] (ec3) at (-3.45,-4.5) {$c_3$};
    \node[check] (ec4) at (-2.25,-4.5) {$c_4$};
    \node[check] (ec5) at (-1.05,-4.5) {$c_5$};
    \node[check] (ec6) at (0.25,-4.5) {$c_6$};
    \node[check] (ec7) at (1.45,-4.5) {$c_7$};
    \node[check] (ec8) at (2.65,-4.5) {$c_8$};
    \node[check] (ec9) at (3.85,-4.5) {$c_9$};
    \node[check] (ec10) at (5.05,-4.5) {$c_{10}$};
    \node[check] (ec11) at (6.25,-4.5) {$c_{11}$};
    \node[check] (ec12) at (7.45,-4.5) {$c_{12}$};

    \draw[edge] (eb1) -- (ec1);
    \draw[edge] (eb4) -- (ec1);
    \draw[edge] (eb4) -- (ec2);
    \draw[edge] (eb5) -- (ec2);
    \draw[edge] (eb5) -- (ec3);
    \draw[edge] (eb1) -- (ec3);
    \draw[edge] (eb2) -- (ec4);
    \draw[edge] (eb4) -- (ec4);
    \draw[edge] (eb4) -- (ec5);
    \draw[edge] (eb6) -- (ec5);
    \draw[edge] (eb6) -- (ec6);
    \draw[edge] (eb2) -- (ec6);
    \draw[edge] (eb1) -- (ec7);
    \draw[edge] (eb2) -- (ec7);
    \draw[edge] (eb2) -- (ec8);
    \draw[edge] (eb3) -- (ec8);
    \draw[edge] (eb3) -- (ec9);
    \draw[edge] (eb1) -- (ec9);
    \draw[edge] (eb3) -- (ec10);
    \draw[edge] (eb5) -- (ec10);
    \draw[edge] (eb5) -- (ec11);
    \draw[edge] (eb6) -- (ec11);
    \draw[edge] (eb6) -- (ec12);
    \draw[edge] (eb3) -- (ec12);
    
    \node[font=\footnotesize,text=gray] at (0,-5.2) {Long-range checks that do not only act on neighboring bits};
\end{tikzpicture}
\caption{Generalized repetition code $\calC_G$ with $n$ bits and $>n$ checks for $n = 6$. $\calC_G$ is constructed from an $(n = 6, m = 4, 2, \Delta = 3)$ left-right expander $\calG = (L, R, E)$ and a $[3, 1, 3]$ repetition code $\calC_0$ as the local code. For example, $c_1, c_2, c_3$ are the parity checks decomposed from a node $u_1 \in V$ that was adjacent to nodes $b_1, b_4, b_5 \in E$. Similarly, $c_4, c_5, c_6$ are the parity checks decomposed from a node $u_2 \in V$ that was adjacent to nodes $b_2, b_4, b_6 \in E$.}
\label{fig:expander-rep}
\end{subfigure}

\caption{Comparison of three repetition code structures: (a) standard repetition code $\calC_S$, (b) circular repetition local code $\calC_0$, and (c) generalized repetition code $\calC_G$.}
\label{fig:repetition-codes}
\end{figure}
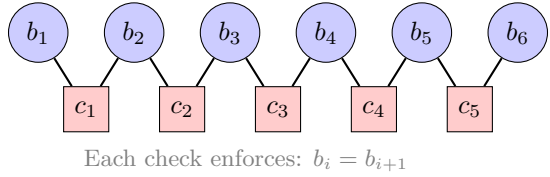
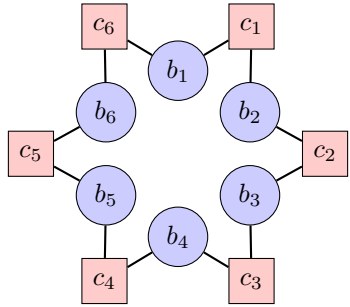
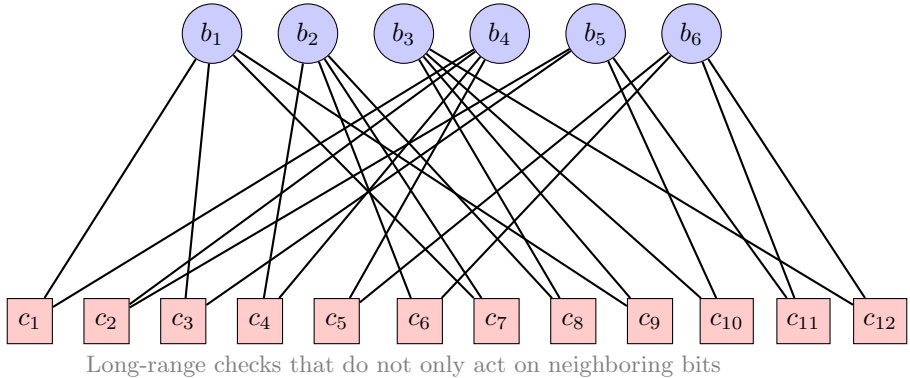

Now that we are done with describing the generalized repetition code, we proceed to describe an important 3D HGP quantum code $\calQ_G$. Consider an arbitrary 2D HGP code $\calQ$ constructed from two classical Sipser-Spielman expander codes. We can construct $\calQ_G$ by taking a tensor product of $\calQ$ with the generalized repetition code $\calC_G$. More formally, we have the following definition.

\begin{definition}[Quantum Code $\calQ_G$ ] \label{def:quantum_codes_QG_QS}
  Suppose we have a 2D HGP code $\calQ$ constructed from two classical Sipser-Spielman expander codes $\calC_1$ and $\calC_2$ that have associated 2-term chain complexes $\calA_{\calC_1}$ and $\calA_{\calC_2}$ respectively. Then, the 2D HGP code $\calQ$ has an associated 3-term chain complex $\calA_{\calQ}$ given by $\calA_{\calQ} = \calA_{\calC_1} \otimes \calA_{\calC_2}$. Define $\calA_{\calQ_G} = \calA_{\calQ} \otimes \calA_{\calC_G}$ i.e., the homological product between the chain complex of the 2D HGP code $\calQ$ and the chain complex $\calC_G$. Then, let the 3D HGP code $\calQ_{G}$ be the code associated with the chain complex $\calA_{\calQ_G}$.
\end{definition}

\subsection{Homomorphic CNOT between $\calQ$ and $\calQ_G$}
\label{sec:homomorphic_cnot}
In this section, we describe the homomorphic CNOT that we can perform between two HGP quantum codes $\calQ_G$ and $\calQ$ which we will define later. To be concrete, the homomorphic CNOT is controlled on the logical qubits of the 3D HGP code $\calQ_G$ and targeted on the logical qubits of the 2D HGP code $\calQ$. We note that our homomorphic CNOT construction is a \emph{one-way CNOT} where the control and target codes are fixed. In other words, we cannot swap the control and target codes. This is a generalization of the homomorphic CNOT construction in Ref.~\cite{heussen2025efficient} and addresses the open question of constructing such a homomorphic CNOT between codes with a growing number of logical qubits. Intuitively, reversing the control and target codes would require a chain map in the opposite direction which may not be sparse considering we would need to map a string-like $X$ logical operator in $\calQ$ to a surface-like $X$ logical operator in $\calQ_G$. Our homomorphic CNOT is a crucial component of our dimension expansion and contraction steps.

To construct the homomorphic CNOT, we first define what we mean by a chain map between two $k$-term chain complexes for arbitrary integer $k$. A chain map is a linear map that commutes with the boundary operators of the two chain complexes.

\begin{definition}[Chain Map between $k$-term Chain Complexes]
\label{def:chain_map_between_complexes}
Suppose we are given two $k$-term $\F_2$-chain complexes for $k \in \Z^+$ 
\[\calA = \left(A_k \xrightarrow{\partial^A_k} A_{k-1} \xrightarrow{\partial^A_{k-1}} \cdots \xrightarrow{\partial^A_1} A_0\right)\text{, and}\] \[\mathcal{B} = \left(B_k \xrightarrow{\partial^B_k} B_{k-1} \xrightarrow{\partial^B_{k-1}} \cdots \xrightarrow{\partial^B_1} B_0\right).\] A \emph{chain map} $\varphi: \calA \to \calB$ is a collection of linear maps $\{\varphi_i: A_i \to B_i\}_{i=0}^k$ such that for all $1 \leq i \leq k$, the following diagram commutes:
\[
\begin{quantikz}[wire types={n,n},nodes={inner sep=2pt}, mystyle]
A_i \arrow[r, "\partial^A_i"] \arrow[d, "\varphi_i"'] & A_{i-1} \arrow[d, "\varphi_{i-1}"] \\
B_i \arrow[r, "\partial^B_i"'] & B_{i-1}
\end{quantikz}
\]
That is, $\varphi_{i-1} \circ \partial^A_i = \partial^B_i \circ \varphi_i$ for all $1 \leq i \leq k$.
\end{definition}

Given $k$-term chain complexes, we are interested in the existence of a chain map between the two chain complexes where each of the linear maps $\varphi_i$ is $w$-limited for some integer $w \in \Z^+$. A $w$-limited linear map is a sparse linear map where its matrix representation has at most $w$ non-zero entries in each column and row. Next, we state the homomorphic CNOT framework defined in Refs.~\cite{huang2023homomorphic,xu2025fast}.

\begin{definition}
[Homomorphic CNOT~{\cite[Restatement of Definition 3]{xu2025fast}}]
\label{def:homomorphic_CNOT}
 Let $\calQ$ and $\calQ^{\prime}$ be two quantum CSS codes associated with two 3-term chain complexes $\calA = \{\{A_i\}_{i = 0}^2, \{\partial_i\}_{i = 1}^2\}$, and $\calA' = \{\{A^{\prime}_i\}_{i = 0}^2, \{\partial^{\prime}_i\}_{i = 1}^2\}$, respectively. Let $\varphi = \{\varphi_i: A_i^{\prime} \rightarrow A_i\}_{i = 0}^2$ be a homomorphism between the two chain complexes, i.e., the following diagram is commutative:
    \begin{equation}
    \label{eq:quantum_code_homo}
        \begin{quantikz}[wire types={n,n},nodes={inner sep=2pt}, mystyle]
	{A_2} & {A_1} & {A_0} \\
	{A_2^{\prime}} & {A_1^{\prime}} & {A_0^{\prime}}
	\arrow["{\partial_2}", from=1-1, to=1-2]
	\arrow["{\partial_1}", from=1-2, to=1-3]
        \arrow["{\partial_2^{\prime}}", from=2-1, to=2-2]
	\arrow["{\partial_1^{\prime}}", from=2-2, to=2-3]
        \arrow["{\varphi_2}", from=2-1, to=1-1]
	\arrow["{\varphi_1}", from=2-2, to=1-2]
        \arrow["{\varphi_0}", from=2-3, to=1-3]
\end{quantikz}
    \end{equation}
     Then physical $\calQ$-controlled CNOTs specified by $\varphi_1$, i.e. a physical CNOT controlled by the $i$-th qubit of $\calQ$ and targeted the $j$-th qubit of $\calQ^{\prime}$ is applied if and only if $\varphi_1[i,j] = 1$, give some $\calQ$-controlled logical CNOT gates between $\calQ$ and $\calQ^{\prime}$. We refer to such a logical gadget as a homomorphic CNOT associated with the homomorphism $\varphi$.
\end{definition}

With Definition~\ref{def:quantum_codes_QG_QS}, we can now describe the homomorphic CNOT from the quantum code $\calQ_G$ to $\calQ$ by applying the framework stated in Definition~\ref{def:homomorphic_CNOT} to the chain complexes $\calA_{\calQ_G}$ and $\calA_{\calQ}$ defined in Definition~\ref{def:quantum_codes_QG_QS} with a $1$-limited chain map $\varphi$ that we define in the following proposition. 

\begin{proposition}[Homomorphic CNOT from $\calQ_G$ to $\calQ_S$]\label{prop:homomorphic_cnot_from_QG_to_QS}
  Let $\calQ$ and $\calQ_G$ be the quantum codes associated with the chain complexes $\calA_{\calQ}$ and $\calA_{\calQ_G}$ defined in Definition~\ref{def:quantum_codes_QG_QS} respectively. Then, there exists a homomorphic CNOT from $\calQ_G$ to $\calQ$ specified by some $1$-limited chain map $\varphi$. In particular, the homomorphic CNOT uses a transversal implementation of inter-block CNOTs and induces a logical CNOT from each logical qubit of $\calQ_G$ to the corresponding logical qubit of $\calQ$.
\end{proposition}

\begin{proof}
  To prove the above proposition, we need to construct a chain map from the chain complex $\calA_{\calQ}$ to the chain complex $\calA_{\calQ_G}$ that is $1$-limited. We can then use Definition~\ref{def:homomorphic_CNOT} to conclude that the physical CNOTs specified by the chain map give a homomorphic CNOT from $\calQ_G$ to $\calQ$. 
  
  We can decompose the 4-term chain complex $\calA_{\calQ_G}$ into the following form:
  \begin{align*}
    A_{\calQ_G, 0} &= A_{\calC_1, 0} \otimes A_{\calC_2, 0} \otimes G_0,\\
    A_{\calQ_G, 1} &= \bigoplus_{i, j, k \in \{0, 1\}\,:\, i + j + k = 1} A_{\calC_1, i} \otimes A_{\calC_2, j} \otimes G_k,\\
    A_{\calQ_G, 2} &= \bigoplus_{i, j, k \in \{0, 1\}\,:\, i + j + k = 2} A_{\calC_1, i} \otimes A_{\calC_2, j} \otimes G_k \\
    A_{\calQ_G, 3} &= A_{\calC_1, 1} \otimes A_{\calC_2, 1} \otimes G_1,
  \end{align*}
  where $G_k$ for $k \in \{0, 1\}$ are the vector spaces in the chain complex $\calA_{\calC_G}$ that is associated with the generalized repetition code $T(\calG, \calC_0)$.

  Now, we can construct the chain map $\varphi: \calA_{\calQ} \to \calA_{\calQ_G}$. Let $g_{0}$ be an arbitrary standard basis vector in $G_0$. Then, we can define the chain map as follows:
  \begin{align*}
  \varphi_0\,:\, A_{\calQ, 0} &\to A_{\calQ_G, 0},\\
  a_0 &\mapsto a_0 \otimes g_{0},\\
  \varphi_1\,:\, A_{\calQ, 1} &\to A_{\calQ_G, 1},\\
  a_1 &\mapsto a_1 \otimes g_{0},\\
  \varphi_2\,:\, A_{\calQ, 2} &\to A_{\calQ_G, 2},\\
  a_2 &\mapsto a_2 \otimes g_{0},
  \end{align*}
  where $a_i$ for $i \in \{0, 1, 2\}$ are arbitrary vectors in the vector spaces $A_{\calQ, i}$ that belong to the chain complex $\calA_{\calQ}$. It is clear from the construction of the chain map $\varphi$ that each of the linear maps $\varphi_i$ for $i \in \{0, 1, 2\}$ is $1$-limited since each basis vector in $A_{\calQ, i}$ is mapped to a unique basis vector in $A_{\calQ_G, i}$. This property directly implies that the physical CNOTs specified by $\varphi^{\calO(1)}_{\calQ,1}$ act transversally: each physical qubit in $\calQ$ is paired with exactly one unique physical qubit in $\calQ_G$ for the CNOT operation, with no overlap or entanglement between different pairs. It is straightforward to verify that the above chain map $\varphi$ satisfies the commutative diagram in Equation~\ref{eq:quantum_code_homo} in Definition~\ref{def:homomorphic_CNOT}.

  Similarly, it is easy to check that the application of the physical CNOTs leads to logical CNOTs between the corresponding logical qubits in $\calQ_G$ and $\calQ$. The $Z$ logical operators are strings in $\calQ$ that are mapped to strings in $\calQ_G$ in the layer associated with the basis vector $g_{0}$; the $X$ logical operators are membranes in $\calQ_G$ whose string-like boundaries that reside on the layer associated with the basis vector $g_{0}$ are mapped to strings in $\calQ$ under the chain map $\varphi$. Thus, we conclude that the homomorphic CNOT specified by $\varphi$ is transversal at the physical level and induces logical CNOTs between individual pairs of logical qubits in $\calQ_G$ and $\calQ$, as claimed.

\end{proof}

\begin{remark}
  We remark that our homomorphic CNOT gadget is compatible with the Grid Pauli Product Measurement framework formulated in Ref.~\cite{xu2025fast}. This stems from the well-known fact that chain maps compose to form another chain map. In particular, if we have two chain maps $\varphi: \calA \to \calB$ and $\psi: \calB \to \calC$ between three $k$-term chain complexes $\calA, \calB,$ and $\calC$, then the composition of the two chain maps $\psi \circ \varphi: \calA \to \calC$ is also a chain map. This property directly implies that if we have a homomorphic CNOT from a quantum code $\calQ_1$ to another quantum code $\calQ_2$ specified by a chain map $\varphi$, and another homomorphic CNOT from $\calQ_2$ to a third quantum code $\calQ_3$ specified by a chain map $\psi$, then we can construct a homomorphic CNOT from $\calQ_1$ to $\calQ_3$ specified by the composition of the two chain maps $\psi \circ \varphi$. In other words, we can puncture or/and augment the classical codes that constitute the HGP codes and only perform a logical CNOT on select logical qubits between the 2D and 3D HGP codes. This flexibility can potentially be useful in practice for addressable code-switching applications.
\end{remark}

\begin{remark} \label{remark:constant-rate-homomorphic-cnot}
  We also remark that the homomorphic CNOT that we have developed can be generalized to the case where the third dimension of the code block $\calQ_G$ is constructed from a different classical code than the generalized repetition code $\calC_G$. Instead of performing a single transversal homomorphic CNOT between $\calQ_G$ and $\calQ$, we now transversal homomorphic CNOTs from $\calQ_G$ to multiple copies of $\calQ$ where the number of copies scales with the dimension of the classical code used to construct the third dimension of $\calQ_G$. By expressing the parity-check matrix of the classical code in some canonical basis where the information bits are separated from the parity bits, the logical strings in $\calQ_G$ can be mapped to disjoint layers in $\calQ_G$. Likewise the boundaries of the logical membranes in $\calQ_G$ can also be mapped to disjoint layers in $\calQ_G$. We can then perform a transversal homomorphic CNOT from each of these layers to a different copy of $\calQ$ in parallel without blowing up the depth of the circuit. This generalization can be useful in practice if we want to use a classical code with a better rate than the generalized repetition code to construct the third dimension of $\calQ_G$. 
\end{remark}
\section{Single-Shot Code Switching}
\label{sec:single_shot_code_switching}

In this section, we describe the code-switching protocol that allows us to switch between different high-rate QLDPC codes in a single-shot manner. Our code-switching involves two kinds of operations: dimensional expansion and dimensional contraction that take place via logical teleportation. 
We now state our single-shot code-switching theorem.

\begin{theorem}[Single-Shot Code-switching]
\label{thm:SSCS}
Given $D \geq 2$, suppose we are given a $D$-dimensional HGP code $\calQ$ with parameters $\llbracket n, \Theta(n), \Omega(n^{1/D})\rrbracket$. Let $\calQ_G$ be an $\llbracket \Theta(n^{(D+1)/D}), \Theta(n), \Omega(n^{1/D})\rrbracket$ $(D+1)$-dimensional HGP code that has a chain complex $\calA_{\calQ_G}$ that is obtained from the tensor product between the chain complex $\calA_{\calQ}$ that corresponds to the HGP code $\calQ$ and a 2-term chain complex $\calA_{\calC_{G}}$ that corresponds to a classical Tanner code $\calC_{G} = T(\calG, \calC_0)$ where $\calG$ is a $\Delta$-regular spectral expander with $\Omega(n^{1/D})$ edges and $\calC_0$ is a $[r, 1, r]$ classical repetition local code for some constants $\Delta > 0$. Then, there exists a single-shot code-switching protocol that fault-tolerantly switches between $\calQ$ and $\calQ_G$ using a constant-depth circuit that can tolerate $\Omega(n^{1/D})$ adversarial errors. In addition, the single-shot code-switching protocol exhibits a threshold against a constant physical error rate under the local stochastic noise model. 
\end{theorem}

\begin{proof}
We prove Theorem~\ref{thm:SSCS} by constructing an explicit single-shot code-switching protocol and provide an intuitive explanation for its fault-tolerance. We defer a more rigorous proof of fault-tolerance under both the adversarial and local stochastic noise models to Sections~\ref{sec:adversarial_noise} and~\ref{sec:local-stochastic} respectively. We now state the dimensional expansion scheme and break it down into several constant-depth circuit components that utilizes known gadgets.

\subsection{Dimensional Expansion}
\label{sec:dimensional_expansion}
In this section, we describe the dimensional expansion protocol for the code $\calQ$ to the code $\calQ_G$.
We utilize the homomorphic CNOT primitive developed in Section~\ref{sec:homomorphic_cnot} as well as a single-shot state preparation scheme that is discussed in greater detail in Section~\ref{sec:single_shot_state_preparation} to perform the dimensional expansion. Our dimensional expansion scheme is effectively a logical teleportation scheme that allows us to transfer the logical state of the code $\calQ$ into the code $\calQ_G$. This addresses the open question raised in Ref.~\cite{heussen2025efficient} of how to generalize the one-way CNOT gate for color codes to high-rate codes for the purpose of code-switching via logical teleportation. The effective logical circuit for the dimensional expansion is shown in Figure~\ref{fig:dimensional_expansion_circuit}.

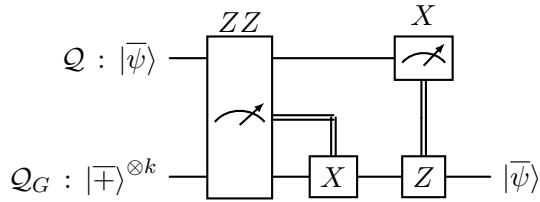
\begin{figure}[H]
	\tikzset{
noisy/.style={starburst,fill=yellow,draw=red,line
width=1pt}
}
	\centering
\begin{quantikz}
  \lstick{$\calQ\,:\,\ket{\overline{\psi}}$} & \meter[3, label style={inner sep=1pt}]{ZZ} & & \meter{X}\wire[d]{c}  \\
  \setwiretype{n} &\wire[r]{c} & & \\  
  \lstick{$\calQ_G\,:\,\ket{\overline{+}}^{\otimes k}$} & \wire[u]{c} & \gate{X}\wire[u]{c} & \gate{Z}\wire[u]{c} & \rstick{$\ket{\overline{\psi}}$}
\end{quantikz}
					\caption{Logical quantum circuit for performing logical teleportation between two code blocks $\calQ$ and $\calQ_G$ using Pauli-based measurements and gates.~\label{fig:dimensional_expansion_circuit}}
\end{figure}

To perform the $X$ measurement shown in Figure~\ref{fig:dimensional_expansion_circuit} in a single-shot and fault-tolerant manner, we can simply employ the Steane measurement technique and measure all of the physical qubits in $\calQ$ in the $X$ basis. As for the $ZZ$ measurement, we can use the Grid Pauli Product Measurement (GPPM) gadget developed in Ref.~\cite{xu2025fast} which gives us the following circuit shown in Figure~\ref{fig:dimensional_expansion_circuit_2}.

This teleportation scheme is very similar to the standard EPR teleportation scheme. From the onset, there are a few technical challenges involved in this dimensional expansion gadget that we have devised here. Typically, we teleport logical qubits between codes of the same ``shape'' or ``size'' which allows us to use standard transversal entangling gadgets. However, when we are teleporting logical qubits from a 2D HGP code to a 3D HGP code, we cannot directly use the standard transversal CNOT gate because the two codes have different ``shapes''. While it is possible to employ the homomorphic CNOT gadget, a naive implementation of the homomorphic CNOT gadget would likely require a physical CNOT circuit with CNOT depth that grows with the length of the third dimension of the 3D HGP code. The intuition behind this is that the logical $X$ operators of the 3D HGP code are membrane-like while the logical $X$ operators of the 2D HGP code are string-like. In order for the physical CNOT gadget to perform the intended logical CNOT action, we have to connect the physical qubits that lie in the support of these string-like operators in the 2D code to the physical qubits that lie in the support of these membrane-like operators in the 3D code.

To sidestep this technical issue, we utilize the transversal homomorphic CNOT gadget that we designed in Section~\ref{sec:homomorphic_cnot}. In other words, we first entangle the logical qubits in the 3D HGP code $\calQ_G$ with the logical qubits in some 2D HGP ancilla code block $\calQ$ that is initialized in the $\ket{\overline{0}}^{\otimes k}$ state using the transversal homomorphic CNOT gadget before entangling 2D ancilla code block with the original 2D HGP data code block $\calQ$ using the standard CSS transversal CNOT gate gadget. This avoids the need to perform a high-depth CNOT circuit, allowing us to preserve the transversal property and constant-time overhead for the procedure.

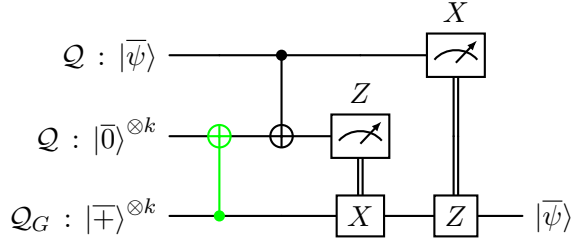
\begin{figure}[htbp]
	\tikzset{
noisy/.style={starburst,fill=yellow,draw=red,line
width=1pt}
}
	\centering
 \begin{quantikz}
  \lstick{$\calQ\,:\,\ket{\overline{\psi}}$}& & \ctrl{1} & &\meter{X}\wire[d]{c} \\
 \lstick{$\calQ\,:\,\ket{\overline{0}}^{\otimes k}$} & \targ[style={green}]{} &  \targ{} &\meter{Z}  & \setwiretype{n}\\
  \lstick{$\calQ_G\,:\,\ket{\overline{+}}^{\otimes k}$}& \ctrl[style={green}]{-1} & &\gate{X}\wire[u]{c} & \gate{Z}\wire[u]{c} & \rstick{$\ket{\overline{\psi}}$}
 \end{quantikz}
					\caption{Logical quantum circuit for performing logical teleportation from a code block $\calQ$ to another code block $\calQ_G$ using an adapter ancilla code block $\calQ$. The sequence of gates including the two CNOTs as well as the $Z$ measurement is effectively the $ZZ$ measurement. The first CNOT (colored in green) is the homomorphic CNOT described in Section~\ref{sec:homomorphic_cnot}. The last CNOT is the standard transversal logical CNOT between two 2D HGP codes.\label{fig:dimensional_expansion_circuit_2}}
\end{figure}


The only remaining thing to address is the state preparation of the ancilla code block $\calQ$ and the 3D HGP code block $\calQ_G$. To prepare the 3D HGP code block $\calQ_G$ in the $\ket{\overline{+}}^{\otimes k}$ state, we can simply perform the standard CSS state preparation protocol and then utilize the redundant $Z$ checks to ensure that the state preparation is single-shot and fault-tolerant. This is essentially the same as the state preparation protocol discussed in Ref.~\cite{hong2024single} except we do not collapse the third dimension of the code block. As for the state preparation of the 2D HGP ancilla code block $\calQ$ in the $\ket{\overline{0}}^{\otimes k}$ state, we can utilize the single-shot state preparation protocol described in Bergamaschi and Liu's recent work~\cite{bergamaschi2024fault}.

The rest of the gadgets shown in Figure~\ref{fig:dimensional_expansion_circuit_2} are just the standard transversal logical CNOT gadget and Pauli logical operations that can be implemented in the standard manner. The $X$ and $Z$ measurements can be done using the Steane measurement technique by measuring all physical qubits in the $X$ and $Z$ basis respectively. The classical post-processing of the measurement outcomes can be done using standard classical computation.

\subsection{Dimensional Contraction}
\label{sec:dimensional_contraction}
In this section, we discuss the process of collapsing the 3D HGP code $\calQ_G$ back to a 2D HGP code $\calQ$. Because our dimensional expansion scheme is essentially a logical teleportation scheme, we can leverage the exact same techniques used in the expansion process to achieve contraction. For the sake of being concrete, we provide a quantum circuit for the dimensional contraction scheme in the following figure shown in Figure~\ref{fig:dimensional_contraction_circuit}.

\begin{figure}[H]
	\centering
 \begin{quantikz}
  \lstick{$\calQ\,:\,\ket{\overline{+}}^{\otimes k}$}& & \ctrl{1} & \gate{X} & \gate{Z}\wire[d]{c} & \rstick{$\ket{\overline{\psi}}$} \\
 \lstick{$\calQ\,:\,\ket{\overline{0}}^{\otimes k}$} & \targ[style={green}]{} & \targ{} &\meter[label style={yshift=-1.2cm}]{Z}\wire[u]{c} &\setwiretype{n} &\\
  \lstick{$\calQ_G = \tilde{\calQ}\,:\,\ket{\overline{\psi}}$}& \ctrl[style={green}]{-1} & & & \meter[label style={yshift=-1.2cm}]{X}\wire[u]{c} 
 \end{quantikz}
					\caption{Logical quantum circuit for performing logical teleportation from a code block $\calQ_G$ to another code block $\calQ$ using an adapter ancilla code block $\calQ$. The sequence of gates including the two CNOTs as well as the $Z$ measurement is effectively the $ZZ$ measurement. The first CNOT (colored in green) is the homomorphic CNOT described in Section~\ref{sec:homomorphic_cnot}. The last CNOT is the standard transversal logical CNOT between two 2D HGP codes.\label{fig:dimensional_contraction_circuit}}
\end{figure}
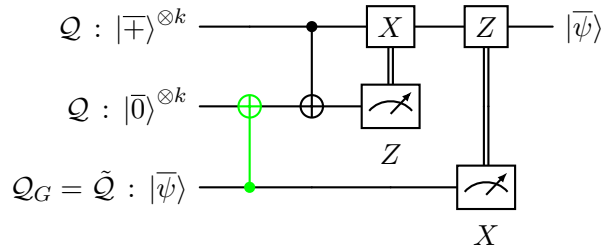

For the fault-tolerance of the code-switching protocol, there are a few different possible error locations. The state preparation step is fault-tolerant by the analysis done in Refs.~\cite{hong2024single} and \cite{bergamaschi2024fault}. 
The transversal CNOT gates including the homomorphic CNOT gate are clearly fault-tolerant. The transversal $X$ and $Z$ measurements that may be done using the Steane measurement technique are also fault-tolerant. Thus, the entire code-switching protocol is fault-tolerant. We defer a more rigorous proof of fault-tolerance under both the adversarial and local stochastic noise models to Sections~\ref{sec:adversarial_noise} and~\ref{sec:local-stochastic} respectively. 

\end{proof}

\begin{remark}
The single-shot code-switching protocol described here can be understood as a generalization of the single-shot lattice surgery protocol to high-rate codes. Suppose we are interested in performing single-shot lattice surgery to merge a 2D HGP code $\calQ$ with another 3D HGP code $\calQ_G$. That can be achieved via teleporting the logical qubits in $\calQ$ to some merged 3D HGP code $\tilde{\calQ}_G$ that is obtained from merging $\calQ_G$ with another copy of $\calQ$. The logical teleportation can be done using the exact same circuit as shown in Figure~\ref{fig:dimensional_expansion_circuit_2}. The splitting step can also be understood as a logical teleportation from the merged 3D HGP code $\tilde{\calQ}_G$ to the original 2D HGP code $\calQ$. This single-shot lattice surgery can be done between higher-dimensional HGP codes as long as the two codes have soundness properties in the appropriate dimensions and basis. \SAM{Polish up this remark}
\end{remark}
\section{Dynamical Fault-Tolerant Universal Computation}
\label{sec:dynamical}

In this section, we present a dynamical fault-tolerant universal computation scheme based on recent HGP constructions of QLDPC codes with transversal CCZ gates~\cite{golowich2025quantum, lin2024transversal, breuckmann2024cups,zhu2025topological, zhu2025transversal}.
Our dynamical scheme involves a code-switching protocol that involves dimensional jumps between two different high-rate QLDPC codes to sidestep the Eastin-Knill theorem~\cite{eastin2009restrictions} and achieve fault-tolerant universal computation.
We describe our scheme in detail for the case of hypergraph product codes, but it can very likely be generalized to other product codes that include lifted product codes and balanced product codes.
We break down our scheme into two components: the first component is a dimensional expansion of an $D$-dimensional HGP code to a $(D+1)$-dimensional HGP code, and the second component is a dimensional contraction of an $(D+1)$-dimensional HGP code to a $D$-dimensional HGP code.
While our scheme holds for any integer $D \geq 2$, we will focus on the case of $D=2$ for simplicity as well as to highlight the change between the Clifford and non-Clifford logical gates in the two codes.

Before we describe the explicit procedures for the dimensional expansion and contraction, we first describe the general framework of our scheme. For the purpose of all Clifford gates and error correction, our code will be a 2D HGP code $\calQ$ with parameters $\llbracket n, k, d_X, d_Z \rrbracket$ constructed using classical codes with the Sipser-Spielman construction. Our construction will be very similar to the one stated by Golowich and Lin in Ref.~\cite{golowich2025quantum} and Zhu in Refs.~\cite{zhu2025topological, zhu2025transversal}. By Lemma~\ref{lem:linear-confinement-expander-codes}, $\calQ$ is amenable to single-shot error correction. With a combination of the fold-transversal logical gates~\cite{breuckmann2024fold, quintavalle2023partitioning} and the parallel Pauli-measurement-based logical gates from Ref.~\cite{xu2025fast}, we can perform addressable Clifford gates in a fault-tolerant manner.
When we want to perform logical non-Clifford gates, we perform a dimensional expansion of the code $\calQ$ by teleporting the logical qubits of $\calQ$ into a 3D HGP code $\calQ_G$ with parameters $\llbracket n_G, k_G, d_{X,G} \cdot n_C, d_{Z,G} \rrbracket$ for integer $n_C = \min(d_X, d_Z)$. The code $\calQ_G$ has the same 2D form as $\calQ$ but its third dimension is given by the classical code $\calC_G$ i.e. the generalized repetition code. Note that $\calQ_G$ is effectively an instance of the 3D HGP code described in Section~\ref{sec:homomorphic_cnot}. We show that a careful choice for the classical codes that are used to construct $\calQ_G$ can make it amenable to inter-block transversal logical CCZ gates. We note that the dimensional expansion step is a single-shot procedure that does not require any additional rounds of syndrome measurements. After we finish performing the logical non-Clifford gates, we perform a dimensional contraction of the code $\calQ_G$ back to $\calQ$ to keep the ``non-constant'' spatial overhead limited to the time steps where we perform the logical non-Clifford gates. The dimensional contraction step is also a single-shot procedure that does not require any additional rounds of syndrome measurements. Lastly, when we want to perform state preparation, we either use the single-shot state preparation scheme developed by Bergamaschi and Liu in Ref.~\cite{bergamaschi2024fault} to prepare logical states directly in the 2D HGP code $\calQ$, or the scheme developed by Hong in Ref.~\cite{hong2024single} to prepare logical states in the 3D HGP code $\calQ_G$ depending on the dimensionality of the code block in which we hope to prepare the logical state.
With this framework, we can perform a fault-tolerant universal computation scheme with single-shot error correction and state preparation with constant spatial overhead except for the time steps where we perform the logical non-Clifford gates and the state preparation. Most importantly, every gadget in our scheme is single-shot and takes constant time to implement. We present a schematic overview of our scheme in Figure~\ref{fig:dimensional_expansion_contraction_overview}. We now state our main theorem for this section.

\begin{theorem}[Single-Shot, Universal Protocol via Code-switching]
\label{thm:universal_FTQC}
Let $\calQ$ be an $\llbracket n, \Theta(n), \Omega(\sqrt{n})\rrbracket$ 2D HGP code constructed from classical Sipser-Spielman codes. 
Let $\calQ_G$ be a 3D HGP code with parameters $\llbracket \Theta(n^{3/2}), \Theta(n), \Omega(\sqrt{n})\rrbracket$ that has a chain complex $\calA_{\calQ_G}$ that is obtained from the tensor product between the chain complex $\calA_{\calQ}$ that corresponds to the HGP code $\calQ$ and a 2-term chain complex $\calA_{\calC_{G}}$ that corresponds to a classical Tanner code $\calC_{G} = T(\calG, \calC_0)$ where $\calG$ is a $\Delta$-regular spectral expander and $\calC_0$ is a $[\Delta, 1, \Delta]$ classical repetition local code for some constants $\Delta, \gamma, \alpha > 0$. Then, there exists a universal fault-tolerant quantum computation protocol implements a universal gate set using only gadgets that satisfy single-shot universality. The protocol can tolerate $\Omega(\sqrt{n})$ adversarial faults and also exhibits a threshold under local-stochastic noise. \SAM{Add local-stochastic noise model if we can show that it works.}
\end{theorem}

\begin{proof}
  We break down the proof into several subsections that addresses each individual gadget. Similar to the proof of Theorem~\ref{thm:SSCS}, we provide some intuition for the fault-tolerance of the protocol and defer a more rigorous analysis of the fault-tolerance to Section~\ref{sec:adversarial_noise}. 

\subsection{Single-Shot State Preparation}
\label{sec:single_shot_state_preparation}
To initialize logical states $\ket{\overline{0}}$ and $\ket{\overline{+}}$ in the 2D HGP code $\calQ$, we simply employ the single-shot state preparation scheme developed by Bergamaschi and Liu in Ref.~\cite{bergamaschi2024fault}. If we are interested in preparing logical states in the 3D HGP code $\calQ_G$, we can either use the same scheme by Bergamaschi and Liu or the state preparation scheme developed by Hong in Ref.~\cite{hong2024single} depending on the basis in which the 3D HGP code $\calQ_G$ is sound. If the 3D HGP code $\calQ_G$ is sound in the $Z$ ($X$) basis, we can use Hong's scheme to prepare $\ket{\overline{+}}$ ($\ket{\overline{0}}$) directly in $\calQ_G$. Otherwise, we can utilize the other single-shot state preparation scheme by Bergamaschi and Liu~\cite{bergamaschi2024fault} to prepare the desired logical state in $\calQ_G$.



\subsection{Single-Shot Error Correction}
\label{sec:single_shot_error_correction}
HGP codes built from classical expander codes are known to support single-shot error correction. Specifically, 2D HGP codes built from random lossless expanders admit single-shot correction under both adversarial and local-stochastic noise models \cite{fawzi2018efficient, campbell2019theory, quintavalle2021single, Hong_2025_thermal} with high probability. We now state a theorem regarding the single-shot correction capability for the 2D HGP codes that we use for our protocol.

\begin{theorem}[Linear confinement of 2D HGP codes from Sipser-Spielman codes~{\cite[Restatement of result]{dinur2023good}}]
\label{thm:linear-confinement-sipser-spielman-2D-HGP}
Let $\calQ$ be a distance $d$ 2D HGP code constructed from two classical Sipser-Spielman codes built from spectral expanders and random local codes. Then, with high probability, $\calQ$ has $(O(d), f)$-confinement such that $f(x) = \Theta(x)$ in both Pauli bases.
\end{theorem}
\begin{proof}
    The proof follows directly from the results in Ref.~\cite{dinur2023good}. The main idea is that the product of two random local codes will satisfy the product expansion property with high probability. Satisfying the product expansion property on good spectral expanders gives good boundary expansion, which in turn implies small-set coboundary expansion and linear confinement in both Pauli bases for the 2D HGP code.
\end{proof}

This shows that the 2D HGP codes built from the Sipser-Spielman codes discussed above exhibit linear confinement with high probability, which is a sufficient condition for single-shot correction under adversarial noise and local-stochastic noise models. 
As for higher dimensional HGP codes constructed with a base 2D HGP code built from such Sipser-Spielman codes, it is still very much an open question. Campbell, Quintavalle, and co-authors \cite{quintavalle2021single} showed that higher-dimensional HGP codes from random lossless expanders achieve single-shot correction under the adversarial model in at least one Pauli basis. 

For our 3D HGP codes, although single-shot correction may be available only in one Pauli basis under adversarial noise, we can postpone correction in the other basis until we return to a 2D HGP where both bases support single-shot procedures. This is adequate for our fault-tolerance analysis because all our gadgets run in constant depth, so error growth can be bounded. Moreover, we occupy the 3D code only briefly—to execute the transversal CCZ—before switching back to 2D, so additional errors in the non-corrected basis are not expected to accumulate significantly.

\subsection{Single-Shot Logical Measurement}
\label{sec:single_shot_logical_measurement}
Logical Pauli measurements on HGP codes can be executed in one shot, i.e., with a single round of physical readout rather than repeated syndrome cycles. Using the grid Pauli product measurement (GPPM) framework of Xu et al.~\cite{xu2025fast}, we first target the desired logical by coupling the 2D data HGP code block to an ancilla 2D HGP code block that is punctured and/or augmented to expose the right operators. The ancilla code block can be prepared with the Bergamaschi and Liu scheme. This coupling is implemented via homomorphic CNOTs, which map the logical Pauli on the data into a Pauli on the ancilla. We then perform transversal single-qubit measurements of the ancilla in a Steane-style scheme. Finally, a lightweight classical postprocessing step—accounting for ancilla stabilizers and known byproducts—combines the raw outcomes to infer the logical eigenvalue fault-tolerantly. The procedure is constant-depth, avoids multiple measurement rounds, and isolates the targeted logical without disturbing others, making it compatible with standard HGP decoding pipelines. If we are planning to measure all logical qubits, we can simply perform transversal single-qubit measurements of the data HGP code block in the desired basis and perform classical postprocessing to infer the logical eigenvalues fault-tolerantly.

\subsection{Logical Clifford Gates}
\label{sec:single_shot_logical_clifford_gates}
2D HGP codes admit several complementary strategies for realizing logical Clifford gates, many of which achieve constant depth. One route is the fold-transversal construction of Breuckmann and Burton~\cite{breuckmann2024fold}, which—together with recent partitioning techniques~\cite{quintavalle2023partitioning}—yields depth-bounded Clifford implementations. A second option is code automorphisms: permutations of qubits and checks that act as logical Cliffords~\cite{berthusen2025automorphism}. A third, measurement-driven route is Pauli-based computation~\cite{bravyi2016trading}. Here we employ the GPPM framework of Xu et al.~\cite{xu2025fast}: the data HGP block is entangled with a punctured and/or augmented HGP ancilla via homomorphic CNOTs, enabling targeted logical Pauli measurements through transversal readout of the ancilla. This allows us to perform addressable logical Clifford gates in a constant-depth circuit when coupled with fold-transversal gates and logical automorphisms. In the case when certain Clifford gates are not achievable in the 2D HGP code $\calQ$, we can temporarily switch to the 3D HGP code $\calQ_G$ to use the logical CCZ gate to descend to the CZ gate and then obtain the logical Hadamard which is sufficient to achieve universal quantum computation along with the logical CCZ gate as shown in Section~\ref{sec:static}. Alternatively, we can use the homomorphic CNOTs discussed in Remark~\ref{remark:constant-rate-homomorphic-cnot} to teleport the logical qubits from our 2D HGP code $\calQ$ to another 2D HGP code that has the necessary symmetries to implement the desired logical Clifford gates. For example, we can use the transversal homomorphic CNOT gadget to teleport the logical qubits from our 2D HGP code $\calQ = \calC_1 \times \calC_2$ to a 3D HGP code $\calC_1 \times \calC_2 \times \calC_2$ before teleporting the logical qubits again with another round of homomorphic CNOT to a 2D HGP code $\calC_2 \times \calC_2$ that has the symmetries to implement the desired logical Hadamard gate. Subsequently, we can teleport the logical qubits back to our original 2D HGP code $\calQ$ with two more rounds of homomorphic CNOTs. This addressable Clifford approach is explored in greater detail in Refs.~\cite{xu2025batched} and \cite{golowich2025constant}.

\subsection{Logical CCZ Gates}
\label{sec:single_shot_logical_ccz_gates}
In this subsection, we describe how to choose the right spectral expanders and local codes to construct 3D HGP codes that allow us to perform logical CCZ gates across them. The construction relies on Theorem~\ref{thm:local-multiplication-ccz-property}, which imposes the local multiplication property on the corresponding local codes of three blocks of 3D HGP codes. We shall use the Tanner code notation of $T(\calG, \calC)$ to denote the Tanner code constructed from a spectral expander $\calG$ and a classical local code $\calC$. Recall that we denote the classical repetition code with a redundant check by $\calC_0$.

\begin{definition}\label{def:simple_construction_3D_HGP_CCZ}
Let $\calQ_1^{3D}$, $\calQ_2^{3D}$, and $\calQ_3^{3D}$ be three 3D HGP codes constructed as follows:
\begin{align*}
  \calQ_1^{3D} &= T(\calG_1, \calC_1) \otimes T(\calG_2, \calC_2^\perp) \otimes T(\calG_3, \calC_0), \\
  \calQ_2^{3D} &= T(\calG_1, \calC_0) \otimes T(\calG_2, \calC_2) \otimes T(\calG_3, \calC_3^\perp), \\
  \calQ_3^{3D} &= T(\calG_1, \calC_1^\perp) \otimes T(\calG_2, \calC_0) \otimes T(\calG_3, \calC_3),
\end{align*}
where $\calC_i^\perp$ denotes the dual code of $\calC_i$ for $i=1,2,3$. 
\end{definition}
Note that the graphs $\calG_1$, $\calG_2$, and $\calG_3$ and the classical codes $\calC_1$, $\calC_2$, and $\calC_3$ need not be the same.
In addition, each $T(\calG_i, \cdot)$ for $i=1,2,3$ can be understood as the skeleton for each of the three geometric dimensions. In other words, we can map $i = 1, 2, 3$ to the $X, Y, Z$ directions, respectively. Then, $\calQ_1^{3D}$ has $T(\calG_1, \calC_1)$, $T(\calG_2, \calC_2^\perp)$, and $T(\calG_3, \calC_0)$ oriented in the $X, Y$, and $Z$ directions respectively. All three codes $\calQ_1^{3D}$, $\calQ_2^{3D}$, and $\calQ_3^{3D}$ are different instances of the 3D HGP code $\calQ_G$. Since dimensional expansion grows a 2D HGP code into a 3D HGP code with $\calC_0$ in the new dimension, we can imagine the string-like $X$ logical operator of the 2D HGP code in the $X-Y$ plane growing into a membrane-like $X$ logical operator in the $Z$ direction. The same interpretation can be applied to $\calQ_2^{3D}$ and $\calQ_3^{3D}$ where their $X$ logical operators grow from the $Y-Z$ and $X-Z$ planes in the $X$ and $Y$ directions respectively. One can intuitively see how the three membrane-like $X$ logical operators of $\calQ_1^{3D}$, $\calQ_2^{3D}$, and $\calQ_3^{3D}$ can intersect at a single point, which is the key requirement for implementing a transversal CCZ gate. We provide a diagram to illustrate this idea in Figure~\ref{fig:ccz_construction}.

\begin{figure}[htbp]
\centering

\begin{subfigure}[b]{0.48\textwidth}
\centering
\begin{tikzpicture}[scale=1.2]
    \coordinate (O) at (0,0,0);
    \coordinate (A) at (3,0,0);
    \coordinate (B) at (3,3,0);
    \coordinate (C) at (0,3,0);
    \coordinate (D) at (0,0,3);
    \coordinate (E) at (3,0,3);
    \coordinate (F) at (3,3,3);
    \coordinate (G) at (0,3,3);

    
    \foreach \x in {0,0.5,...,3} {
        \draw[gray!50,thin] (\x,0,0) -- (\x,3,0);
    }
    \foreach \y in {0,0.5,...,3} {
        \draw[gray!50,thin] (0,\y,0) -- (3,\y,0);
    }
    \draw[blue] (O) -- (A) -- (B) -- (C) -- cycle;
    \draw[blue, ultra thick] (0.5, 1.5, 0) -- (2.5, 1.5, 0);
    \node[blue, above] at (1.5, 1.5,0) {$\overline{X}_2$};
    \node[blue,above] at (1.5,3,0) {$\mathcal{Q}_2^{2D}$ (Y-Z)};
    
    \foreach \y in {0,0.5,...,3} {
        \draw[gray!50,thin] (0,\y,0) -- (0,\y,3);
    }
    \foreach \z in {0,0.5,...,3} {
        \draw[gray!50,thin] (0,0,\z) -- (0,3,\z);
    }
    \draw[red] (O) -- (C) -- (G) -- (D) -- cycle;
    \draw[red, ultra thick] (0, 0.5, 1.5) -- (0, 2.5, 1.5);
    \node[red, right] at (0, 1.5,1.5) {$\overline{X}_3$};
    \node[red,above] at (-0.75,3,1.5) {$\mathcal{Q}_3^{2D}$ (X-Z)};
    
    \foreach \x in {0,0.5,...,3} {
        \draw[gray!50,thin] (\x,0,0) -- (\x,0,3);
    }
    \foreach \z in {0,0.5,...,3} {
        \draw[gray!50,thin] (0,0,\z) -- (3,0,\z);
    }
    \draw[green] (O) -- (A) -- (E) -- (D) -- cycle;
    \draw[green, ultra thick] (1.5, 0, 0.5) -- (1.5, 0, 2.5);
    \node[green, right] at (1.75, 0,1.5) {$\overline{X}_1$};
    \node[green,below right] at (2.5,-0.5,1) {$\mathcal{Q}_1^{2D}$ (X-Y)};
    
    \draw[->] (O) -- (3.5,0,0) node[right] {Y};
    \draw[->] (O) -- (0,3.5,0) node[above] {Z};
    \draw[->] (O) -- (0,0,3.5) node[left] {X};
    
\end{tikzpicture}
\caption{Initial 2D HGP codes on orthogonal faces}
\label{fig:2d_codes}
\end{subfigure}
\hfill
\begin{subfigure}[b]{0.48\textwidth}
\centering
\begin{tikzpicture}[scale=1.2]
    \coordinate (O) at (0,0,0);
    \coordinate (A) at (3,0,0);
    \coordinate (B) at (3,3,0);
    \coordinate (C) at (0,3,0);
    \coordinate (D) at (0,0,3);
    \coordinate (E) at (3,0,3);
    \coordinate (F) at (3,3,3);
    \coordinate (G) at (0,3,3);
    
    \draw[gray,thin] (O) -- (A) -- (B) -- (C) -- cycle;
    \draw[gray,thin] (D) -- (E) -- (F) -- (G) -- cycle;
    \draw[gray,thin] (O) -- (D);
    \draw[gray,thin] (A) -- (E);
    \draw[gray,thin] (B) -- (F);
    \draw[gray,thin] (C) -- (G);

    \node[green,below right] at (2.5,-0.5,1) {$\mathcal{Q}_1^{3D}$};
    \node[red,above] at (-0.75,3,1.5) {$\mathcal{Q}_3^{3D}$};
    \node[blue,above] at (1.5,3,0) {$\mathcal{Q}_2^{3D}$};

    \fill[blue!40,opacity=0.75] (0.5,1.5,0) -- (2.5,1.5,0) -- (2.5,1.5,3) -- (0.5,1.5,3) -- cycle;
    \draw[blue,ultra thick] (0.5,1.5,0) -- (2.5,1.5,0);
    \node[blue, right] at (2.4,1.5,0) {$\overline{X}_2$};
    \fill[red!40,opacity=0.5] (0,0.5,1.5) -- (3,0.5,1.5) -- (3,2.5,1.5) -- (0,2.5,1.5) -- cycle;
    \draw[red,ultra thick] (0,0.5,1.5) -- (0,2.5,1.5);
    \node[red,above] at (0.2,2.5,1.5) {$\overline{X}_3$};
    \fill[green!40,opacity=0.5] (1.5,0,0.5) -- (1.5,3,0.5) -- (1.5,3,2.5) -- (1.5,0,2.5) -- cycle;
    \draw[green,ultra thick] (1.5,0,0.5) -- (1.5,0,2.5);
    \node[green,right] at (1.75,0,1.75) {$\overline{X}_1$};
    \fill[yellow,opacity=0.9] (1.5,1.5,1.5) circle (3pt);
    \draw[black,thick] (1.5,1.5,1.5) circle (3pt);

    \draw[gray, thin, dashed] (1.5, 3, 1.5) -- (1.5, 0, 1.5);
    \draw[gray, thin, dashed] (0, 1.5, 1.5) -- (3, 1.5, 1.5);
    \draw[gray, thin, dashed] (1.5, 1.5, 0) -- (1.5, 1.5,3);
    \draw[->] (O) -- (3.5,0,0) node[right] {Y};
    \draw[->] (O) -- (0,3.5,0) node[above] {Z};
    \draw[->] (O) -- (0,0,3.5) node[left] {X};
\end{tikzpicture}
\caption{Three perfectly superimposed 3D HGP codes with intersecting membrane $X$ logical operators}
\label{fig:3d_codes}
\end{subfigure}

\begin{subfigure}[b]{0.9\textwidth}
\centering
\begin{tikzpicture}[scale=1.0]
    \begin{scope}[shift={(-4,-2)}]
        \node[above,font=\small] at (1.25,4,0) {\textbf{2D HGP Code}};
        
         \coordinate (O) at (0,0,0);
        \coordinate (A) at (3,0,0);
        \coordinate (B) at (3,3,0);
        \coordinate (C) at (0,3,0);
        \coordinate (D) at (0,0,3);
        \coordinate (E) at (3,0,3);
        \coordinate (F) at (3,3,3);
        \coordinate (G) at (0,3,3);
        
        
        \foreach \y in {0,0.5,...,3} {
            \draw[gray!50,thin] (0,\y,0) -- (0,\y,3);
        }
        \foreach \z in {0,0.5,...,3} {
            \draw[gray!50,thin] (0,0,\z) -- (0,3,\z);
        }
        \draw[red] (O) -- (C) -- (G) -- (D) -- cycle;
        \draw[red, ultra thick] (0, 0.5, 1.5) -- (0, 2.5, 1.5);
        \node[red, left] at (0.05, 1.5,1.5) {{\footnotesize $\overline{X}_3$}};
        \node[red,above] at (-0.75,3,1.5) {$\mathcal{Q}_3^{2D}$ (X-Z)};
        
        
    \draw[->] (O) -- (3.5,0,0) node[right] {Y};
    \draw[->] (O) -- (0,3.5,0) node[above] {Z};
    \draw[->] (O) -- (0,0,3.5) node[left] {X};
    \end{scope}
    
    \draw[->,ultra thick] (-1.5,0,0) -- (1.5,0,0);
    \node[above] at (0,0.3,0) {\textbf{Code-switch}};
    
    \begin{scope}[shift={(4,-2)}]
        \node[above,font=\small] at (1.25,4,0) {\textbf{3D HGP Code}};
        
         \coordinate (O) at (0,0,0);
        \coordinate (A) at (3,0,0);
        \coordinate (B) at (3,3,0);
        \coordinate (C) at (0,3,0);
        \coordinate (D) at (0,0,3);
        \coordinate (E) at (3,0,3);
        \coordinate (F) at (3,3,3);
        \coordinate (G) at (0,3,3);

        \foreach \x in {0,0.5,...,3} {
          \foreach \y in {0,0.5,...,3} {
              \draw[gray!50,thin] (\x,\y,0) -- (\x,\y,3);
          }
          \foreach \z in {0,0.5,...,3} {
              \draw[gray!50,thin] (\x,0,\z) -- (\x,3,\z);
          }
        }

        \foreach \y in {0, 0.5,...,3} {
          \foreach \z in {0,0.5,...,3} {
              \draw[gray!50,thin] (0,\y,\z) -- (3,\y,\z);
          }
        }

        \draw[red] (O) -- (A) -- (B) -- (C) -- cycle;
        \draw[red] (D) -- (E) -- (F) -- (G) -- cycle;
        \draw[red] (O) -- (D);
        \draw[red] (A) -- (E);
        \draw[red] (B) -- (F);
        \draw[red] (C) -- (G);

    \fill[red!40,opacity=0.5] (0,0.5,1.5) -- (3,0.5,1.5) -- (3,2.5,1.5) -- (0,2.5,1.5) -- cycle;
    \draw[red,ultra thick] (0,0.5,1.5) -- (0,2.5,1.5);
    \node[red, left] at (0.05, 1.5,1.5) {{\footnotesize $\overline{X}_3$}};           
    \node[red,above] at (-0.75,3,1.5) {$\mathcal{Q}_3^{3D}$};
     
    \draw[red,<->] (0, 3.2, 0) -- (3, 3.2, 0);
    \node[red, above] at (1.5, 3.2, 0) {\footnotesize $T(\calG_2, \calC_0)$};
        
    \draw[->] (O) -- (3.5,0,0) node[right] {Y};
    \draw[->] (O) -- (0,3.5,0) node[above] {Z};
    \draw[->] (O) -- (0,0,3.5) node[left] {X};
    \end{scope}
\end{tikzpicture}
\caption{Single-shot code-switching from 2D to 3D HGP code}
\label{fig:transformation}
\end{subfigure}

\caption{Construction of 3D hypergraph product codes for transversal logical CCZ gates. (a) Three 2D HGP codes $\mathcal{Q}_1^{2D}$ (X-Y), $\mathcal{Q}_2^{2D}$ (Y-Z), and $\mathcal{Q}_3^{2D}$ (X-Z) positioned on orthogonal faces of a cube. Note that $\mathcal{Q}_1^{2D}$ (X-Y) is simply $\mathcal{Q}_1^{3D}$ but without the generalized repetition code component in the $Z$ direction. The other 2D HGP codes can be defined in the same way. (b) After code-switching, each 2D HGP code grows into a 3D HGP code. The three 3D HGP codes can be perfectly superimposed on each other. The now membrane-like logical $X$ operators $\overline{X}_1$, $\overline{X}_2$, and $\overline{X}_3$ intersect at a single point (yellow), enabling transversal CCZ implementation. (c) Schematic of the single-shot code-switching from 2D to 3D HGP code for $\calQ_3$. The generalized repetition code component $T(\calG_2, \calC_0)$ in the $Y$ direction is indicated. The $X$ logical operator grows from a string to a membrane. The same procedure applies to $\calQ_1$ and $\calQ_2$.}
\label{fig:ccz_construction}
\end{figure}
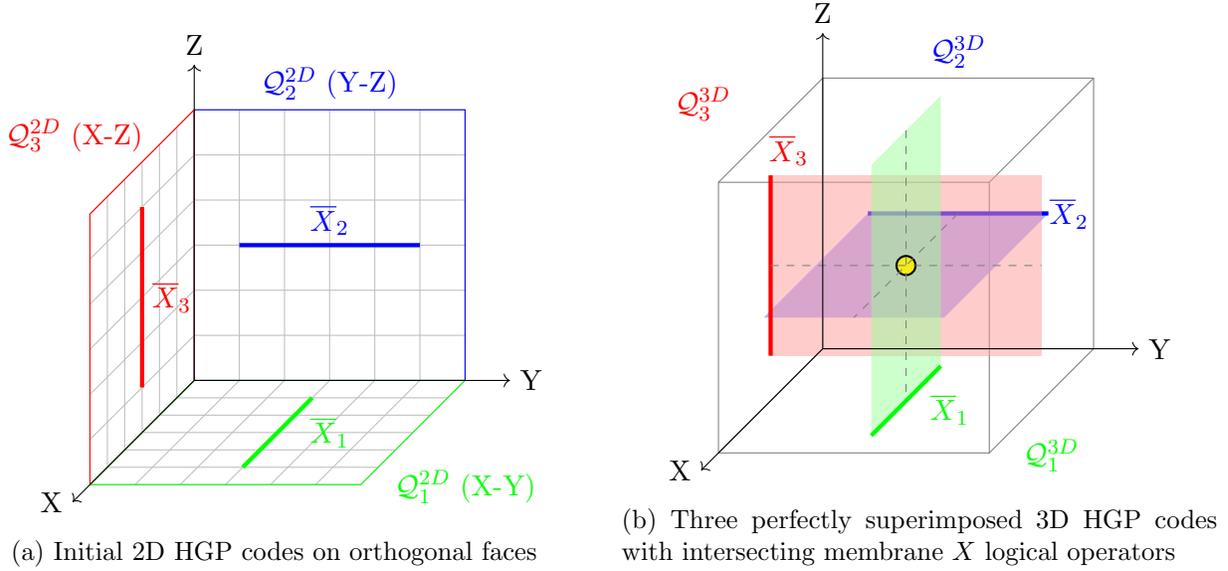
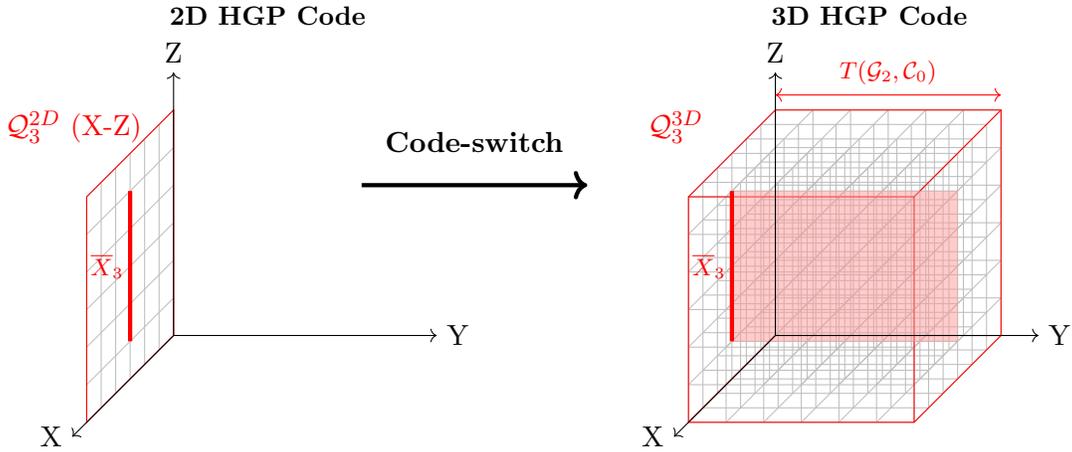

Now, we state the following lemma that shows how the above simple construction satisfies the local multiplication property.

\begin{theorem}[Satisfaction of the Local Multiplication Property]\label{thm:satisfaction_of_local_multiplication_property}
  Given arbitrary spectral expanders $\calG_1$, $\calG_2$, and $\calG_3$ as well as arbitrary classical local codes $\calC_1$, $\calC_2$, and $\calC_3$, the 3D HGP codes $\calQ_1$, $\calQ_2$, and $\calQ_3$ constructed above in Definition~\ref{def:simple_construction_3D_HGP_CCZ} satisfy the local multiplication property required for Theorem~\ref{thm:local-multiplication-ccz-property}.
\end{theorem}

\begin{proof}
  We need to show that the local codes of $\calQ_1$, $\calQ_2$, and $\calQ_3$ satisfy the local multiplication property. This follows from the structure of the local codes and their relationships.
  Specifically, we need to verify that for any three copies of Tanner codes $T(\calG_i, \,\cdot\,)$ that corresponds to any $i \in \{1, 2, 3\}$, the local codes $\calC_{1,\mathrm{loc}}$, $\calC_{2,\mathrm{loc}}$, $\calC_{3,\mathrm{loc}} \subseteq \F_2^\Delta$ corresponding to $\calQ_1$, $\calQ_2$, and $\calQ_3$ satisfy the following condition:
  \[
  \sum_i^\Delta (\mbf{c}_1)_i (\mbf{c}_2)_i (\mbf{c}_3)_i = 0 \mod 2, \quad \forall \mbf{c}_1 \in \calC_{1,\mathrm{loc}}, \mbf{c}_2 \in \calC_{2,\mathrm{loc}}, \mbf{c}_3 \in \calC_{3,\mathrm{loc}}.
  \]
  This can be shown by analyzing the structure of the local codes and their interactions. Let us do that for $\calG_1$ and the same argument applies to $\calG_2$ and $\calG_3$.
  In this case, the three local codes that we have to check are given by $\calC_1, \calC_0, \calC_1^\perp$ respectively. By the definition of the dual code, we have
  \[\sum_i^\Delta (\mbf{c}_1)_i (\mbf{c}_3)_i = 0 \mod 2, \quad \forall \mbf{c}_1 \in \calC_1, \mbf{c}_3 \in \calC_1^\perp.\]
  Because the codewords of $\calC_0$ are either $0^{\Delta}$ or $1^{\Delta}$, we have
  \[\sum_i^\Delta (\mbf{c}_1)_i (\mbf{c}_2)_i (\mbf{c}_3)_i = 0 \mod 2, \quad \forall \mbf{c}_1 \in \calC_1, \mbf{c}_2 \in \calC_0, \mbf{c}_3 \in \calC_1^\perp.\]
  This completes the proof for $\calG_1$. The same argument applies to $\calG_2$ and $\calG_3$.
\end{proof}

The lemma above provides the foundation for implementing transversal CCZ gates in the described construction. The local multiplication property ensures compatibility between the local codes, which is critical for fault-tolerant CCZ computation. This construction is inspired by the ideas in Zhu's paper~\cite{zhu2025topological} and can be viewed as a vastly simpler and cleaner version of the construction used in Ref.~\cite{golowich2025quantum}, where much more advanced techniques were used to achieve better code parameters with punctured Reed-Solomon local codes on a special expander graph at the expense of the LDPC property. Similar to the construction in Ref.~\cite{zhu2025topological}, our construction offers great flexibility in the choice of the expander graphs and local codes, allowing for a wide range of code parameters and properties. This flexibility is particularly useful for finding interesting finite size examples that may have useful logical gates and code automorphisms~\cite{berthusen2025automorphism}.
We note that Theorem~2 in Ref.~\cite{zhu2025topological} provides a concrete way to identify the necessary cohomology classes that give rise to non-trivial logical CCZ gates. By using Zhu's result, it is possible to find the number of logical CCZ gates that can be implemented transversally for a given choice of the expander graphs and local codes. 

We now proceed to show how the construction defined in Definition~\ref{def:simple_construction_3D_HGP_CCZ} can satisfy the rate scaling described in Theorem~\ref{thm:universal_FTQC}.

\begin{lemma}\label{lem:3D_HGP_code_rate}
  The following 2D HGP codes constructed using the Sipser-Spielman codes built from spectral expander graphs $\calG_i$ and classical local codes $\calC_i$ for $i=1,2,3$ have constant rate:
  \begin{itemize}
    \item $\calQ_1^{2D} = T(\calG_1, \calC_1) \otimes T(\calG_2, \calC_2^\perp)$
    \item $\calQ_2^{2D} = T(\calG_2, \calC_2) \otimes T(\calG_3, \calC_3^\perp)$
    \item $\calQ_3^{2D} = T(\calG_1, \calC_1^\perp) \otimes T(\calG_3, \calC_3)$
  \end{itemize}
    if the classical local codes $\calC_1$, $\calC_2$, and $\calC_3$ have rates $r_1$, $r_2$, and $r_3$, respectively with $r_1, r_2, r_3 > 1/2$.
\end{lemma}
\begin{proof}    
    It is known that the rate of a Sipser-Spielman code is at least $2r - 1$ where $r$ is the rate of the classical local code~\cite{sipser2002expander}.
    Thus, when we define $\calC_i$ to have rate greater than 1/2, the Sipser-Spielman code $T(\calG_i, \calC_i)$ has positive rate, implying that the first homology group associated with its chain complex has dimension that scales linearly with the number of physical bits. On the other hand, $T(\calG_i, \calC_i^\perp)$ would have a vanishing rate but its transpose would have positive rate, implying that the zeroth homology group associated with its chain complex has dimension that scales linearly with the number of physical bits.     
    By the K\"unneth formula stated in Proposition~\ref{prop:kunneth-formula}, the 2D HGP codes $\calQ_1^{2D}$, $\calQ_2^{2D}$, and $\calQ_3^{2D}$ would achieve constant rate. 
\end{proof}

The rate of the 3D HGP codes $\calQ_1^{3D}$, $\calQ_2^{3D}$, and $\calQ_3^{3D}$ constructed in Definition~\ref{def:simple_construction_3D_HGP_CCZ} can be derived from the rate of the corresponding 2D HGP codes $\calQ_1^{2D}$, $\calQ_2^{2D}$, and $\calQ_3^{2D}$ and the rate of the generalized repetition code $T(\calG_i, \calC_0)$ using the K\"unneth formula stated in Proposition~\ref{prop:kunneth-formula}. Thus, we have proven the theorem statement.
\end{proof}

\section{Linear Soundness in the 3D HGP Code}
\label{sec:3D_soundness}

In this section, we will show that our 3D HGP code $\calQ_G$ exhibits small-set linear soundness in the membrane basis.
Suppose our 3D HGP code of length $\tilde{\Delta} = \Theta(n^{3/2})$ is extended in the $X$ basis ($X$ logicals are membrane-like). Then we have the following 4-term cochain complex:
\begin{align}\label{eq:3D HGP chain}
    \tilde{S}_X \xlongrightarrow{\tilde{H}^\top_X} \tilde{Q} \xlongrightarrow{\tilde{H}_Z} \tilde{S}_Z \xlongrightarrow{\tilde{M}_Z} \tilde{R}_Z \, ,
\end{align}
where $S_X$ is the space of $X$ syndromes, $Q$ the space of physical qubits, $S_Z$ the space of $Z$ syndromes, and $R_Z$ the space of $Z$-check relations. $M_Z$ is the $Z$-metacheck matrix, and $H_X$ and $H_Z$ are the usual CSS parity-check matrices. We want to show that $H_Z$ is $(t,f)$-sound with $t=\Theta(n^{1/3})$ and $f(x)=\alpha x$ for some $\alpha=O(1)$. To get this result, we will show small-set coboundary expansion for the metachecks $M_Z$. It is known that small-set coboundary expansion can be reduced to the property of spectral expansion of the graphs and the product expansion of the underlying classical codes \cite{dinur2024expansion}. By construction the graphs we use are spectral expanders, and so it suffices to check the product expansion of the underlying classical codes. It is known that a pair of random codes are product-expanding \cite{dinur2023good, kalachev2022two}. As such, the local codes we will use are two random codes and a repetition code. We show below that adding a repetition code to a product-expanding pair of codes preserves the product expansion property.

Given two linear codes $C_1, C_2\subseteq \F_2^{\Delta}$, we define $C_1\boxplus C_2 = C_1\otimes \F_2^{\Delta}+\F_2^{\Delta}\otimes C_2 \subseteq \F_2^{\Delta\times\Delta}$. More generally, given a collection of linear codes $C = (C_1, C_2, ..., C_D)$ of linear codes over $\F_2$, we define the codes
\begin{equation}
    C^{(i)} := \F_2^{\Delta} \otimes\dots\otimes C_i \otimes\dots\otimes \F_2^{\Delta}
\end{equation}
which is formed by tensors $x \in \F_2^{\Delta^D}$ whose restriction along the $i$-th axis are codewords in $C_i$. It is clear that
\begin{equation}
  C_1\boxplus \dots \boxplus C_D = C^{(1)}+ \dots + C^{(D)}.
\end{equation}
For each $i \in [D]$, let $L_i$ be the set of lines parallel to the $i$-th axis in the $D$-dimensional grid $[\Delta]^D$. For each tensor $x\in \F_2^{\Delta^D}$, we define $|x|$ as the Hamming weight of $x$, and define $|x|_i$ as the number of lines $\ell\in L_i$ such that $x$ is nonzero on that line, $x|_\ell \ne 0$.

\begin{definition}[Product expansion \cite{kalachev2025maximally}]
\label{def:product expansion}
  Given a collection of linear codes $C_i \subset \F_2^\Delta$,
    we say that $(C_1, ..., C_D)$ is $\rho$-product-expanding
    if every codeword $c\in C_1\boxplus \dots \boxplus C_D$ can be represented as a sum
      $c = \sum_{i=1}^D a_i$,
      where $a_i\in C^{(i)}$ for all $i\in [D]$,
    and the following inequality holds:
    \begin{equation}\label{eq:prod-exp}
      \rho \Delta \sum_{i\in [D]} |a_i|_i \le |c|.
    \end{equation}
\end{definition}

\begin{lemma}[Product expansion with repetition]
\label{lem:3D product expansion}
  If $(C_1, C_2)$ is $\rho$-product-expanding, then the $(C_1, C_2, C_{\rm rep})$ is $\rho/3$-product-expanding.
\end{lemma}

\begin{proof}
  Let $c \in C_1\boxplus C_2 \boxplus C_{\rm rep}$ be a codeword.
  Consider the 2D slices of $c$ parallel to the first two axes,
    which we write $c = \sum_{i=1}^\Delta c_i \otimes e_i$
    where $c_i \in \F_2^{\Delta\times \Delta}$ is the $i$-th slice
    and $e_i \in \F_2^\Delta$ is the $i$-th standard basis vector.
  Because $C_{\rm rep}$ is the repetition code,
    any tensor in $C^{(3)}$ is constant on each line parallel to the third axis.
  Therefore, $c_i - c_j \in C_1\boxplus C_2$ for all $i,j \in [\Delta]$.
Let $\varepsilon \in \F_2^{\Delta\times \Delta}$
    be the difference between $c_1$ and its closest codeword in $C_1\boxplus C_2$
    in the Hamming distance.
  Then $c_i - \varepsilon \in C_1\boxplus C_2$ for all $i \in [\Delta]$.
  It is $|\varepsilon| \le |c_i|$ for all $i,j \in [\Delta]$.
Let $c' = \sum_{i=1}^\Delta (c_i - \varepsilon) \otimes e_i$
    and let $c'' = \varepsilon \otimes (1,1,...,1)$.
  Then $c = c' + c''$ where $c'\in C^{(1)} + C^{(2)}$ and $c''\in C^{(3)}$.
  Since $(C_1, C_2)$ is $\rho$-product-expanding,
    we can write $c_i - \varepsilon = a_{i,1} + a_{i,2}$,
    where $a_{i,1} \in C_1 \otimes \F_2^{\Delta}$
      and $a_{i,2} \in \F_2^{\Delta} \otimes C_2$
    such that $\rho \Delta (|a_{i,1}|_1 + |a_{i,2}|_2) \le |c_i - \varepsilon|$.

  Using $a_{i,1}$, $a_{i,2}$, and $\varepsilon$,
    we are ready to construct a decomposition of $c$ into $c = a_1 + a_2 + a_3$.
  Let $a_1$ be the tensor whose $i$-th slice is $a_{i,1}$,
    and let $a_2$ be the tensor whose $i$-th slice is $a_{i,2}$.
  Let $a_3$ be the tensor $\varepsilon \otimes (1,1,...,1)$
    which is constant on each line parallel to the third axis.
  It is clear that $a_1 \in C^{(1)}$, $a_2 \in C^{(2)}$, and $a_3 \in C^{(3)}$.

  We now verify the inequality \eqref{eq:prod-exp}.
  For each 2D slice, we have
    \begin{equation}
      \rho \Delta (|a_{i,1}|_1 + |a_{i,2}|_2) \le |c_i - \varepsilon|,
    \end{equation}
    which implies
    \begin{equation}
      |a_{i,1}|_1 + |a_{i,2}|_2 + \frac{|\varepsilon|}{\Delta}
      \le \frac{|c_i - \varepsilon|}{\rho \Delta} + \frac{|\varepsilon|}{\Delta}
      \le \frac{|c_i| + |\varepsilon| + |\varepsilon|}{\rho \Delta},
      \le \frac{3 |c_i|}{\rho \Delta},
    \end{equation}
    where we used the fact that $\rho \le 1$
    and $|\varepsilon| \le |c_i|$.
  Summing over all $i \in [\Delta]$,
    \begin{equation}
      |a_1|_1 + |a_2|_2 + |a_3|_3
      = \sum_{i=1}^\Delta \left(|a_{i,1}|_1 + |a_{i,2}|_2 + \frac{|\varepsilon|}{\Delta}\right)
      \le \sum_{i=1}^\Delta \frac{3 |c_i|}{\rho \Delta}
      = \frac{3 |c|}{\rho \Delta}.
    \end{equation}
  It follows that $(C_1, C_2, C_{\rm rep})$ is $\rho/3$-product-expanding.
\end{proof}

Using the product-expansion promised by Lemma \ref{lem:3D product expansion}, we can now show small-set coboundary expansion of $\tilde{M}_Z$ in \eqref{eq:3D HGP chain}.

\begin{theorem}[3D Metacheck expansion]
\label{thm:3D metacheck expansion}
  The 3D HGP code $\calQ_G$ exhibits small-set coboundary expansion in $\tilde{M}_Z$.
  In particular,
    there exists a $t = \Theta(\tilde{n}^{1/3})$ such that for any $s \in S_Z$ with $\abs{s}\leq t$,
    \begin{equation}\label{eq:3D metacheck expansion}
      \abs{\tilde{M}_Z s} \ge \kappa \min_{e \in Q} \abs{s + \tilde{H}_Z e},
    \end{equation}
    where $\kappa = \Theta_{\rho, \lambda, \Delta}(1)$.
  Moreover, we also have $\abs{e} \leq \alpha\abs{s}$.
\end{theorem}

\begin{proof}
    Since $\calG_1, \calG_2, \calG_3$ are $\lambda$-spectral expanders,
    the subgraphs of $\calG_1 \times \calG_2 \times \calG_3$
    corresponding to the $i$-th direction
    are also $\lambda$-spectral expanders
    up to size $\Theta(\abs{\calG_i}) = \Theta(\tilde{n}^{1/3})$.
  Combined with the fact that $(C_1, C_2, C_{\rm rep})$ is $\rho$-product-expanding, it follows from \cite[Proposition 7.1]{dinur2024expansion}
    that there exists $t = \Theta(\tilde{n}^{1/3}) < d_{\rm coLM}$,
    where $d_{\rm coLM}$ is the locally co-minimal distance (Definition \ref{def:locally minimal distance}). Small-set coboundary expansion \eqref{eq:3D metacheck expansion} similarly follows from Theorem 3.6 in \cite{dinur2024expansion}.

  We now prove the statement $\abs{e} \leq \alpha\abs{s}$. If $s$ is not locally co-minimal (Definition \ref{def:locally minimal}), then 
    there exists a local flip
    $s_1 = \tilde{H}_Z e_1$
    such that $|s + s_1| < |s|$.
  Since $s + s_1$ has a smaller weight than $|s| < d_{\rm coLM}$,
    we can repeat the process to obtain a sequence of local flips
    that eventually reduces $s$ to its locally co-minimal representative.
  Let $s_1, s_2, \dots, s_T$ be these local flips,
    where each $s_i = \tilde{H}_Z e_i$.
  This implies $s = \sum_{i=1}^T s_i$
    and we define $e = \sum_{i=1}^T e_i$.
  It is clear that $s = \tilde{H}_Z e$.

  To bound the weight of $e$,
    note that each step decreases the weight of $s$ by at least $1$,
    which implies $T \le |s|$.
  Since the chain complex is LDPC,
    each local flip $e_i$ has bounded weight $|e_i| = \Theta(1)$.
  Therefore, we have
  \begin{equation}
    |e| \le \sum_{i=1}^T |e_i| = O(T) = O(|s|),
  \end{equation}
  which implies $\abs{e} \leq \alpha\abs{s}$ for some constant $\alpha$.
\end{proof}

An immediate corollary of Theorem \ref{thm:3D metacheck expansion}, for the case when $s \in \ker{\tilde{M}_Z}$, is (small-set) linear soundness in our 3D HGP code.

\begin{corollary}
\label{cor:3D HGP linear soundness}
  The 3D HGP code $\calQ_G$ is $(t,f)$-sound with $t=\Theta(\tilde{n}^{1/3})$ and $f(x)=\alpha x$ for $\alpha=O(1)$.
  In particular,
    for any $s \in \ker \tilde{M}_Z$ with $|s| \le t$,
    there exists $e \in Q$ such that $s = \tilde{H}_Z e$
    and $\abs{e} \leq \alpha\abs{s}$.
\end{corollary}

\section{Fault Tolerance for Adversarial Noise}\label{sec:adversarial_noise}
In this section, we provide a rigorous analysis of the fault tolerance of the gadgets in our universal protocol against adversarial noise. To be explicit, we consider an adversarial noise model where an adversary can choose any set of $s$ fault locations in the circuit and apply arbitrary Pauli errors at those locations. This includes errors on the data and ancilla qubit errors, measurement errors, preparation errors, and gate errors. We first review the necessary background on the fault tolerance properties of gadgets. Subsequently, we show that our gadgets satisfy these properties, thereby establishing the fault tolerance of our universal protocol by a simple composition argument.

\subsection{Fault Tolerance of Gadgets}
\label{sec:fault_tolerance_gadgets}
In this section, we mostly restate the definitions and results from Ref.~\cite{gottesman2024surviving}.
We start off by reviewing what it means for a gate gadget to be fault-tolerant under the single-shot universality framework. Note that a gate gadget essentially captures the idea of a logical gate implemented on an encoded state, and so it is a more general notion than a transversal gate. In fact, it also encapsulates lattice surgery and code switching protocols.

\begin{definition}[Fault-tolerant gate gadget~{\cite[Restatement of Definitions 10.10 and 10.11]{gottesman2024surviving}}]
\label{def:fault_tolerant_gate_gadget}
    Suppose we have a gate gadget that is associated with a QLDPC code with distance $2t+1$ that has $u$ input and output bits for some constant $u$. Let $r_i$ for $i=1,\ldots,u$ be the incoming faults for each of the $u$ input qubits, and let $s$ be the total number of faults that occur during the gadget. 
    A gate gadget is said to be fault-tolerant if it satisfies the following properties when $\sum_{i=1}^u r_i + f(s) \leq t$ for some increasing function $f$:
    \begin{itemize}
        \item (Gate Error Propagation Property) The output $u$-qubit state has at most $\sum_{i=1}^u r_i + f(s)$ faults on each of the $u$ output qubits.
        \item (Gate Correctness Property) The output $u$-qubit state when passed through an ideal decoder implements the intended logical operation on the input state. In other words, the output state from the noisy gate gadget when passed through an ideal decoder is the same as if we had passed the input state through an ideal decoder followed by the ideal logical gate operation.
    \end{itemize}
\end{definition}

It should not be too hard to see that transversal, fold-transversal, and swap-transversal gates are all examples of fault-tolerant gate gadgets with a constant factor difference in the function $f$.
Next, we state the definition of fault-tolerance for a state-preparation gadget, which is a procedure to prepare an encoded logical state.

\begin{definition}[Fault-tolerant state-preparation gadget~{\cite[Restatement of Definitions 10.13 and 10.14]{gottesman2024surviving}}]
\label{def:fault_tolerant_state_preparation_gadget}
    Suppose we have a state-preparation gadget that is associated with a QLDPC code with distance $2t+1$. Let $s$ be the total number of faults that occur during the gadget. A state-preparation gadget is said to be fault-tolerant if it satisfies the following properties when $f(s) \leq t$ for some increasing function $f$:
    \begin{itemize}
        \item (State-Preparation Error Propagation Property) The output state has at most $f(s)$ faults.
        \item (State-Preparation Correctness Property) The output state when passed through an ideal decoder is the intended logical state.
    \end{itemize}
\end{definition}

Subsequently, we state the definition of fault-tolerance for a logical measurement gadget, which is a procedure to measure a logical operator on an encoded state.

\begin{definition}[Fault-tolerant measurement gadget~{\cite[Restatement of Definition 10.15]{gottesman2024surviving}}]
\label{def:fault_tolerant_logical_measurement_gadget}
    Suppose we have a measurement gadget that is associated with a QLDPC code with distance $2t+1$. Let $r$ be the number of incoming faults, and let $s$ be the total number of faults that occur during the gadget. A logical measurement gadget is said to be fault-tolerant if it satisfies the following properties when $r + f(s) \leq t$ for some increasing function $f$:
    \begin{itemize}
        \item (Measurement Correctness Property) The measurement outcome is the same as if we had passed the input state through an ideal decoder followed by an ideal logical measurement.
    \end{itemize}
\end{definition}

Finally, we state the definition of fault-tolerance for a single-shot error-correction gadget.

\begin{definition}[Fault-tolerant single-shot error-correction gadget~{\cite[Restatement of Definitions 10.17 and 10.18]{gottesman2024surviving}}]
\label{def:fault_tolerant_single_shot_error_correction_gadget}
    Suppose we have a single-shot error-correction gadget that is associated with a QLDPC code with distance $2t+1$. Let $s$ be the total number of faults that occur during the gadget. A single-shot error-correction gadget is said to be fault-tolerant if it satisfies the following properties:
    \begin{itemize}
        \item (Error-Correction Error Recovery Property) When $f(s) \leq t$ for some increasing function $f$, the output state has at most $f(s)$ faults.
        \item (Error-Correction Correctness Property) Let $r$ be the number of incoming faults into the gadget. When $r + f(s) \leq t$ for some increasing function $f$, the output state when passed through an ideal decoder is the same as if we had passed the input state through an ideal decoder.
    \end{itemize}
\end{definition}

One can interpret the error recovery property as a guarantee that our fault-tolerant gadget will always return us to some codeword with some bounded number of errors no matter how many faults there are to start with. On the other hand, the correcness property guarantees that when the total number of faults is small enough, the encoded logical state does not deviate from the ideal logical state during the gadget.

\subsection{Fault Tolerance of Single-Shot State Preparation}
\label{sec:fault_tolerance_state_preparation}
We now analyze the fault tolerance of our single-shot state-preparation gadgets against adversarial errors. With slight abuse of notation, we will reuse the same letter for both a vector and its support, and it should be clear by context which one we are referring to.

For the 3D HGP code, we have the following chain complex for the membrane sector:
\begin{align}
    Q \xlongrightarrow{\mbf{H}} S \xlongrightarrow{\mbf{M}} R
\end{align}
where $Q,S,R$ are the vector spaces of qubits, syndromes and relations respectively, and $\mbf{H}$ and $\mbf{M}$ are the (low-density) parity-check and metacheck matrices respectively. We call $\im{\mbf{H}}$ the space of \emph{physical} syndromes since these are the ones that can be produced by errors on the data qubits in $Q$. The systolic distance, or smallest nontrivial weight in $\ker{\mbf{M}}\,\setminus\,\im{\mbf{H}}$, is called the single-shot or metacode distance and is labeled as $d_{\rm ss}$.

For the setting of adversarial noise, the 3D HGP code is amenable to single-shot error correction in at least one of the Pauli basis using the metacheck matrix, which has been described by Campbell \cite{campbell2019theory}. For completeness, we recite the relevant results below. We will assume a theoretical minimum-weight decoder for both the metacode and the physical code; note that some constants (like 1/2) can be modified and provide the same qualitative behavior if one wishes to relax the minimum-weight assumption on the decoder.

We first introduce some relevant notation. Let $\mbf{s}_{\rm true}$ be the true syndrome and $\mbf{s}_{\rm err}$ the syndrome error such that $\mbf{s}_{\rm obs} = \mbf{s}_{\rm true} + \mbf{s}_{\rm err}$ is the observed, corrupted syndrome. Our decoder for $M$, which we will call a metadecoder, will output a syndrome correction $\mbf{s}_{\rm corr}$ such that $\mbf{s}_{\rm rep} = \mbf{s}_{\rm obs} + \mbf{s}_{\rm corr}$ is the repaired or inferred syndrome. The metadecoder succeeds if $\mbf{s}_{\rm rep} \in \im{\mbf{H}}$; since $\mbf{s}_{\rm true} \in \im{\mbf{H}}$ by default, the success requirement is equivalent to $\mbf{s}_{\rm err} + \mbf{s}_{\rm corr} \equiv \mbf{s}_{\rm res} \in \im{\mbf{H}}$. We call $\mbf{s}_{\rm res}$ the \emph{residual} syndrome, since it will end up as the syndrome for the eventual residual error.

At a high level, the distance $d_{\rm ss}$ of the metacode will tell us how many syndrome errors we can tolerate before $\mbf{s}_{\rm res} = \mbf{s}_{\rm err}+\mbf{s}_{\rm corr}$ lies outside of $\im{\mbf{H}}$ and hence cannot be decoded to a physical correction. Then the soundness property will tell us how large the corresponding residual error $\mbf{e}_{\rm res}$ can be, given the residual syndrome $\mbf{s}_{\rm res} \in \im{\mbf{H}}$. If the initial physical and syndrome errors are small enough, then the combination of the residual error with a new round of physical and syndrome errors will still be correctable, and hence sustainable single-shot QEC is achieved over many (noisy) syndrome extraction rounds.

\begin{lemma}[Metacode success; Lemma 1 in \cite{campbell2019theory}]
    If $\abs{\mbf{s}_{\rm err}} < d_{\rm ss}/2$, and we have access to a minimum-weight metadecoder, then the residual syndrome $\mbf{s}_{\rm res} \in \im{\mbf{H}}$ is physical.
\end{lemma}

\begin{proof}
    Since $\mbf{s}_{\rm err}$ itself is also a valid syndrome correction, by minimality of our metadecoder, we have $\abs{\mbf{s}_{\rm corr}} \leq \abs{\mbf{s}_{\rm err}}$. Hence $\abs{\mbf{s}_{\rm res}} = \abs{\mbf{s}_{\rm err} + \mbf{s}_{\rm corr}} \leq \abs{\mbf{s}_{\rm err}} + \abs{\mbf{s}_{\rm corr}} \leq 2\abs{\mbf{s}_{\rm err}} < d_{\rm ss}$, which implies that $\mbf{s}_{\rm res}$ belongs to the trivial class in $\ker{\mbf{M}}/\im{\mbf{H}}$, i.e. $\mbf{s}_{\rm res} \in \im{\mbf{H}}$.
\end{proof}

Note that the definition of soundness (Definition \ref{def:soundness}) requires $t<d_{\rm ss}$, since it asserts that any syndrome with weight less than or equal to $t$ is physical and so belongs to $\im{\mbf{H}}$.

\begin{lemma}[Residual error bound; Lemma 2 in \cite{campbell2019theory}]
    \label{lem:residual error bound}
    Assume our code is $(t,f)$-sound with $f$ an increasing function, and let $\mbf{s}_{\rm err}$ the syndrome error. If $\abs{\mbf{s}_{\rm err}} < t/2$, then the residual error satisfies 
    \begin{align}
        \abs{\mbf{e}_{\rm res}} \leq f(2\abs{\mbf{s}_{\rm err}}) < f(t) \, .
    \end{align}
\end{lemma}

\begin{proof}
    By minimality of our metadecoder, $\abs{\mbf{s}_{\rm res}} \leq 2\abs{\mbf{s}_{\rm err}} < t$, and so the residual syndrome lies within the soundness radius of the code. Therefore, by Definition \ref{def:soundness} and the fact that $f$ is increasing, we have $\abs{\mbf{e}_{\rm res}} < f(t)$.
\end{proof}

For our 3D HGP $\calQ_G$, Lemma 6 in \cite{campbell2019theory} tells us that our soundness radius $t = \Theta(d_c)$, assuming that the repetition code in the third dimension has distance $\Theta(d_c)$, where $d_c$ is the distance of the two classical codes in the 2D HGP. All that remains is to determine the soundness function $f$. By Corollary~\ref{cor:3D HGP linear soundness}, we have $f(x) = \Theta(x)$, i.e., linear soundness.

Now, we are ready to prove the fault tolerance of our single-shot state-preparation gadgets. Note that the following statement holds for both the 2D HGP codes and the 3D HGP codes $\calQ_S$ and $\calQ_G$. In other words, we can prepare the logical $\ket{\overline{0}}$ and $\ket{\overline{+}}$ states of both the 2D and 3D HGP codes in a fault-tolerant manner. For the 2D HGP code case, there is just the additional dimensional reduction step whose fault tolerance in the adversarial noise model has already been established in Ref.~\cite{hong2024single}.

\begin{lemma}[Fault-tolerant single-shot state preparation] \label{lem:fault_tolerant_state_preparation}
    Our single-shot state-preparation gadgets for the logical computational basis states of the 3D HGP code $\calQ_G$ (in the basis with the redundant checks) are fault-tolerant according to Definition~\ref{def:fault_tolerant_state_preparation_gadget} with $f(s) = \Theta(s)$.
\end{lemma}

\begin{proof}
    Suppose our 3D HGP code has $X$ membrane logical operators and $Z$ string logical operators without loss of generality. In addition, consider the case where we have Pauli $X$ errors $\mbf{e}_X$, Pauli $Z$ errors $\mbf{e}_Z$, and syndrome measurement errors $\mbf{s}_{\rm err}$ during the state-preparation gadget. The total number of faults is $s = \abs{\mbf{e}_X} + \abs{\mbf{e}_Z} + \abs{\mbf{s}_{\rm err}} \leq t$ where $t = (d-1)/2$ and $d$ is the distance of the 3D HGP code. Then, a minimum-weight metadecoder and the decoders stated in Ref.~\cite{hong2024single} ensure that $\mbf{e}_X$ and $\mbf{s}_{\rm err}$ will be converted into an $X$ residual error with weight that is bounded by some linear function of $\abs{\mbf{s}_{\rm err}}$ by Lemma~\ref{lem:residual error bound} and the linear soundness of our 3D HGP codes given by Corollary~\ref{cor:3D HGP linear soundness}. Combining this residual error with the original $Z$ error $\mbf{e}_Z$ will give us a total residual error with weight that is upper-bounded by some linear function of $s$. This establishes the error propagation property. By choosing to define the linear function $f$ with the appropriate constants, we can also ensure that $f(s) \leq t$, which establishes the correctness property when given an ideal decoder. 
\end{proof}


For the case of preparing the logical computational basis state in the 2D HGP code $\calQ$, we committed to using the Bergamaschi and Liu protocol~\cite{bergamaschi2024fault}. While the fault-tolerance of their state preparation gadget has only been proven for the local stochastic noise model, their protocol can be easily adapted to be fault-tolerant under the adversarial noise model. As long as we set the ``thickness'' of the gadget which they parameterized as $T$ to be the same as the distance of the code $d$, then one can show that their protocol can fail when there are at least $\Omega(d)$ faults in the circuit. Therefore, by a similar argument as in Lemma~\ref{lem:fault_tolerant_state_preparation}, their state-preparation gadget is also fault-tolerant under the adversarial noise model with $f(s) = \Theta(s)$. 

\subsection{Fault Tolerance of Single-Shot Error Correction}
\label{sec:fault_tolerance_single_shot_error_correction}
The fault tolerance of our single-shot error-correction gadgets against adversarial errors. Recall that we only perform QEC in the 2D HGP code. From Theorem~\ref{thm:linear-confinement-sipser-spielman-2D-HGP}, we know that our 2D HGP codes have good linear confinement. 
Thus, they satisfy the residual error bound stated in Lemma 1 of Ref.~\cite{quintavalle2021single}. In particular, the residual error has weight that scales linearly with the number of syndrome measurement errors. Therefore, by the same argument as in Lemma~\ref{lem:fault_tolerant_state_preparation}, our single-shot error-correction gadgets satisfy all the conditions for fault tolerance for single-shot error correction under adversarial noise as stated in Definition~\ref{def:fault_tolerant_single_shot_error_correction_gadget}.

\subsection{Fault Tolerance of Single-Shot Code-Switching}
\label{sec:fault_tolerance_code_switching}
The fault tolerance of our single-shot code-switching gadgets against adversarial errors is effectively a composition of the fault-tolerance of the single-shot state preparation gadget, transversal and homomorphic CNOT gate gadgets, and transversal measurement of code blocks. Because the transversal and homomorphic CNOT gates as well as the transversal measurements are trivially fault-tolerant according to Definitions~\ref{def:fault_tolerant_gate_gadget} and \ref{def:fault_tolerant_logical_measurement_gadget} under the adversarial noise model, all of the gadgets that we use in our code-switching protocol are fault-tolerant according to the appropriate definitions. Because all the gadgets are fault-tolerant with growing function $f$ that are at most linear in the number of faults, the entire code-switching protocol is fault-tolerant according to Definition~\ref{def:fault_tolerant_gate_gadget} with a growing function $f$ that is also linear in the number of faults.

\subsection{Fault-Tolerance of Logical Gate Gadgets}
\label{sec:fault_tolerance_logical_gate_gadgets}
We begin by analyzing the fault-tolerance of our logical Clifford gate gadgets against adversarial errors. The (fold)-transversal gates are trivially fault-tolerant according to Definition~\ref{def:fault_tolerant_gate_gadget} under the adversarial noise model, since a transversal gate will only propagate errors on each qubit to at most one other qubit. Thus, it is only essential to analyze the fault-tolerance of the GPPM gadget introduced in Ref.~\cite{xu2025fast}. The GPPM gadget is effectively a constant number of compositions of single-shot state preparation, homomorphic CNOT gates, transversal single-qubit gates, and transversal measurements. In fact, our general Clifford gate gadget discussed in Section~\ref{sec:single_shot_logical_clifford_gates} is just a composition of a constant number of the gadgets mentioned above
Because all of these gadgets are fault-tolerant with at worst linear growing functions according to the appropriate definitions under the adversarial noise model, the entire GPPM gadget is fault-tolerant according to Definition~\ref{def:fault_tolerant_gate_gadget} with a growing function $f$ that is linear in the number of faults. 

For the logical CCZ gate gadget, it is simply a single round of code-switching from 2D to 3D HGP code $\calQ_G$, followed by a constant depth circuit of CCZ gates, and then code-switching back to the 2D HGP code. Again, because all the gadgets that we use in our CCZ gadget are fault-tolerant with at worst linear growing functions according to the appropriate definitions under the adversarial noise model, the entire CCZ gadget is fault-tolerant according to Definition~\ref{def:fault_tolerant_gate_gadget} with a growing function $f$ that is linear in the number of faults.

\subsection{Fault-Tolerance of the Single-Shot, Universal Protocol}
\label{sec:dynamical_mechanism_fault_tolerance}
In order for our universal fault-tolerant scheme to be robust over many logical operations, we need to interleave single-shot error correction between all of our constant-depth gadgets. The reasoning is as follows. Suppose we have some gadget that is not a single-shot error correction gadget and it is fault-tolerant according to the definitions stated in Section~\ref{sec:fault_tolerance_gadgets} with a growing function $f$ that is polynomial in the number of faults $s$ that occur during the gadget. In addition, suppose that the code has $r$ physical faults going into the gadget. Then, the output state will have at most $r + f(s)$ faults. If we do not perform error correction after this gadget, then the next gadget will have $r + f(s)$ incoming faults instead of just $r$ incoming faults. If we keep stacking gadgets without performing error correction, then the number of incoming faults will keep increasing and eventually exceed the error-correcting radius of the code, and so the entire protocol will not be fault-tolerant. On the other hand, if we perform single-shot error correction after the gadget, then we have $r + f(s)$ faults going into the single-shot error-correction gadget. Suppose $s'$ faults occur during the single-shot error-correction gadget. By the error-recovery and correctness properties of the single-shot error-correction gadget stated in Definition~\ref{def:fault_tolerant_single_shot_error_correction_gadget}, the output state will have at most $f'(s')$ faults for some growing function $f'$ that is linear in $s'$ provided that $r + f(s) + f'(s')$ is less than or equal to the number of errors that the code can correct. The subsequent gadget will then have at most $f'(s')$ incoming faults instead of $r + f(s)$ incoming faults, and so the number of incoming faults will not keep increasing over many gadgets. Therefore, by interleaving single-shot error correction between all of our constant-depth gadgets, we can ensure that the entire protocol is fault-tolerant according to the appropriate definitions under the adversarial noise model. In particular, the combination of the residual error from the previous round of single-shot error correction with the new errors in the current round remains correctable over many error-correction cycles. In other words, the residual error remains controlled over many noisy QEC cycles.

It is perhaps not too difficult to see that the correctness of the entire protocol follows from the correctness properties of the individual gadgets by a simple composition argument. Nonetheless, this guarantee is broken when too many faults occur in any one of the gadgets. To be more concrete, when $r + f(s)$ exceeds the number of errors that the code can correct for any one of the gadgets, then the correctness of the entire protocol is not guaranteed. Thus, the number of adversarial faults that our entire protocol can tolerate is upper-bounded by the smallest number of adversarial faults that any one of the gadgets can tolerate. In particular, this is characterized by the largest $f$ of the gadgets. Notably, when all the gadgets have $f$ that is linear in the number of faults, then the entire protocol can tolerate a number of adversarial faults that is $O(d)$ where $d$ is the distance of the code. This is because each gadget now outputs a state with a number of residual faults that is linear in the number of physical faults that occur during the gadget. As long as the constants in the linear number of faults are managed appropriately, one can easily show that the single-shot error-correction gadget can always correct the combination of the residual error from the previous gadget with any new errors that occur in the error-correction gadget, provided that the total number of faults is still less than or equal to half the distance of the code. Thus, the entire protocol can tolerate a number of adversarial faults that scales $O(d) = O(\sqrt{n})$.

\section{Fault Tolerance for Local-Stochastic Noise}
\label{sec:local-stochastic}

In this section, we show that all 3D gadgets in our protocol exhibit thresholds under local-stochastic noise. Results guaranteeing fault tolerance against local-stochastic noise are typically more relevant than adversarial noise since local-stochastic noise can encompass more realistic error models for LDPC codes such as circuit-level noise. Specifically, we will show that any LDPC code equipped with an LDPC syndrome metacode and exhibiting small-set linear soundness admits a single-shot threshold for transversal logical state preparation. While our analytical results will require certain expansion properties of the underlying code or chain complex, we note that numerical evidence already suggests that generic (LDPC) 3D HGP codes, whose underlying chain complexes are not necessarily expanding, exhibit sustainable single-shot thresholds with respect to two-stage \cite{quintavalle2021single} and single-stage decoders \cite{Higgott_2023_improved}.

We first review the definition of local-stochastic noise, which encompasses both the independent error model as well as short-range correlations from constant-depth circuit-level noise.

\begin{definition}[Local-stochastic error]
\label{def:local-stochastic}
    An error pattern $e$ on a vertex set $V$ is local-stochastic with parameter $0\leq p\leq 1$ if, for all subsets $v \subseteq V$,
    \begin{align}\label{eq:LS noise}
        \mathbb{P}[v \subseteq e] \leq p^{\abs{v}} \, .
    \end{align}
\end{definition}

Note that the independent single-qubit error model with probability $p$ corresponds to an equality in \eqref{eq:LS noise}. For a local-stochastic error model on $n$ qubits, there will typically be $\Theta(n)$ errors at a given instance, which can be much larger than the minimum distance of a code; e.g. all of our codes have $d=o(n)$. One may then worry that such a large error will be uncorrectable. However, one can show that these random errors rarely conspire to form a logical operator of the code \cite{Kovalev_2013, gottesman2013fault}. The main idea is that since our code is LDPC, there are large regions of physical qubits that do not talk to each other through the parity checks. As a consequence, a typical error configuration at low error rates consists of small, disconnected clusters that are each handled independently by a decoder. To quantify how large a typical error cluster can be, also known as the cluster size distribution, we will utilize some known combinatorial results on graphs with bounded degrees.

\begin{lemma}[Cluster bound; Lemma 2 in \cite{gottesman2013fault}, Lemma 5 in \cite{Aliferis_2008}] \label{lem:cluster bound}
    Consider a vertex set $V$ of size $t$ in a simple graph with maximum degree $z$. Let $N_s(V)$ be the number of vertex sets of size $s\geq t$ whose connected components each contain a vertex in $V$. Then
    \begin{align}
        N_s(V) \leq {\rm e}^{t-1}(z{\rm e})^{s-t} \, ,
    \end{align}
    where $\mathrm{e}$ is the usual base of the natural logarithm.
\end{lemma}

\subsection{State preparation of the 3D HGP code}

In this subsection, we will prove that the logical $\ket{\overline{0}}/\ket{\overline{+}}$ state preparation of 3D HGP code $\calQ_G$ is fault-tolerant against sufficiently small local-stochastic noise. In particular, assuming that the initial syndrome error is local-stochastic with a sufficiently small but constant parameter $q<1$, we will show that the residual error after two-stage decoding is also local-stochastic with a small (but possibly worse) parameter $\gamma(q)<1$ that vanishes with $q$. As a reminder, to (transversally) prepare the logical $\ket{\overline{0}}$ state, we initialize all physical qubits in $\ket{0}$ and subsequently measure all $X$-checks.

It will be useful to define the syndrome adjacency graph $\mathcal{G}_S$ as the simple graph where parity checks are vertices, and two vertices share an edge if and only if the corresponding parity checks share a metacheck in $\mbf{M}$. If $\mbf{M}$ is $(v_{\rm c},v_{\rm r})$-LDPC, any check participates in at most $v_{\rm c}$ metachecks, and each of these neighboring metachecks involves at most $v_{\rm r}-1$ other checks. Hence, $\mathcal{G}_S$ has maximum degree $v_{\rm c}(v_{\rm r}-1)$. Likewise, define the qubit adjacency graph $\mathcal{G}_Q$ as the simple graph where qubits are vertices, and two vertices share an edge if and only if the corresponding qubits share a parity check in $\mbf{H}$. If $\mbf{H}$ is $(w_{\rm c},w_{\rm r})$-LDPC, then $\mathcal{G}_Q$ has maximum degree $w_{\rm c}(w_{\rm r}-1)$.

Let $\mbf{s}_{\rm true}$ be the true syndrome and $\mbf{s}_{\rm err}$ the syndrome error such that $\mbf{s}_{\rm obs} = \mbf{s}_{\rm true} + \mbf{s}_{\rm err}$ is the observed, corrupted syndrome. Our decoder for $\mbf{M}$, which we will call a metadecoder, will output a syndrome correction $\mbf{s}_{\rm corr}$ such that $\mbf{s}_{\rm rep} = \mbf{s}_{\rm obs} + \mbf{s}_{\rm corr}$ is the repaired or inferred syndrome. The metadecoder succeeds if $\mbf{s}_{\rm rep} \in \im{\mbf{H}}$; since $\mbf{s}_{\rm true} \in \im{\mbf{H}}$ by default, the success requirement is equivalent to $\mbf{s}_{\rm err} + \mbf{s}_{\rm corr} \equiv \mbf{s}_{\rm res} \in \im{\mbf{H}}$. We call $\mbf{s}_{\rm res}$ the \emph{residual} syndrome, since it will end up as the syndrome for the eventual residual error.

First, assuming a certain minimality condition on our metadecoder, we will adapt an argument by Bomb\'in \cite{Bombin_2015_single} to show that the residual syndrome inherits the local-stochastic property.

\begin{lemma}[Residual syndrome is local-stochastic \cite{Bombin_2015_single}]
\label{lem:LS residual syndrome}
    Suppose $\mbf{M} \in \mathbb{F}^{r \times m}_2$ is $(v_{\rm c},v_{\rm r})$-LDPC with $v_{\rm c},v_{\rm r}=O(1)$, and we have access to a minimum-weight decoder for $\mbf{M}$. Suppose the original syndrome error was local-stochastic with parameter $q<1/(2\nu)^2$ with $\nu = {\rm e} v_{\rm c}(v_{\rm r}-1)$. Then for any $s \in S$,
    \begin{align}
        \mathbb{P}[s \subseteq s_{\rm res}] \leq \beta^{\abs{s}} \quad,\quad \beta = \frac{2\nu\sqrt{q}}{1-2\nu\sqrt{q}} \, .
    \end{align}
    In other words, the residual syndrome is local-stochastic with parameter $\beta$.
\end{lemma}

\begin{proof}
    The metadecoder takes in a syndrome error $s_{\rm err}$ and outputs a residual syndrome $s_{\rm res} \in \ker{\mbf{M}}$. Since codewords of $\mbf{M}$ form connected clusters on the syndrome adjacency graph $\mathcal{G}_S$, the residual syndrome $s_{\rm res}$ will partition into connected components on $\mathcal{G}_S$. Since, within each component $c_{\rm err} \subseteq s_{\rm err}$, our metadecoder chooses a $c_{\rm corr}$ of minimal weight, we have $\abs{c_{\rm corr}} \leq \abs{c_{\rm err}}$ by noting that $c_{\rm err}$ is itself a possible solution for $c_{\rm corr}$. As a consequence, within the residual syndrome $s_{\rm res}$, at least half of its support must contain the entire syndrome error. For any connected $c \in \ker{\mbf{M}}$, the probability that $c$ is a connected component of the residual syndrome can then be bounded by
    \begin{align}\label{eq:P[s=s_res]}
        \mathbb{P}[c \subseteq s_{\rm res}] \leq \sum_{j=\lceil \abs{c}/2\rceil}^{\abs{c}} {\abs{c}\choose j} q^{\abs{j}} \leq 2^{\abs{c}} q^{\abs{c}/2} = (2\sqrt{q})^{\abs{c}} \, .
    \end{align}
    Now, the probability that $s \in S$ is a subset of the residual syndrome can be bounded by counting the number of clusters with connected components that contain $s$ in the syndrome adjacency graph $\mathcal{G}_S$ and multiplying this number by \eqref{eq:P[s=s_res]}. Let $\nu = {\rm e} v_{\rm c}(v_{\rm r}-1)$, i.e. the maximum degree of $\mathcal{G}_S$ times ${\rm e}$. The number of such clusters $N_u(s)$ of size $u\geq\abs{s}$ that contain $s$ is upper-bounded by $N_u(s) \leq {\rm e}^{s-1}\nu^{u-s} \leq \nu^{u}$ from Lemma \ref{lem:cluster bound}. Hence, we have
    \begin{align}
        \mathbb{P}[s \subseteq s_{\rm res}] \leq \sum_{\substack{c\supseteq s \\ c \in \ker{\mbf{M}}}} \mathbb{P}[c \subseteq s_{\rm res}] \leq \sum_{j=\abs{s}}^{m} \nu^{j} (2\sqrt{q})^{j} = \sum_{j=\abs{s}}^{m} (2\nu\sqrt{q})^j \, .
    \end{align}
    When $2\nu\sqrt{q}<1$, the geometric series in the last expression converges, and we arrive at
    \begin{align}
        \mathbb{P}[s\subseteq s_{\rm res}] \leq \frac{(2\nu\sqrt{q})^{\abs{s}}}{1-2\nu\sqrt{q}} \leq \left( \frac{2\nu\sqrt{q}}{1-2\nu\sqrt{q}} \right)^{\abs{s}} \equiv \beta^{\abs{s}} \, ,
    \end{align}
    where $\beta \rightarrow 0$ when $q \rightarrow 0$.
\end{proof}

We will now show that, with high probability, the residual syndrome outputted by the metadecoder is physical, i.e. $s_{\rm res} \in \im{\mbf{H}}$.

\begin{lemma}[Syndrome repair success]
\label{lem:LS metacheck success}
    Suppose $\mbf{M} \in \mathbb{F}^{r \times m}_2$ is $(v_{\rm c},v_{\rm r})$-LDPC with $v_{\rm c},v_{\rm r}=O(1)$, and we have access to a minimum-weight decoder for $\mbf{M}$. If the syndrome noise is local-stochastic with parameter $q<1/(4\nu)^2$ with $\nu = {\rm e}v_{\rm c}(v_{\rm r}-1)$ and $d_{\rm ss} = \Omega(m^b)$ for some $b>0$, then
    \begin{align}
        \mathbb{P}[s_{\rm res} \in \im{\mbf{H}}] = 1 - \mathrm{e}^{-\Omega(d_{\rm ss})} \, .
    \end{align}
\end{lemma}

\begin{proof}
    A metacheck failure occurs when the metadecoder outputs a repaired syndrome belonging to a nontrivial homology class in $\ker{\mbf{M}}\,/\,\im{\mbf{H}}$. A nontrivial operator in $\ker{\mbf{M}}\,\setminus\,\im{\mbf{H}}$ must contain a connected cluster of size at least $d_{\rm ss}$ on $\mathcal{G}_S$; otherwise, each of its connected components is individually in $\ker{\mbf{M}}$ and has size less than $d_{\rm ss}$. By the Union lemma (Lemma 2 in \cite{BPT_2010}), the entire operator then belongs to $\im{\mbf{H}}$ by definition of $d_{\rm ss}$.
    
    Our minimum-weight metadecoder can only fail when there exists a connected component of the residual syndrome on $\mathcal{G}_S$ of size $u\geq d_{\rm ss}$. Let $\nu = {\rm e} v_{\rm c}(v_{\rm r}-1)$. The number of connected clusters $N_u$ of size $u$ can be crudely upper-bounded by $N_u \leq m \nu^u$, since there are $m$ initial vertices to seed the cluster and at most $\nu^u$ clusters of size $u$ originating from the seed vertex from Lemma \ref{lem:cluster bound}. The probability of a metacheck failure can then be bounded by
    \begin{align}
        \mathbb{P}[s_{\rm res} \notin \im{\mbf{H}}] \leq \sum_{\substack{\text{connected } c \\ \abs{c}\geq d_{\rm ss}}} \mathbb{P}[c \subseteq s_{\rm res}] \leq \sum_{u=d_{\rm ss}}^{m} m\nu^{u} (2\sqrt{q})^{u} = m \sum_{u=d_{\rm ss}}^{m} (2\nu\sqrt{q})^{u} \, ,
    \end{align}
    where in the second inequality we used \eqref{eq:P[s=s_res]}. When $2\nu\sqrt{q}<1$, the geometric series in the last expression converges, and we arrive at
    \begin{align}
        \mathbb{P}[s_{\rm res} \notin \im{\mbf{H}}] \leq m \cdot \frac{(2\nu\sqrt{q})^{d_{\rm ss}}}{1-2\nu\sqrt{q}} \leq m\cdot \left( \frac{2\nu\sqrt{q}}{1-2\nu\sqrt{q}} \right)^{d_{\rm ss}} = \mathrm{e}^{-\Omega(d_{\rm ss})}
    \end{align}
    which is exponentially decaying with $d_{\rm ss}$ when $2\nu\sqrt{q}<1/2$ and $d_{\rm ss} = \Omega(m^b)$ for any $b>0$.
\end{proof}

At this point, we have shown that the entire residual syndrome is physical with high probability by showing it for each connected component. It remains to show that the residual error corresponding to this residual syndrome is also local-stochastic. To obtain the residual error from the residual syndrome, we will again employ a minimum-weight decoder.

\begin{lemma}[Residual error is local-stochastic]
\label{lem:LS residual error}
    Suppose that $H \in \F^{m\times n}_2$ is $(w_{\rm c},w_{\rm r})$-LDPC and $(t,f)$-sound with $t=\Omega(n^b)< d_{\rm ss}$ and $f(x)=\alpha x$ for $b,\alpha>0$. Let $\omega = {\rm e}w_{\rm c}(w_{\rm r}-1)$. Then for $\omega\beta<1$, with probability $1-\mathrm{e}^{-\Omega(t)}$, for any error $e \in Q$,
    \begin{align}
        \mathbb{P}[e \subseteq e_{\rm res}] \leq \gamma^{\norm{e}} \quad,\quad \gamma = \left(\frac{\omega\beta}{1-\omega\beta}\right)^{1/\alpha} \, ,
    \end{align}
    where $\gamma\rightarrow 0$ when $\beta\rightarrow 0$.
\end{lemma}

\begin{proof}
    Lemma \ref{lem:LS metacheck success} tells us that when $\nu\beta<1$, then with probability $1-\mathrm{e}^{-\Omega(d_{\rm ss})}$, the residual syndrome $s_{\rm res} \in \im{H}$ is the union of connected syndrome clusters on $\mathcal{G}_S$ each of size less than $d_{\rm ss}$. Note that $t<d_{\rm ss}$ by the definitions of soundness and single-shot distance. Lemma \ref{lem:soundness implies confinement} tells us that our linear soundness implies linear confinement up to radius $t_c=t/w_{\rm c}$. For $t=\Omega(n^b)$ with $b>0$, the same argument tells us that with probability $1-\mathrm{e}^{-\Omega(t_c)}$, each residual syndrome component is associated with a residual error cluster with size less than the linear confinement radius $t_c$. Since each residual syndrome component is disjoint by definition, the entire residual error obeys linear confinement.
    The probability of an error $e\in Q$ being the residual error is bounded by
    \begin{align}
        \mathbb{P}[e=e_{\rm res}] = \mathbb{P}[s(e) = s_{\rm res}] \leq (2\sqrt{q})^{\norm{e}/\alpha} \, .
    \end{align}
    Now consider the subset event $e \subseteq e_{\rm res}$ for some $e \in Q$. Linear confinement implies that the residual syndrome must then have weight at least $\norm{e}/\alpha$. Observe that if the other errors in $e_{\rm res} - e$ are disconnected from $e$ in $\mathcal{G}_Q$, then they would not affect the fact that $s(e) \subseteq s_{\rm res}$, and we would have $\mathbb{P}[e\subseteq e_{\rm res}] = \mathbb{P}[s(e) \subseteq s_{\rm res}]$. However, there could be errors connected to $e$ in $\mathcal{G}_Q$ that could cause some unsatisfied checks in $s(e)$ to become satisfied and thus nullify the condition $s(e) \subseteq s_{\rm res}$. Nonetheless, due to linear confinement, doing so will spawn new unsatisfied checks. We can upper bound $\mathbb{P}[e \subseteq e_{\rm res}]$ by counting the number of vertex sets in $\mathcal{G}_Q$ (with maximum degree $w_{\rm c}(w_{\rm r}-1)$) whose connected components contain a vertex in $e$ and multiplying each vertex set by the probability of encountering its syndrome in the residual syndrome. Let $\omega = {\rm e}w_{\rm c}(w_{\rm r}-1)$. The number of choices of such $c \in Q$ of size $\ell$ that can cover $e$ in $\mathcal{G}_Q$ is at most $N_\ell(e) \leq {\rm e}^{\abs{e}-1} \omega^{\ell-\abs{e}} \leq \omega^\ell$ by Lemma \ref{lem:cluster bound}. The probability that $e$ belongs to the residual error is then bounded by
    \begin{align}
        \mathbb{P}[e \subseteq e_{\rm res}] \leq \sum_{c \supseteq e} \mathbb{P}[s(c) \subseteq s_{\rm res}] \leq \sum_{\ell\geq\norm{e}/\alpha} N_\ell(e)\, \beta^\ell \leq \sum_{\ell\geq\norm{e}/\alpha} (\omega\beta)^\ell \, .
    \end{align}
    When $\omega\beta<1$, the geometric series in the last expression converges, and we arrive at
    \begin{align}
        \mathbb{P}[e \subseteq e_{\rm res}] \leq \frac{(\omega\beta)^{\norm{e}/\alpha}}{1-\omega\beta} \leq \left(\frac{\omega\beta}{1-\omega\beta}\right)^{\norm{e}/\alpha} = \gamma^{\norm{e}} \, ,
    \end{align}
    where $\gamma\rightarrow 0$ when $\beta\rightarrow 0$.
\end{proof}

We note that linear soundness is a key ingredient of this lemma, in order to translate from a local-stochastic syndrome to a local-stochastic error. Specifically, we used the exponential suppression from linear confinement to beat the combinatorial growth when counting covering sets. To this end, Corollary \ref{cor:3D HGP linear soundness} tells us that our 3D HGP code $\calQ_G$ exhibits the required linear soundness for the residual error to be local-stochastic. The final ingredient is to show that this local-stochastic residual error is correctable by our minimum-weight decoder with high probability. For LDPC codes, it has been shown that $d=\Omega(\log n)$ is sufficient for a threshold against sufficiently small local-stochastic noise, using percolation arguments \cite{Kovalev_2013, gottesman2013fault}. We will use essentially the same arguments, but adapted to our setting.

\begin{theorem}[Residual error is correctable]
    Suppose $M \in \mathbb{F}^{r \times m}_2$ is $(v_{\rm c},v_{\rm r})$-LDPC with $v_{\rm c},v_{\rm r}=O(1)$ and $d_{\rm ss} = m^a$ with $a>0$. In addition, suppose that $H \in \mathbb{F}^{m \times n}_2$ is $(w_{\rm c},w_{\rm r})$-LDPC with $w_{\rm c},w_{\rm r}=O(1)$ and also $(t,f)$-sound with $t=\Omega(n^b)< d_{\rm ss}$ and $f(x)=\alpha x$ for $b,\alpha>0$. Assume we also have access to minimum-weight decoders for both $M$ and $H$. If the syndrome error is local-stochastic with parameter $q<\left[ 2\nu(1+\omega+16^\alpha \omega^{1+2\alpha}) \right]^{-2}$ for $\nu={\rm e}v_{\rm c}(v_{\rm r}-1)$ and $\omega={\rm e}w_{\rm c}(w_{\rm r}-1)$, then with probability $1-\mathrm{e}^{-\Omega(t)}$, the residual error $e_{\rm res}$ is correctable with minimum-weight decoding.
\end{theorem}

\begin{proof}
    Lemma \ref{lem:LS metacheck success} tells us that our residual syndrome $s_{\rm res} \in \im{H}$ with probability $1-\mathrm{e}^{-\Omega(t)}$ when $q<1/(4\nu)^2$. Lemma \ref{lem:LS residual error} further tells us that, with probability $1-\mathrm{e}^{-\Omega(t)}$, the residual error is local-stochastic with parameter $\gamma(q)$. Since our decoder chooses a correction $\hat{e}_{\rm res}$ of minimal weight, we have $\norm{\hat{e}_{\rm res}} \leq \norm{e_{\rm res}}$ by minimality. By an analogous calculation to \eqref{eq:P[s=s_res]}, the combined residual error and its correction $\tilde{e}_{\rm res} = e_{\rm res} + \hat{e}_{\rm res}$ is then local-stochastic with parameter $2\sqrt{\gamma}$. By definition of correction, we must have $\tilde{e}_{\rm res} \in \ker{H}$, which means that it is either trivial ($\norm{\tilde{e}_{\rm res}}=0$) or a nontrivial logical operator ($\norm{\tilde{e}_{\rm res}} \geq d$).

    Recall that the qubit adjacency graph $\mathcal{G}_Q$ has maximum degree $w_{\rm c}(w_{\rm r}-1)$. By the Union lemma \cite{BPT_2010}, a nontrivial logical operator in $\ker{H}$ must correspond to a connected cluster of size at least $d$ on $\mathcal{G}_Q$; otherwise, one of its connected components will be simultaneously nontrivial and have size less than $d$, which would contradict the definition of minimum distance. Since we have $d > t_c = \Omega(n^b)$ by the definition of confinement, a connected cluster of size $d$ would also imply a connected cluster of size $t_c$. Hence the decoding failure probability can be bounded by the probability of encountering a cluster of size $t_c$ or greater. Let $\omega={\rm e}w_{\rm c}(w_{\rm r}-1)$. The total number of clusters of size $u$ can be bounded by $N_u \leq n\omega^u$ by Lemma \ref{lem:cluster bound}. At the same time, Lemma \ref{lem:LS residual error} tells us that the probability of encountering any particular cluster decays exponentially in its size. Hence, we have
    \begin{align}
        \mathbb{P}[e_{\rm res} \text{ uncorrectable}] \leq \sum_{j=t_c}^{n} n \omega^j (2\sqrt{\gamma})^j \leq n \sum_{j\geq t_c} (2\omega\sqrt{\gamma})^j \, .
    \end{align}
    When $2\omega\sqrt{\gamma} < 1$, the geometric series in the last expression converges, and we arrive at
    \begin{align}
        \mathbb{P}[e_{\rm res}\text{ uncorrectable}] \leq n \cdot \frac{(2\omega\sqrt{\gamma})^{t_c}}{1-2\omega\sqrt{\gamma}} \leq n \cdot \left( \frac{2\omega\sqrt{\gamma}}{1-2\omega\sqrt{\gamma}} \right)^{t_c} = \mathrm{e}^{-\Omega(t_c)}
    \end{align}
    when $2\omega\sqrt{\gamma}<1/2$ and $t_c = \Omega(n^b)$ for any $b>0$. Plugging in the definitions of $\gamma$ (Lemma \ref{lem:LS residual error}) and $\beta$ (Lemma \ref{lem:LS residual syndrome}), the previous condition becomes $q<\left[ 2\nu(1+\omega+16^\alpha \omega^{1+2\alpha}) \right]^{-2}$.
\end{proof}

As a final remark, we note that the above proof can straightforwardly be adapted to the case where one wishes to add additional (independent) local-stochastic physical noise on top of the residual error: since both distributions are local-stochastic, one can substitute in the combined distribution, which would also be local-stochastic (with a slightly worse parameter).

\subsection{Fault tolerance of the entire protocol}

We note that the composability of our gadgets with local-stochastic noise follows a similar analysis to that with adversarial noise; we briefly mention the modifications here. In particular, the residual errors for the state preparation gadgets and single-shot decoder is local-stochastic, which means that each transversal gadget in Figure \ref{fig:dimensional_expansion_circuit_2} will only add additional local-stochastic noise to the already existing local-stochastic noise for the next gadget. Assuming that the local-stochastic noise in the separate gadgets are independent, then their composition is also local-stochastic. As long as the composed local-stochastic noise is small enough, then it will be corrected by the single-shot decoder.

\bibliographystyle{alpha}
\bibliography{bib/ref}

\end{document}